%% file: main.tex
\documentclass[manuscript,nonacm]{acmart}

\usepackage[caption=false,font=footnotesize,labelfont=sf,textfont=sf]{subfig}

\usepackage{pifont}
\usepackage[ruled, vlined, noend, linesnumbered]{algorithm2e}
\usepackage{xcolor}
\usepackage{url}

\newtheorem{lemma}{Lemma}
\newtheorem{theorem}{Theorem}
\newtheorem{proposition}{Proposition}
\newtheorem{definition}{Definition}

\input{defines}

\AtBeginDocument{%
  \providecommand\BibTeX{{%
    \normalfont B\kern-0.5em{\scshape i\kern-0.25em b}\kern-0.8em\TeX}}}

\begin{document}
\title{Fast Parallel Algorithms for Enumeration of Simple, Temporal, and Hop-Constrained Cycles}

\author{Jovan Blanu\v{s}a}
\affiliation{
  \institution{IBM Research Europe - Zurich}
  \city{Zurich}
  \country{Switzerland}
}
\email{jov@zurich.ibm.com}
\additionalaffiliation{
  \institution{Ecole Polytechnique Fédérale de Lausanne}
  \department{School of Computer and Communication Sciences}
  \city{CH-1015 Lausanne}
  \country{Switzerland}
}

\author{Kubilay Atasu}
\affiliation{
  \institution{IBM Research Europe - Zurich}
  \city{Zurich}
  \country{Switzerland}
}
\email{kat@zurich.ibm.com}

\author{Paolo Ienne}
\affiliation{
  \institution{Ecole Polytechnique Fédérale de Lausanne}
  \department{School of Computer and Communication Sciences}
  \city{CH-1015 Lausanne}
  \country{Switzerland}
}
\email{paolo.ienne@epfl.ch}



\input{sections/abstract}


\maketitle

\input{sections/introduction}

\input{sections/relatedWork}

\input{sections/background}

\input{sections/coarseGrainedParallel}
   
\input{sections/fineGrainedJ}

\input{sections/fineGrainedRT}

\input{sections/timeLenConstraints}

\input{sections/experiments}

\input{sections/conclusions}

\begin{acks}
The support of Swiss National Science Foundation (project number 172610) for this work is gratefully acknowledged.
\end{acks}

\bibliographystyle{ACM-Reference-Format}
\bibliography{References}

\end{document}

%% file: defines.tex
\newcommand{\figname}{Fig.}
\newcommand\fig[1]{\figname~\ref{#1}}

\SetCommentSty{mycommfont}

\newcommand{\cmark}{\ding{51}}
\newcommand{\slfrac}[2]{\left.#1\middle/#2\right.}

\makeatletter
\patchcmd{\@algocf@start}% <cmd>
  {-1.5em}% <search>
  {0pt}% <replace>
  {}{}% <success><failure>
\makeatother

%% file: sections/abstract.tex
\begin{abstract}
Cycles are one of the fundamental subgraph patterns and being able to enumerate them in graphs enables important applications in a wide variety of fields, including finance, biology, chemistry, and network science.
However, to enable cycle enumeration in real-world applications, efficient parallel algorithms are required.
In this work, we propose scalable parallelisation of state-of-the-art sequential algorithms for enumerating simple, temporal, and hop-constrained cycles.
First, we focus on the simple cycle enumeration problem and parallelise the algorithms by Johnson and by Read and Tarjan in a fine-grained manner.
We theoretically show that our resulting fine-grained parallel algorithms are scalable, with the fine-grained parallel Read-Tarjan algorithm being strongly scalable.
In contrast, we show that straightforward coarse-grained parallel versions of these simple cycle enumeration algorithms that exploit edge- or vertex-level parallelism are not scalable.
Next, we adapt our fine-grained approach to enable the enumeration of cycles under time-window, temporal, and hop constraints.
Our evaluation on a cluster with $256$ CPU cores that can execute up to $1024$ simultaneous threads demonstrates a near-linear scalability of our fine-grained parallel algorithms when enumerating cycles under the aforementioned constraints.
On the same cluster, our fine-grained parallel algorithms achieve, on average, one order of magnitude speedup compared to the respective coarse-grained parallel versions of the state-of-the-art algorithms for cycle enumeration.
The performance gap between the fine-grained and the coarse-grained parallel algorithms increases as we use more CPU cores.
\end{abstract}

%% file: sections/introduction.tex
\section{Introduction}

Graphs are widely adopted for usage as a data representation tool across many domains~\cite{neo4j_whitepaper_finance, wang_recent_2020, noel_cygraph_2016, neo4j_whitepaper_reccomendation}.
A method of analysing graph-based data is to enumerate subgraph patterns, such as cycles, cliques, and motifs, in graphs~\cite{aggarwal_managing_2010}.
However, enumerating subgraph patterns often leads to long execution times because the number of subgraph patterns that can exist in a graph~\cite{aldred_maximum_2008, das_shared-memory_2020} can be orders of magnitude greater than the number of graph vertices.
As a result, fast subgraph enumeration algorithms are required that can exploit the parallelisim in modern multi-core processors.
In this paper, we focus on enumerating simple cycles in directed graphs and introduce scalable parallel algorithms for that problem.
A simple cycle is a sequence of edges that starts and ends with the same vertex and visits other vertices at most once.
Enumerating simple cycles has important applications in several domains.
For example, in electronic design automation, combinatorial loops in circuits are typically forbidden~\cite{gupta_acyclic_2005, pothukuchi_dhuria_2021}, and such loops can be detected by enumerating simple cycles.
In a software bug tracking system, a dependency between two software bugs requires one bug to be addressed before the other~\cite{sas_example}. Circular bug dependencies are undesirable and can be detected by finding simple cycles.
Other applications include detecting feedback loops in biological networks~\cite{kwon_analysis_2007, klamt_computing_2009} and detecting unstable relationships in social networks~\cite{giscard_evaluating_2017,zhou_cycle_2018}.

Various types of constraints are often imposed on the simple cycles because the search for simple cycles may otherwise be computationally impossible~\cite{peng_towards_2019, qiu_real-time_2018, kumar_2scent_2018}.
For instance, temporal ordering constraints can be imposed when searching for simple cycles in temporal graphs that have edges annotated with timestamps.
Simple cycles enumerated under this constraint contain edges ordered in time; such cycles are referred to as temporal cycles~\cite{kumar_2scent_2018}.
Enumerating temporal cycles has applications in the financial domain, where a temporal cycle can be an indicator of money laundering~\cite{hajdu_temporal_2020, li_flowscope_2020, AMLSim}, credit card fraud~\cite{qiu_real-time_2018}, or circular trading used for manipulating stock prices~\cite{palshikar_collusion_2008, islam_approach_2009, jiang_trading_2013}.
Other types of constraints include hop constraints~\cite{peng_towards_2019, qiu_real-time_2018}, which limit the length of paths explored during the search for cycles, and time-window constraints~\cite{kumar_2scent_2018}, which restrict the search to cycles that occur within a time window of a given size.
Hop-constrained cycles can be used to detect fraudulent behaviour in e-commerce networks~\cite{qiu_real-time_2018}.
In addition, long cycles in financial transaction networks are less likely to be associated with money laundering because they increase the risk for fraudsters of being caught~\cite{li_flowscope_2020}, and imposing hop-constraints can filter out such cycles.
Furthermore, searching for cycles under temporal ordering, hop, and time-window constraints reduces the number of paths explored during the search, making the cycle enumeration problem more tractable.
Therefore, we focus on searching for cycles under these constraints.

\input{figures/motivFig}

\textbf{Parallelisation challenges.} We focus on parallelising the algorithms by Johnson~\cite{johnson_finding_1975} and by Read and Tarjan~\cite{read_bounds_1975} for enumerating simple cycles because these algorithms achieve the lowest time complexity bounds reported for directed graphs~\cite{mateti_algorithms_1976, kao_enumeration_2016}.
Both algorithms are recursively formulated and construct a recursion tree in a depth-first fashion.
However, these algorithms employ different pruning techniques to limit the amount of work they perform.
In practice, the Johnson algorithm is faster than the Read-Tarjan algorithm due to more aggressive pruning techniques~\cite{kao_enumeration_2016,mateti_algorithms_1976}.
Furthermore, the state-of-the-art algorithms for temporal and hop-constrained cycle enumeration are extensions of the Johnson algorithm~\cite{kumar_2scent_2018, peng_towards_2019}.
Thus, parallelising the Johnson algorithm also enables parallelisation of these temporal and hop-constrained cycle enumeration algorithms.

The na\"ive way of parallelising the Johnson and the Read-Tarjan algorithms involves searching for cycles starting from different vertices or edges in parallel, which we refer to as the \textit{coarse-grained parallel} methods.
Such coarse-grained parallel approaches are straightforward to implement using the popular vertex-centric~\cite{malewicz_pregel_2010, mccune_thinking_2015} and edge-centric~\cite{roy_x-stream_2013} graph processing frameworks.
However, real-world graphs often exhibit a power-law or a log-normal distribution of vertex degrees~\cite{barabasi_network_2016, broido_scale-free_2019}.
In such graphs, the execution time of coarse-grained parallel approaches is dominated by searches that start from a small set of vertices or edges as illustrated in \fig{fig:motivfig}a.
This behaviour leads to a workload imbalance and limits scalability of parallel implementations.

The shortcomings of coarse-grained parallel approaches can be addressed by decomposing the search for cycles starting from a given edge or vertex into finer-grained tasks~\cite{blanusa_manycore_2020, das_shared-memory_2020, abdelhamid_scalemine_2016}.
However, parallelising the Johnson algorithm using the fine-grained approach is challenging because the pruning efficiency of this algorithm depends on a strictly sequential depth-first-search-based recursion tree exploration. 
We demonstrate that the lesser-known Read-Tarjan algorithm does not have such a requirement, and, thus, it is easier to decompose into fine-grained tasks.

\input{tables/theorySummary}

\textbf{Contributions.}
This paper presents an extension of the work by Blanu\v{s}a et al.~\cite{blanusa_scalable_2022}, which introduces the following contributions:

\emph{(i) Scalable fine-grained parallelisation of the Johnson and the Read-Tarjan algorithms.} 
To our knowledge, we are the first ones to parallelise these asymptotically-optimal cycle enumeration algorithms in a fine-grained manner and achieve an almost linear performance scaling on a system that can execute up to a thousand concurrent software threads.
Such a scalability is enabled by our decomposition of long sequential searches into fine-grained tasks, which are then dynamically scheduled across CPU cores.
To decompose the Johnson algorithm into fine-grained tasks, we have relaxed its strictly depth-first-search-based exploration, which enables this algorithm to perform multiple independent depth-first searches in parallel.
As a result, our fine-grained parallel Johnson algorithm is able to achieve an ideal load balancing as shown in \fig{fig:motivfig}b.

\emph{(ii) Theoretical analysis of the coarse- and fine-grained parallel algorithm.}
We theoretically show that both of our fine-grained parallel algorithms are scalable, which is not the case for the Johnson and the Read-Tarjan algorithms parallelised in a coarse-grained manner.
Moreover, we show that our fine-grained parallel Read-Tarjan algorithm performs asymptotically the same amount of work as its serial version, whereas our fine-grained parallel Johnson algorithm does not.
Therefore, our fine-grained parallel Read-Tarjan algorithm is the only parallel algorithm based on an asymptotically-optimal cycle enumeration algorithm that is both work-efficient and scalable, as shown in Table~\ref{tab:theoSummary}.
Interestingly, despite not being work-efficient, our fine-grained Johnson algorithm outperforms our fine-grained parallel Read-Tarjan algorithm in most of our experiments.

In this paper, we extend our prior work~\cite{blanusa_scalable_2022} with the following contributions:

\emph{(iii) General framework for parallelising temporal and hop-constrained cycle enumeration.}
We show that our method for parallelising the Johnson algorithm in a fine-grained manner can be adapted to parallelise the state-of-the-art algorithms for temporal and hop-constrained cycle enumeration.
This adaptation is possible because these state-of-the-art algorithms, such as the 2SCENT algorithm for temporal cycle enumeration~\cite{kumar_2scent_2018} and the BC-DFS algorithm~\cite{peng_towards_2019} for hop-constrained cycle enumeration, are extensions of the Johnson algorithm.
By parallelising these algorithms using our fine-grained method, we were able to achieve speedups of up to $40\times$ and $61\times$ compared to the coarse-grained parallel versions of 2SCENT and BC-DFS, respectively.

\emph{(iv) Improvements to the pruning efficiency of the Read-Tarjan algorithm.}
To make this algorithm competitive with the Johnson algorithm, we have introduced several optimisations that enhance the pruning efficiency of the Read-Tarjan algorithm.
The optimisations reduce the amount of unnecessary vertex visits that this algorithm performs.
As a result, our improved version of the Read-Tarjan algorithm is up to $6.8\times$ faster than the original version of this algorithm.

\textbf{Paper structure.} The remainder of this paper is organised as follows.
The related work and background are presented in Section~\ref{sect:related_work} and Section~\ref{section:background}, respectively.
Coarse-grained parallel versions of the Johnson and the Read-Tarjan algorithms are covered in Section~\ref{sect:vertEdgePar}.
Section~\ref{sect:tpJohnson} and Section~\ref{sect:tpReadTarjan} introduce our fine-grained parallel versions of the Johnson and the Read-Tarjan algorithms, respectively.
Section~\ref{sect:tpReadTarjan} also includes our optimisations for improving the pruning efficiency of the Read-Tarjan algorithm.
Our general framework for parallelising temporal and hop-constrained cycle enumeration algorithms is presented in Section~\ref{sect:lcycle}.
In Section~\ref{sect:experiments}, we provide an experimental evaluation of our fine-grained parallel algorithms.
Finally, we conclude our work in Section~\ref{sect:conclusion}.

%% file: figures/motivFig.tex
\begin{figure}[t]
    \includegraphics[width=0.67\linewidth]{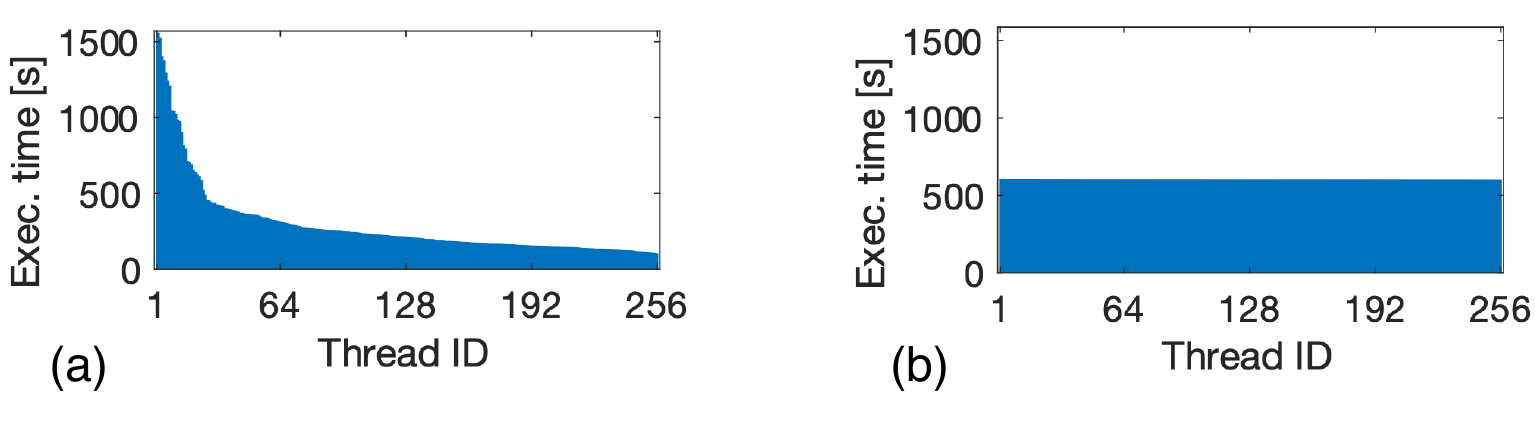}
	\vspace{-.2in}
	\caption{Per-thread execution time of (a) the coarse-grained Johnson algorithm vs. (b) our fine-grained Johnson algorithm
	using the WT graph and a $12h$ time window.
	Thanks to a perfect load balancing, our fine-grained method is $3\times$ faster on a $64$-core CPU executing $256$ threads.}
	\vspace{-.15in}
	\label{fig:motivfig}
\end{figure}

%% file: tables/theorySummary.tex
\begin{table}[t]
\centering
\caption{
Our fine-grained parallel Read-Tarjan algorithm is the only solution that is both work-efficient and scalable.}
\vspace{-.05in}
\addtolength{\tabcolsep}{-1pt}
\begin{tabular}{l|cc}
\textbf{Parallel algorithm}               &  \textbf{Work-efficient} & \textbf{Scalable}  \\ \hline
Coarse-grained parallel algorithms  &  \cmark  &        \\
Our fine-grained parallel Johnson    &  & \cmark   \\
Our fine-grained parallel Read-Tarjan  & \cmark  & \cmark     \\ \hline
\end{tabular}
\label{tab:theoSummary}
\vspace{-.2in}
\end{table}

%% file: sections/relatedWork.tex
\section{Related work}
\label{sect:related_work}
\textbf{Simple cycle enumeration algorithms.}
Enumeration of simple cycles of graphs is a classical computer science problem~\cite{tiernan_efficient_1970, tarjan_enumeration_1973, johnson_finding_1975, read_bounds_1975, mateti_algorithms_1976, Szwarcfiter1976ASS, kao_enumeration_2016, weinblatt_new_1972, loizou_enumerating_1982, birmele_optimal_2013, agarwal_finding_2016}. 
The backtracking-based algorithms by Johnson~\cite{johnson_finding_1975}, Read and Tarjan~\cite{read_bounds_1975}, and Szwarcfiter and Lauer~\cite{Szwarcfiter1976ASS} achieve the lowest time complexity bounds for enumerating simple cycles in directed graphs.
These algorithms implement advanced recursion tree pruning techniques to improve on the brute-force Tiernan algorithm~\cite{tiernan_efficient_1970}. Section~\ref{section:back_johnson} covers such pruning techniques in further detail.
A cycle enumeration algorithm that is asymptotically faster than the aforementioned algorithms~\cite{johnson_finding_1975, read_bounds_1975, Szwarcfiter1976ASS} has been proposed in Birmelé et al.~\cite{birmele_optimal_2013}, however, it is applicable only to undirected graphs.
Simple cycles can also be enumerated by computing the powers of the adjacency matrix~\cite{danielson_finding_1968, kamae_systematic_1967, ponstein_self-avoiding_1966} or by using circuit vector space algorithms~\cite{mateti_algorithms_1976, gibbs_cycle_1969, welch_numerical_1965}, but the complexity of such approaches grows exponentially with the size of the cycles or the size of the input graphs.

\input{tables/relatedWorkTable}

\textbf{Time-window, temporal ordering, and hop constraints.}
It is common to search for cycles under some additional constraints. 
For instance, in temporal graphs, it is common to search for cycles within a sliding time window, such as in Kumar and Calders~\cite{kumar_2scent_2018} and Qiu et al~\cite{qiu_real-time_2018}.
In addition, temporal ordering constraints can be imposed when searching for cycles in temporal graphs, such as in Kumar and Calders~\cite{kumar_2scent_2018}.
Furthermore, the maximum number of hops in cycles or paths can be constrained, such as in Gupta and Suzumura~\cite{gupta_finding_2021} and Peng et al.~\cite{peng_towards_2019}.
Note that hop-constrained simple cycles can also be enumerated using incremental algorithms, such as in Qiu et al.~\cite{qiu_real-time_2018}.
However, this algorithm is based on the brute-force Tiernan algorithm~\cite{tiernan_efficient_1970}, which makes it slower than nonincremental algorithms that use recursion tree pruning techniques~\cite{peng_towards_2019}.
Additionally, because incremental algorithms maintain auxiliary data structures, such as paths, to be able to construct cycles incrementally, they are not as memory-efficient as nonincremental algorithms~\cite{peng_towards_2019}.
Table~\ref{tab:relWork} offers comparisons between the capabilities of these methods and~ours. 

\textbf{Parallel and distributed algorithms for cycle enumeration.}
Cui et al. \cite{cui_multi-threading_2017} proposed a multi-threaded algorithm for detecting and removing simple cycles of a directed graph.
The algorithm divides the graph into its strongly-connected components and each thread performs a depth-first search on a different component to find cycles.
However, sizes of the strongly-connected components in real-world graphs can vary significantly~\cite{meusel_graph_2014}, which leads to a workload imbalance.
Rocha and Thatte~\cite{rocha_distributed_2015} proposed a distributed algorithm for simple cycle enumeration based on the bulk-synchronous parallel model~\cite{valiant_bridging_1990}, but it searches for cycles in a brute-force manner.
Qing et al.~\cite{nah_efficient_2020} introduced a parallel algorithm for finding length-constrained simple cycles. %that start from a given vertex.
It is the only other fine-grained parallel algorithm we are aware of in the sense that it can search for cycles starting from the same vertex in parallel.
However, the way this algorithm searches for cycles is similar to the way the brute-force Tiernan algorithm~\cite{tiernan_efficient_1970} works.
To our knowledge, we are the first ones to introduce fine-grained parallel versions of asymptotically-optimal simple cycle enumeration algorithms, which do not rely on a brute-force search, as we show in Table~\ref{tab:relWork}.
Distributed algorithms for detecting the presence of cycles in graphs readily exist~\cite{Bader99apractical, oliva_distributed_2018, fraigniaud_distributed_2019}.
However, our focus is on discovering all simple cycles of a graph rather than detecting whether a graph has a cycle or not.

%% file: tables/relatedWorkTable.tex
\begin{table}[t]
\centering
\caption{Capabilities of the related work versus our own. Competing algorithms either fail to exploit fine-grained parallelism or do it on top of asymptotically inferior \mbox{algorithms}.}
\addtolength{\tabcolsep}{-1pt}
\begin{tabular}{l|cccccc}
\textbf{Related work} & \textbf{\cite{kumar_2scent_2018}} & \textbf{\cite{qiu_real-time_2018}} & \textbf{\cite{peng_towards_2019}} &  \textbf{\cite{nah_efficient_2020}} & \textbf{\cite{gupta_finding_2021}}  & \textbf{Ours} \\ 
\hline 
Fine-grained parallelism    &       &       &       & \cmark &       & \cmark \\
Asymptotic optimality       & \cmark &       & \cmark &       & \cmark & \cmark \\ 
Temporal cycles             & \cmark &       &       &       &       & \cmark \\
Time-window constraints     & \cmark & \cmark &       &       &       & \cmark \\
Hop constraints    &       & \cmark & \cmark & \cmark & \cmark & \cmark \\
\hline
\end{tabular}
\label{tab:relWork}
\vspace{-.15in}
\end{table}

%% file: sections/background.tex
\section{Background}
\label{section:background}

This section introduces the main theoretical concepts used in this paper and provides an overview of the most prominent simple cycle enumeration algorithms.
The notation used is given in Table~\ref{tab:notation}.

\input{sections/back_preliminaries}

\subsection{Simple cycle enumeration algorithms}

The following algorithms for simple cycle enumeration perform recursive searches to incrementally update simple paths that can lead to cycles.
Each algorithm iterates the vertices or edges of the graph and independently constructs a recursion tree to enumerate all the cycles starting from that vertex or edge.
The difference between these algorithms is to what extent they reduce the redundant work performed during the recursive search, which we discuss next.

\input{sections/back_Tiernan}

\input{sections/back_Johnson}

\input{figures/backgroundGraph}

\input{sections/back_readTarjan}

%% file: sections/back_preliminaries.tex
\subsection{Preliminaries}
\label{sect:back_prelim}

\input{tables/notation}

We consider a directed graph $\mathcal{G}(\mathcal{V}, \mathcal{E})$ having a set of vertices $\mathcal{V}$ and a set of directed edges $\mathcal{E} = \{ u \rightarrow v \mid u, v \in \mathcal{V}\}$.
The set of neighbours of a given vertex $v$ is defined as $\mathcal{N}(v) = \{ w \mid \; v \rightarrow w \in \mathcal{E}\}$.
We refer to the vertex $v$ of an edge $v \rightarrow u$ as its source vertex and to the vertex $u$ as its destination vertex.
An outgoing edge of a given vertex $v$ is defined as $v \rightarrow w$ and an incoming edge is defined as $u \rightarrow v$, where $v \rightarrow w, u \rightarrow v  \in \mathcal{E}$.
A \emph{path} between the vertices $v_0$ and $v_k$, denoted as $v_0 \rightarrow v_1 \ldots \rightarrow v_k$, is a sequence of vertices such that there exists an edge between every two consecutive vertices of the sequence.
A \emph{simple path} is a path with no repeated vertices.
A simple path is \emph{maximal} if the last vertex of the path has no neighbours or all of its neighbours are already in the path~\cite{erdos_maximal_1959}.
A \textbf{cycle} is a path of non-zero length from a vertex $v$ to the same vertex $v$.
A \textbf{simple cycle} is a cycle with no repeated vertices except for the first and last vertex.
The number of maximal simple paths and the number of simple cycles in a graph are denoted as $s$ and $c$, respectively (see Table~\ref{tab:notation}).
Note that $s$ can be exponentially larger than $c$~\cite{tarjan_enumeration_1973}.
A path or a cycle is said to satisfy a \textbf{hop-constraint} $L$ if the number of edges in that path or cycle is less than or equal to $L$.
The goal of \textbf{simple cycle enumeration} is to compute all simple cycles of a directed graph $\mathcal{G}$, ideally without computing all maximal simple paths of it.

A \textbf{temporal graph} is a graph that has its edges annotated with timestamps.~\cite{paranjape_motifs_2017}.
Such a graph might contain \textit{parallel edges}, which are edges with the same source and destination vertices~\cite{Balakrishnan1997}.
An example of a temporal graph with parallel edges is given in \fig{fig:time-window}.
In temporal graphs, a \textbf{temporal cycle} is a simple cycle, in which the edges appear in the increasing order of their timestamps.
A simple cycle or a temporal cycle of a temporal graph occurs within a \textbf{time window} $\left[t_{w1}: t_{w2} \right]$ if every edge of that cycle has a timestamp $t_s$ such that $t_{w1} \leq t_s \leq t_{w2}$.
\fig{fig:time-window} shows the simple cycles of a temporal graph that occur within two different time windows of size $\delta = 5$.
This graph contains two simple cycles in the time window $\left[2: 7\right]$ (\fig{fig:time-window}a), which are also temporal cycles, and two simple cycles in the time window $\left[10: 15\right]$ (\fig{fig:time-window}b), neither being a temporal cycle.
Note that the existence of parallel edges in temporal graphs makes it possible to have several simple cycles that contain the same sequence of vertices, as shown in \fig{fig:time-window}a.
The union of several cycles that contain the same sequence of vertices is called a \textit{cycle bundle}~\cite{kumar_2scent_2018}.

\input{figures/time-window}

\subsection{Task-level parallelism}
\label{sect:tlp}
The parallel algorithms described in this paper can be implemented using shared-memory parallel processing frameworks, such as TBB~\cite{kukanov_foundations_2007}, Cilk~\cite{blumofe_cilk_1996}, and OpenMP~\cite{quinn_parallel_2004}.
These frameworks enable the decomposition of a program into tasks that can be independently executed by different software threads.
In our setup, tasks are created and scheduled dynamically. A \emph{parent} task can \emph{spawn} several \emph{child} tasks.
The \textit{depth} of a task is the number of its direct ancestors.
A dynamic task management system assigns the tasks created to the work queues of the available threads.
Furthermore, a work-stealing scheduler~\cite{blumofe_scheduling_1999, kukanov_foundations_2007, blumofe_cilk_1996} enables a thread that is not executing a task to \emph{steal} a task from the work queue of another thread.
Stealing tasks enables dynamic load balancing and ensures full utilisation of the threads when there are sufficiently many tasks.

\subsection{Work efficiency and scalability}
\label{sect:back_eff}

We use the notions of \emph{work efficiency} and \emph{scalability} to analyse parallel algorithms~\cite{par_algos}.
We refer to the time to execute a parallel algorithm on a problem of size $n$ using $p$ threads as $T_p(n)$. The size of a graph is determined by the number of vertices $n$ as well as the number of edges $e$, but we will refer only to $n$ for simplicity.
The \emph{depth} of an algorithm is the length of the longest sequence of dependent operations in the algorithm. 
The time required to execute such a sequence is equal to the execution time of the parallel algorithm using an infinite number of threads, denoted by $T_{\infty}$.
Furthermore, \emph{work} performed by a parallel algorithm on a problem of size $n$ using $p$ threads, denoted as $W_p(n)$, is the sum of the execution times of the individual threads.
The \emph{work efficiency} and the \emph{scalability} are formally defined as follows.
\begin{definition}
\label{def:workEfficiency}
(\textit{Work efficiency}) 
A parallel algorithm is work-efficient if and only if $W_p(n) \in O(T_1(n))$.
\end{definition}

\begin{definition}
\label{def:scalability}
(\textit{Scalability}) 
A parallel algorithm is scalable if and only if $\lim\limits_{n\to\infty} \left(\lim\limits_{p\to\infty} \dfrac{T_p(n)}{T_1(n)} \right) = 0$.
\end{definition}

Informally, a work-efficient parallel algorithm performs the same amount of work as its serial version, within a constant factor.
Scalability implies that, for sufficiently large inputs, increasing the number of threads increases the speedup of the parallel algorithm with respect to its serial version.

We also define the notion of \emph{strong scalability} as follows~\cite{JaJa1992-va}.

\begin{definition}
\label{def:strongScalability}
(\textit{Strong scalability}) 
A parallel algorithm is strongly scalable if and only if $\dfrac{T_1(n)}{T_p(n)} = \Theta(p)$ for large enough~$n$.
\vspace{-.03in}
\end{definition}

Whereas Definition~\ref{def:scalability} implies that the speedup $T_1(n)/T_p(n)$ achieved by a parallel algorithm with respect to its serial execution is infinite when the number of threads $p$ is infinite, Definition~\ref{def:strongScalability} implies that the speedup is always in the order of $p$.
Another related concept is weak scalability, which requires the speedup to be in the order of $p$ when the input size per thread is constant.
Note that both strong scalability and weak scalability imply scalability.

%% file: tables/notation.tex
\begin{table}[t]
\centering
\caption{Summary of the notation used in the paper.
\vspace{-.05in}
}
\begin{tabular}{l|l|l|l}
	\textbf{Symbol}            &    \textbf{Description}   & \textbf{Symbol}            &    \textbf{Description}      \\ \hline

    \textbf{$\mathcal{G}(\mathcal{V}, \mathcal{E})$} & Graph with vertices $\mathcal{V}$ and edges $\mathcal{E}$. &
    \textbf{$\Pi$}  & Current simple path explored by an algorithm. \\
    \textbf{$\mathcal{N}(v)$} & The set of neighbours of the vertex $v$. &
    \textbf{$\mathit{Blk}$} & Set of blocked vertices. \\
	\textbf{$u \rightarrow v$} & A directed edge from $u$ to $v$. &
    \textbf{$\mathit{Blist}$} & Unblock list of the Johnson algorithm. \\
    \textbf{$n$, $e$}    & No. vertices and edges in a graph.    &
    \textbf{$\mathit{E}$} & Path extension of the Read-Tarjan algorithm. \\
    \textbf{$\delta$} & Size of a time window. &
    \textbf{$\mathit{X_{T_i}}$} & Data structure $X$ is maintained by the thread $T_i$. \\
    \textbf{$L$} & Hop constraint. &
	\textbf{$p$}       &  No. threads used by a parallel algorithm. \\
	\textbf{$c$}       & No.  simple cycles in a graph. &
	\textbf{$T_{p}(n)$}       & Execution time of a parallel algorithm. \\ 
    \textbf{$s$} & No. maximal simple paths in a graph. &
	\textbf{$W_{p}(n)$}       & Amount of work a parallel algorithm performs. \\ 
	\hline
\end{tabular}
\label{tab:notation}
\vspace{-.15in}
\end{table}

%% file: figures/time-window.tex
\begin{figure}[t]
        \includegraphics[width=0.7\linewidth]{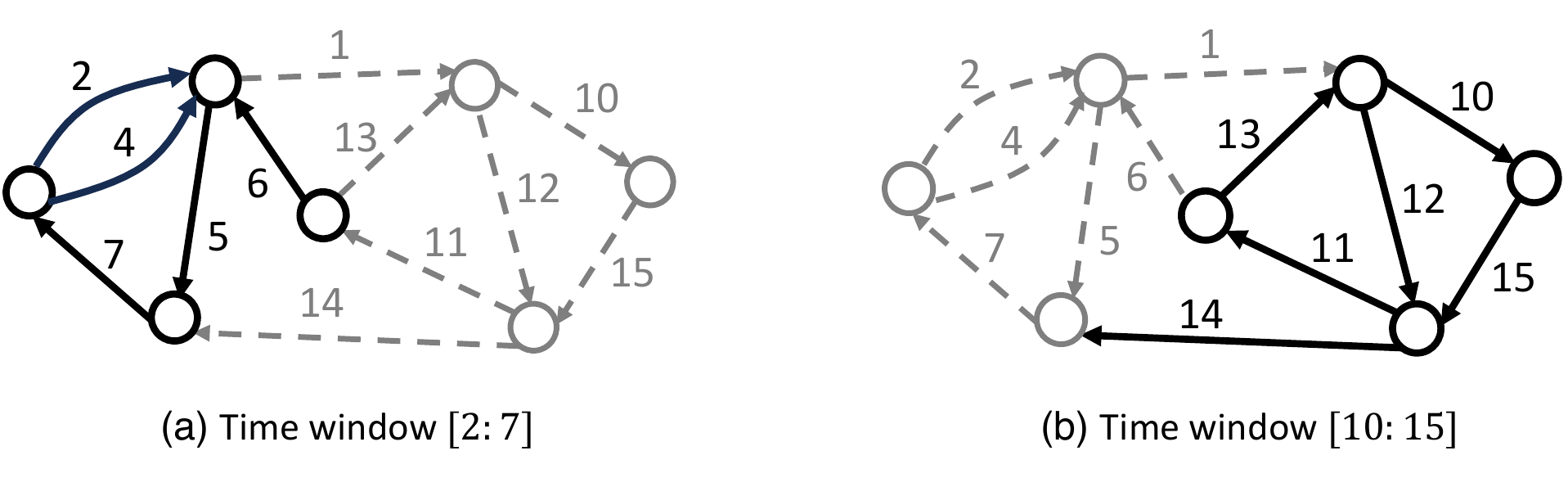}
	\vspace{-.1in}
	\caption{Two snapshots of a temporal graph associated with two different time windows of size $\delta = 5$. The solid arrows indicate the edges that belong to the respective time windows.}
	\label{fig:time-window}
	\vspace{-.1in}
\end{figure}

%% file: sections/back_Tiernan.tex
\label{section:back_tiernan}

\textbf{The Tiernan algorithm}~\cite{tiernan_efficient_1970} enumerates simple cycles using a brute-force search.
It recursively extends a simple path $\Pi$ by appending a neighbour $u$ of the last vertex $v$ of $\Pi$ provided that  $u$ is not already in $\Pi$.
A clear downside of this algorithm is that it can repeatedly visit vertices that can never lead to a cycle.
When searching for cycles in the graph shown in \fig{fig:background}a starting from the vertex $v_0$, this algorithm would explore the path containing $b_{1}, \ldots, b_{k}$ $2m$ times.
From each vertex $w_i$ and $u_i$, with $i \in \{1, \ldots, m\}$, the Tiernan algorithm would explore this path only to discover that it cannot lead to a simple cycle. 
As noted by Tarjan~\cite{tarjan_enumeration_1973}, the Tiernan algorithm explores every simple path and, consequently, all maximal simple paths of a graph.
Exploring a maximal simple path takes $O(e)$ time because it requires visiting each edge of the graph in the worst case.
Given a graph with $s$ maximal simple paths (see Table~\ref{tab:notation}), the worst-case time complexity of the Tiernan algorithm is~$O(se)$.

%% file: sections/back_Johnson.tex
\label{section:back_johnson}

\textbf{The Johnson algorithm}~\cite{johnson_finding_1975} improves upon the Tiernan algorithm by avoiding the vertices that cannot lead to simple cycles when appended to the current simple path $\Pi$.
For this purpose, the Johnson algorithm maintains a set of blocked vertices $\mathit{Blk}$ that are avoided during the search.
In addition, a list of vertices $\mathit{Blist}[w]$ is stored for each blocked vertex $w$.
Whenever a vertex $w$ is unblocked (i.e., removed from $\mathit{Blk}$) by the Johnson algorithm, the vertices in $\mathit{Blist}[w]$ are also unblocked. 
This unblocking process is performed recursively until no more vertices can be unblocked, which we refer to as the \emph{recursive unblocking} procedure.

A vertex $v$ is blocked (i.e., added to $\mathit{Blk}$) when visited by the algorithm.
If a cycle is found after recursively exploring every neighbour of $v$ that is not blocked, the vertex $v$ is unblocked.
However, $v$ is not immediately unblocked if no cycles are found after exploring its neighbours.
Instead, the $\mathit{Blist}$ data structure is updated to enable unblocking of $v$ in a later step by adding $v$ to the list $\mathit{Blist}[w]$ of every neighbour $w$ of $v$.
This delayed unblocking of the vertices enables the Johnson algorithm to discover each cycle in $O(e)$ time in the worst case.
Because this algorithm requires $O(n+e)$ time to determine that there are no cycles, its worst-case time complexity is $O\left(n+e +ec\right)$~\cite{Szwarcfiter1976ASS}.
Note that because $s$ can be exponentially larger than $c$~\cite{tarjan_enumeration_1973}, the Johnson algorithm is asymptotically faster than the Tiernan algorithm.

In the example shown in \fig{fig:background}a, every simple path $\Pi$ that starts from $v_0$ and contains vertices $b_1,\ldots,b_k$ is a maximal simple path, and, thus, it cannot lead to a simple cycle. 
The Johnson algorithm would block $b_1,\ldots,b_k$ immediately after visiting this sequence once and then keep these vertices blocked until it finishes exploring the neighbours of $v_2$.
As a result, the Johnson algorithm visits vertices $b_1,\ldots,b_k$ only once, rather than $2m$ times the Tiernan algorithm would visit them.
Note that because these vertices get blocked during the exploration of the left subtree of the recursion tree, they are not going to be visited again during the exploration of the right subtree. 
Effectively, a portion of the right subtree is pruned (see the dotted path in \fig{fig:background}b) based on the updates made on $\mathit{Blk}$ and $\mathit{Blist}$ during the exploration of the left subtree.
This strictly sequential depth-first exploration of the recursion tree is critically important for achieving a high pruning efficiency, but it also makes scalable parallelisation of the Johnson algorithm extremely challenging, which we are going to cover in Section~\ref{sect:tpJohnson}.

%% file: figures/backgroundGraph.tex
\begin{figure}[t]
	\centerline{
		\includegraphics[width=0.7\linewidth]{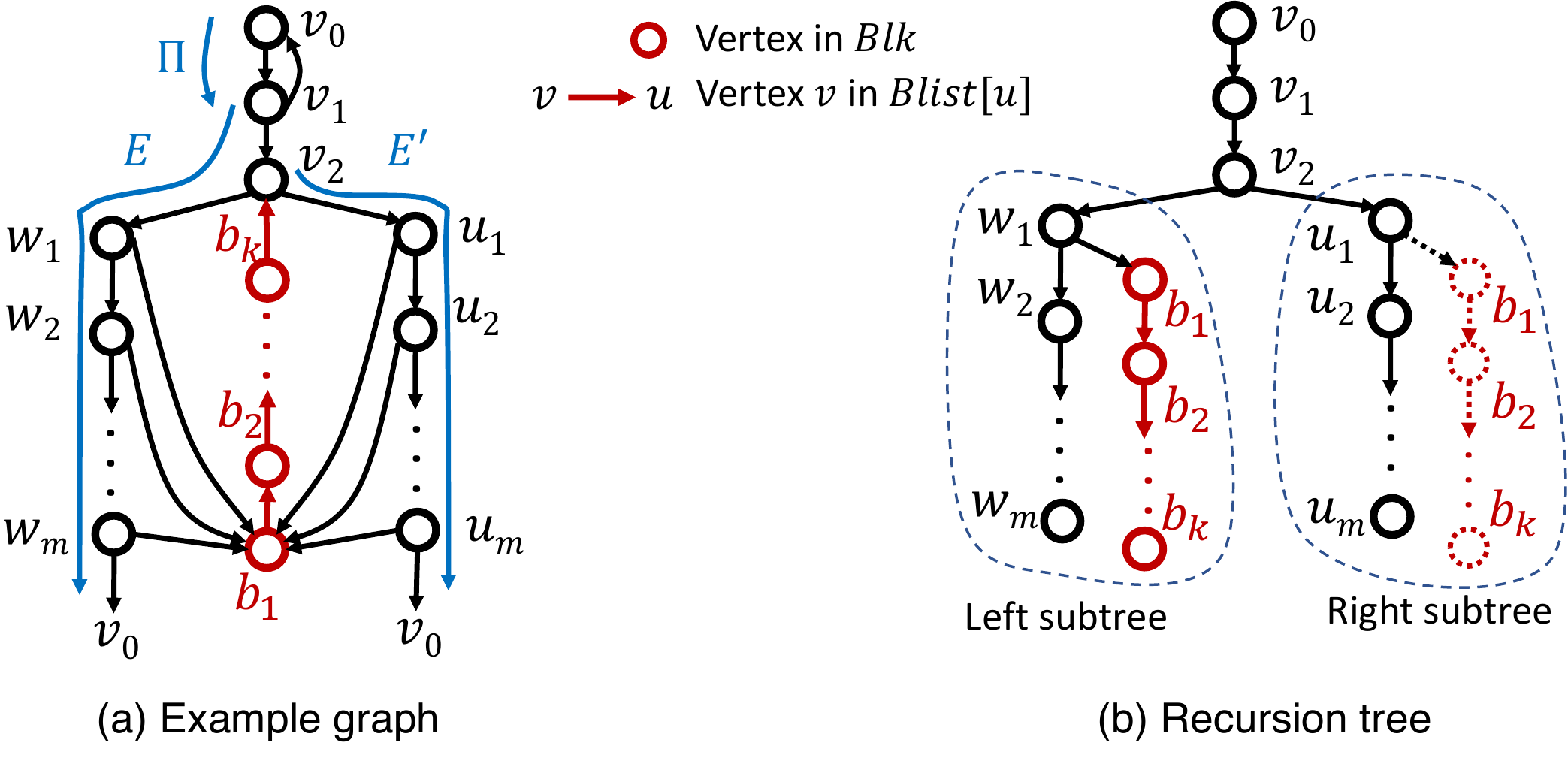}
	}
	\vspace{-.15in}
	\caption{
	(a) An example graph and (b) the recursion tree constructed when searching for cycles that start from $v_0$.
	The nodes of the recursion tree represent the recursive calls of the depth-first search.
	The dotted path of the right subtree is explored only by the Read-Tarjan algorithm.
	}
	\label{fig:background}
	\vspace{-.2in}
\end{figure}

%% file: sections/back_readTarjan.tex
\label{sect:read_tarjan}

\textbf{The Read-Tarjan algorithm}~\cite{read_bounds_1975} also has a worst-case time complexity of $O\left(n + e + ec\right)$.
This algorithm maintains a current path $\Pi$ between a starting vertex and a frontier vertex.
A recursive call of this algorithm iterates the neighbours of the current frontier vertex and performs a depth-first search (DFS). Assume that $v_0$ is the starting vertex and $v_1$ is the frontier vertex of $\Pi$ (see \fig{fig:background}a). From each neighbour $y \in \{v_0, v_2\}$ of $v_1$, a DFS tries to find a path extension $E$ back to $v_0$ that would form a simple cycle when appended to $\Pi$.
In the example shown in \fig{fig:background}a, the algorithm finds two path extensions, one indicated as $E$ and one that consists of the edge $v_1 \rightarrow v_0$.
The algorithm then explores each path extension by iteratively appending the vertices from it to the path $\Pi$.  
For each vertex $x$ added to $\Pi$, the algorithm also searches for an alternate path extension from that vertex $x$ to $v_0$ using a DFS.
In the example given in \fig{fig:background}a, the algorithm iterates through the vertices of the path extension $E$ and finds an alternate path extension $E^{\prime}$ from the neighbour $u_1$ of $v_2$.
If an alternate path extension is found, a child recursive call is invoked with the updated current path $\Pi$, which is $v_0 \rightarrow v_1 \rightarrow v_2$ in our example.
Otherwise, if all the vertices in $E$ have already been added to the current path $\Pi$, $\Pi$ is reported as a simple cycle. 
In our example, the Read-Tarjan algorithm explores both $E$ and $E^{\prime}$ path extensions, and each leads to the discovery~of~a~cycle.

The Read-Tarjan algorithm also maintains a set of blocked vertices $\mathit{Blk}$ for recursion-tree pruning. However, differently from the Johnson algorithm, $\mathit{Blk}$ only keeps track of the vertices that cannot lead to new cycles when exploring the current path extension within the same recursive call.
The vertices in $\mathit{Blk}$ are avoided while searching for additional path extensions that branch from the current path extension.
For instance, the left subtree of the recursion tree shown in \fig{fig:background}b demonstrates the exploration of the path extension $E$ shown in \fig{fig:background}a.
During the exploration of $E$, the vertices $b_1, \ldots, b_k$ are added to $\mathit{Blk}$ immediately after visiting $w_1$, and they are not visited again while exploring $E$.
However, when exploring another path extension $E^{\prime}$ in the right subtree, the vertices $b_1, \ldots, b_k$ are visited once again (see the dotted path of the right subtree).
As a result, the Read-Tarjan algorithm visits $b_1, \ldots, b_k$ twice instead of just once.
As we are going to show in Section~\ref{sect:tpReadTarjan}, this drawback becomes an advantage when parallelising the Read-Tarjan algorithm because it enables independent exploration of different subtrees of the recursion tree.

%% file: sections/coarseGrainedParallel.tex
\section{Coarse-grained parallel methods}
\label{sect:vertEdgePar}

The most straightforward way of parallelising the Johnson and the Read-Tarjan algorithms is to search for cycles that start from different vertices in parallel.
Each such search can then be executed by a different thread that explores its own recursion tree.
This approach is beneficial because it is work-efficient and can be implemented using one of the existing graph processing frameworks, such as Pregel~\cite{malewicz_pregel_2010}, in a manner similar to the method used by Rocha and Thatte~\cite{rocha_distributed_2015}.
We refer to this parallelisation approach as the coarse-grained parallel approach.

The coarse-grained approach can express more parallelism if each thread performs a search for cycles that start from a different edge rather than a different vertex.
This assumption is supported by the fact that graphs typically have more edges than vertices.
Nevertheless, the coarse-grained approach is not scalable, which we prove here.

\begin{proposition}
\label{theorem:cg_workEff}
The coarse-grained parallel Johnson and Read-Tarjan algorithms are work-efficient.
\end{proposition}

The proof of Proposition~\ref{theorem:cg_workEff} is trivial, and we omit it for brevity.

\input{figures/wcExample}

\input{tables/workDepth}

\begin{theorem}
\label{theorem:cg_scalability}
The coarse-grained parallel Johnson and Read-Tarjan algorithms are not scalable.
\end{theorem}

\begin{proof}
In this case, the depth $T_{\infty}(n)$ represents the worst-case execution time of a search for cycles that starts from a single vertex or edge, and it depends on the number of cycles found during this search.
In the worst case, a single recursive search can discover all cycles of a graph.
An example of such graph is given in \fig{fig:wc_example}a, where each vertex $v_i$, with $i\in\{1,\ldots,n-1\}$, is connected to $v_0$ and to every vertex $v_j$ such that $j > i$.
In that graph, any subset of vertices $v_2,\ldots, v_{n-1}$ defines a different cycle.
Therefore, the total number of cycles in this graph is equal to the number of all such subsets $c = 2^{n-2}$.
Before the search for cycles, both the Johnson and the Read-Tarjan algorithm find all vertices that start a cycle, which is only $v_0$ in this case.
Therefore, the search for cycles will be performed only by one thread.
Because both the Johnson and the Read-Tarjan algorithms require $O(e)$ time to find each cycle, the depth of the coarse-grained algorithms is $T_{\infty}(n) \in O(e c)$.
Because $\lim\limits_{n\to\infty} T_{\infty}(n)/T_1(n) \neq 0$, the coarse-grained algorithms are not scalable based on Definition~\ref{def:scalability}.
\vspace{-.07in}
\end{proof}

Theorem~\ref{theorem:cg_scalability} shows that the main drawback of the coarse-grained parallel algorithms is their limited scalability.
This limitation is apparent for the graph shown in \fig{fig:wc_example}a, which has an exponential number of cycles in $n$. 
When using a coarse-grained parallel algorithm on this graph, all the cycles will be discovered by a single thread, and, thus, the depth of this algorithm grows linearly with $c$, as shown in Table~\ref{tab:workDepth}.
Because only one thread can be effectively utilised, increasing the number of threads will not result in a reduction of the overall execution time of the coarse-grained parallel algorithm. \fig{fig:motivfig}a shows the workload imbalance exhibited by the coarse-grained parallel algorithms in practice. Section~\ref{sect:experiments} demonstrates the limited scalability of coarse-grained parallel algorithms in further detail.

%% file: figures/wcExample.tex
\begin{figure}[t]
    \centering
   	\centerline{
		\includegraphics[width=0.8\linewidth]{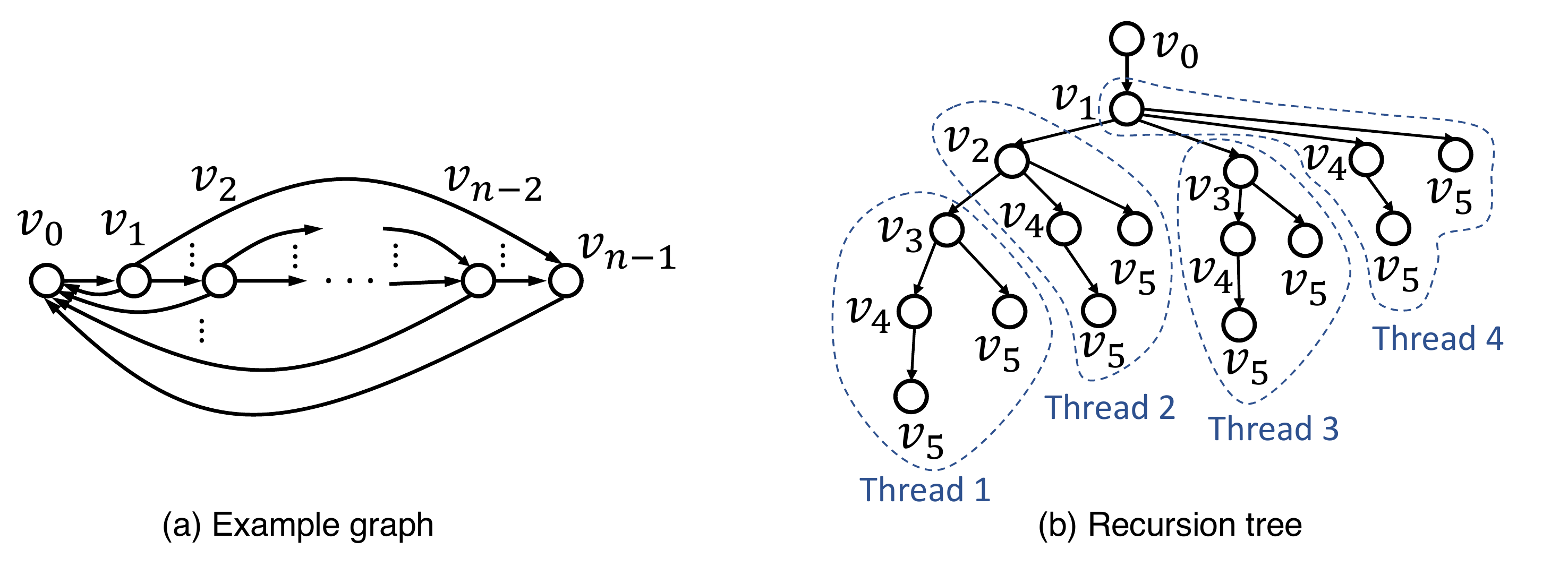}
	}
	\vspace{-.15in}
	\caption{
	(a) A graph with an exponential number of simple cycles.
	(b) The recursion tree of the Johnson algorithm for $n=6$ constructed when the algorithm starts from $v_0$. Whereas a coarse-grained parallel algorithm explores the complete recursion tree using a single thread, our fine-grained parallel algorithms can explore different regions of the recursion tree in parallel using several threads.
	}
	\vspace{-.1in}
	\label{fig:wc_example}
\end{figure}

%% file: tables/workDepth.tex
\begin{table}[t]
\centering
\caption{
Work and depth of the coarse- and fine-grained parallel algorithms.}
\vspace{-.05in}
\addtolength{\tabcolsep}{-1pt}
\begin{tabular}{l|cc}
\textbf{Parallel algorithm}               &  \textbf{Work} & \textbf{Depth}  \\ \hline
    Coarse-grained algorithms &  $O\left(n+e+ec\right)$  & $O\left(ec\right)$        \\
    Fine-grained Johnson algorithm  &  $O\left(n+e+\min\{pce, se\}\right)$ & $O\left(e\right)$   \\
    Fine-grained Read-Tarjan algorithm & $O\left(n+e+ec\right)$  & $O\left(ne\right)$     \\ \hline
\end{tabular}
\label{tab:workDepth}
\vspace{-.1in}
\end{table}

%% file: sections/fineGrainedJ.tex
\section{Fine-grained parallel Johnson}
\label{sect:tpJohnson}

\input{figures/fgj_example}

To address the load imbalance issues that manifest themselves in the coarse-grained parallel Johnson algorithm, we introduce the \textit{fine-grained parallel Johnson} algorithm.
The main goal of our fine-grained algorithm is to enable several threads to explore a recursion tree concurrently, as shown in \fig{fig:wc_example}b, where each thread executes a subset of the recursive calls of this tree.
However, enabling concurrent exploration of a recursion tree is in conflict with the sequential depth-first exploration, required by the Johnson algorithm to achieve a high pruning efficiency.

In this section, we first discuss the challenges that arise when parallelising the exploration of a recursion tree of the Johnson algorithm.
Then, we introduce the \textit{copy-on-steal} mechanism used to address these challenges and present our fine-grained parallel Johnson algorithm.
Finally, we theoretically analyse our algorithm and show that it is scalable.

\subsection{Fine-grained parallelisation challenges}
The requirement of the sequential depth-first exploration of the Johnson algorithm makes it challenging to efficiently parallelise this algorithm in a fine-grained manner.
This requirement is enforced by maintaining a set of blocked vertices $\mathit{Blk}$ throughout the exploration of a recursion tree.
If threads exploring the same recursion tree simply share the same set of blocked vertices $\mathit{Blk}$, the parallel algorithm could produce incorrect results.
For example, considering the graph given in \fig{fig:fgj_example}a, a thread exploring the path $\Pi = v_0\rightarrow v_1 \rightarrow u_1 \rightarrow v_2$ visit and block the vertex $u_4$ in this case because $u_4$ cannot participate in a simple cycle that begins with $\Pi$.
Because the threads exploring this graph share the blocked vertices, another thread attempting to discover the cycle $v_0 \rightarrow v_1\rightarrow u_4\rightarrow v_2\rightarrow v_0$ would fail to do so because $u_4$ is blocked.
Therefore, this approach might not discover all cycles in a graph.

To enable several threads to correctly find all cycles while exploring the same recursion tree, the algorithm could forward a new copy of the $\mathit{Blk}$ and $\mathit{Blist}$ data structures when invoking each child recursive call.
However, this approach would redundantly explore many paths in a graph.
The reason is that a recursive call would be unaware of the vertices visited and blocked by other calls that precede it in the depth-first order except for its direct ancestors in the recursion tree.
When enumerating the simple cycles of the graph shown in \fig{fig:fgj_example}a starting from $v_0$, this approach explores all $4 \times 2^{m-1}+3$ maximal simple paths instead of just seven, that the Johnson algorithm would explore.
Hence, this approach exhaustively explores all maximal simple paths in the graph and is identical to the brute-force solution of Tiernan (see Section~\ref{section:back_tiernan}).
Next, we propose a fine-grained parallel algorithm that addresses the aforementioned parallelisation challenges.

\input{sections/fgj-implementation}

\input{sections/fgj-theory.tex}

%% file: figures/fgj_example.tex
\begin{figure}[t]
	\centerline{
		\includegraphics[width=0.75\linewidth]{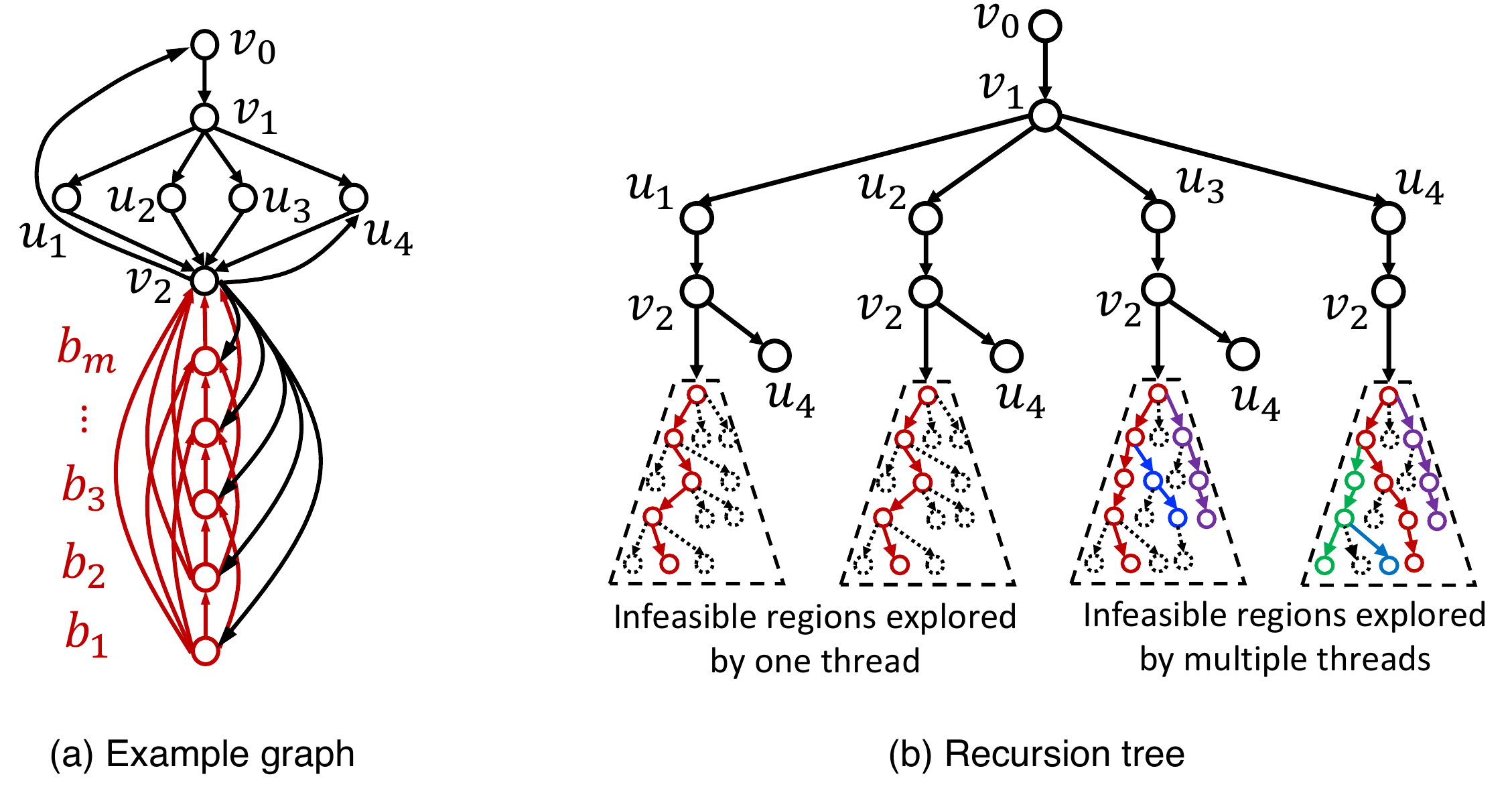}
	}
	\vspace{-.15in}
	\caption{
	(a) An example graph and (b) the recursion tree of our fine-grained Johnson algorithm when enumerating cycles that start from $v_0$.
    Each thread of our fine-grained Johnson algorithm explores the vertices $b_1,\ldots,b_m$ at most once.
	}
	\label{fig:fgj_example}
	\vspace{-.15in}
\end{figure}

%% file: sections/fgj-implementation.tex
\subsection{Copy-on-steal}
\label{sect:fgj_copyOnSteal}

\input{algorithms/FGJ_task}

To enable different threads to concurrently explore the recursion tree in a depth-first fashion while also taking advantage of the powerful pruning capabilities of the Johnson algorithm, each thread executing our fine-grained parallel Johnson algorithm maintains its own copy of the $\Pi$, $\mathit{Blk}$, and $\mathit{Blist}$ data structures.
These data structures are copied between threads only when these threads attempt to explore the same recursion tree.
To achieve this behaviour, our fine-grained parallel Johnson algorithm implements each recursive call of the Johnson algorithm as a separate task.
The pseudocode of this task is given in Algorithm~\ref{algo:fgj_task}, where a data structure $X$, maintained by the thread $T_i$, is denoted as $X_{T_i}$ (see Table~\ref{tab:notation}).
If a child task and its parent task are executed by the same thread $T_i$, the child task reuses the $\Pi_{T_i}$, $\mathit{Blk}_{T_i}$, and $\mathit{Blist}_{T_i}$  data structures of the parent task.
However, if a child task has been stolen---i.e., it is executed by a thread other than the thread that created it, the child task will allocate a new copy of these data structures (line~\ref{algline:fgj_task_stolen} of Algorithm~\ref{algo:fgj_task}).
We refer to this mechanism as \emph{copy-on-steal}.

\input{figures/copyOnSteal_example}

The problem with copying data structures between different threads upon task stealing is that the thread that has created the stolen task (i.e., the \emph{victim thread}) can modify its data structures before this task is stolen by another thread (i.e., the \emph{stealing thread}). 
This problem can be observed in the example shown in \fig{fig:cos_example}.
There, the victim thread $T_1$ and the stealing thread $T_2$ explore the same recursion tree given in \fig{fig:cos_example}b while searching for cycles that start with $P_{1} = v_0 \rightarrow v_1 \rightarrow v_2$ and $P_{2} = v_0 \rightarrow v_1 \rightarrow v_7$, respectively.
In this case, $T_2$ steals a task created by $T_1$ that explores $v_7$, as indicated in \fig{fig:cos_example}b, and receives a copy of the blocked vertices $\mathit{Blk}_{T_1} = \{v_4,v_5,v_6\}$ discovered by $T_1$.
The thread $T_1$ blocked these vertices because they cannot participate in any simple cycle that begins with $P_1$.
If $T_2$ simply uses a copy of these blocked vertices $\mathit{Blk}_{T_1}$ without modifications, $T_2$ will be unable to find the cycle $v_0\rightarrow v_1 \rightarrow v_7 \rightarrow v_4 \rightarrow v_2 \rightarrow v_3 \rightarrow v_0$ because $v_4$ is blocked.
Therefore, a method for unblocking vertices after copy-on-steal is required to correctly find all cycles.

We explore two solutions for this problem:

\textbf{(i) Copy-on-steal with complete unblocking.}
To enable the threads of our algorithm to find cycles after performing copy-on-steal, the stealing thread could unblock all vertices that the victim thread had blocked after creating the stolen task.
In our example given in \fig{fig:cos_example}, the stealing thread $T_2$ unblocks all vertices $\mathit{Blk}_{T_1} = \{v_4,v_5,v_6\}$ it received from the victim thread $T_1$.
Although this approach enables $T_2$ to correctly find cycles, it also fails to take advantage of the information collected by $T_1$ to reduce the redundant work of $T_2$.
For instance, in \fig{fig:cos_example}, $T_2$ visits $v_5$ and $v_6$, even though $T_1$ already concluded that these vertices cannot participate in any simple cycle that begins with $P = v_0 \rightarrow v_1$, where $P$ is the largest common prefix of all the paths explored by $T_1$ and $T_2$.
As a result, $T_2$ redundantly visits the dotted part of the recursion tree given in \fig{fig:cos_example}b.

\textbf{(ii) Copy-on-steal with recursive unblocking.}
In this approach, the stealing thread capitalises on the information already discovered by the victim thread.
The stealing thread $T_2$ can reuse a subset $B \subset \mathit{Blk}_{T_1}$ of the blocked vertices discovered by $T_1$ if the vertices in $B$ cannot participate in simple cycles that begin with $P$, where $P$ is the largest common prefix of all the paths explored by $T_1$ and $T_2$.
Because any path discovered by $T_2$ begins with $P$, $T_2$ can avoid visiting vertices from $B$.
Thus, to correctly find simple cycles, it is sufficient for $T_2$ to unblock the vertices from $\mathit{Blk}_{T_1} \setminus B$.
To achieve this behaviour, $T_2$ invokes a recursive unblocking procedure of the Johnson algorithm for every vertex $v \in \mathit{\Pi}_{T_1} \setminus P$, as shown in Algorithm~\ref{algo:fgj_cos}, where $\mathit{\Pi}_{T_1}$ is the path $T_1$ is exploring during task stealing.
The vertices in $B$ can only be unblocked by a recursive unblocking invoked for $v \in P$; hence, the vertices in $B$ remain blocked.
In the example given in \fig{fig:cos_example}, $T_2$ invokes a recursive unblocking procedure for $\mathit{\Pi}_{T_1} \setminus P = \{v_2\} $, which results in unblocking of $v_4$.
Thus, $T_2$ is able to discover a cycle that contains $v_4$. 
The vertices $B = \{ v_5, v_6 \}$ will not be unblocked because they cannot take part in any simple cycle that begins with $P  = v_0 \rightarrow v_1$.
Therefore, thread $T_2$ avoids visiting the dotted part of the recursion tree given in \fig{fig:cos_example}b.

\input{algorithms/FGJ_copyOnSteal}

\input{algorithms/fgj}

Without countermeasures, our algorithm can suffer from race conditions because its data structures can be accessed concurrently by different threads.
For instance, a stealing thread $T_2$ can copy the data structures of a victim thread $T_1$ while $T_1$ performs a recursive unblocking, in which case $T_2$ could receive the vertex set $\mathit{Blk}_{T_1}$ that is partially unblocked.
When using copy-on-steal with recursive unblocking, $T_2$ may not be able to continue the interrupted unblocking of $\mathit{Blk}_{T_1}$, causing the algorithm to miss certain cycles.
To avoid this problem, we define critical sections in lines~\ref{algline:fgj_task_lock}--\ref{algline:fgj_task_unlock} of Algorithm~\ref{algo:fgj_task} and in lines~\ref{algline:fgj_cos_lock}--\ref{algline:fgj_cos_unlock} of Algorithm~\ref{algo:fgj_cos} using coarse-grained locking by maintaining a mutex per thread. However, such a locking mechanism is not required when using copy-on-steal with complete unblocking because $T_2$ can correctly unblock vertices in $\mathit{Blk}_{T_1}$ simply by removing all vertices from $\mathit{Blk}_{T_1}$ inserted after the stolen task was created.
Thus, it is sufficient to enable thread-safe operations on $\Pi$, $\mathit{Blk}$, and $\mathit{Blist}$ using fine-grained locking.
As a result, the critical sections are shorter when the copy-on-steal with complete unblocking approach is used.

Nevertheless, we opt to use the copy-on-steal with recursive unblocking approach in our fine-grained parallel Johnson algorithm because this approach leads to less redundant work and rarely suffers from synchronisation overheads.
The pseudocode of our fine-grained parallel Johnson algorithm is given in Algorithm~\ref{algo:fgj}.

%% file: algorithms/FGJ_task.tex
\begin{algorithm}[t]
\SetAlgoLined
\SetKwProg{Fn}{Function}{}{}
\SetKwProg{Proc}{Procedure}{}{}
\SetKwProg{Task}{\textcolor{orange}{Task}}{}{}
\SetKwInOut{InOut}{InOut}
\SetKwComment{Comment}{$\triangleright$\ }{}
\KwIn{$\mathrm{v}$ - the current vertex, $\mathrm{v_0}$ - the starting vertex\DontPrintSemicolon\;
\hspace*{10mm}$\mathrm{d}$ - the depth of this task}
\InOut{$\mathrm{T_1}$ - the thread that created this task\Comment*[f]{$\mathrm{T_1}$ maintains $\Pi_{\mathrm{T_1}}$, $\mathrm{Blk}_{\mathrm{T_1}}$, $\mathrm{Blist}_{\mathrm{T_1}}$, and $\mathrm{Mutex}_{\mathrm{T_1}}$}}

\KwOut{\textit{true} if a cycle was found}
    \BlankLine
    $\mathrm{T_2}$ = the thread executing this task\Comment*[r]{$\mathrm{T_2}$ maintains $\Pi_{\mathrm{T_2}}$, $\mathrm{Blk}_{\mathrm{T_2}}$, $\mathrm{Blist}_{\mathrm{T_2}}$, and $\mathrm{Mutex}_{\mathrm{T_2}}$}
    \lIf(\Comment*[f]{Check if this task is stolen}){$\mathrm{T_1} \neq \mathrm{T_2}$}{FGJ\_copyOnSteal($\mathrm{d}$, $\mathrm{T_1}$, $\mathrm{T_2}$)} \label{algline:fgj_task_stolen} 
    $\mathrm{Mutex}_{\mathrm{T_2}}.\mathrm{lock}()$\;
    $\Pi_{\mathrm{T_2}}.\mathrm{push}\mathrm{(v)}$; \label{algline:fgj_task_push}
    $\mathrm{Blk}_{\mathrm{T_2}} = \mathrm{Blk}_{\mathrm{T_2}} \cup \{\mathrm{v}\}$\;
    $\mathrm{Mutex}_{\mathrm{T_2}}.\mathrm{unlock}()$\;

    \BlankLine
    $\mathrm{found} = \mathit{false}$\;
 	\ForEach(\Comment*[f]{Recursively explore the neighbours of $\mathrm{v}$}){$\mathrm{u}$ : $\mathcal{N}\mathrm{(v)}$ $\mathbf{s.t.}$ $\mathrm{u.id} > \mathrm{v_0.id}$}{
 		\If{$\mathrm{u} = \mathrm{v_0}$}{
 			report cycle $\Pi_{\mathrm{T_2}}$\;
            $\mathrm{found} = \mathit{true}$\;
 		}
 		\ElseIf{$\mathrm{u} \notin \mathrm{Blk}_{\mathrm{T_2}}$}{
 			$\mathrm{f} =$ \textcolor{orange}{\textbf{spawn}} FGJ\_task($\mathrm{u}$, $\mathrm{v_0}$, $\mathrm{d+1}$, $\mathrm{T_2}$) \Comment*[r]{Create a child task}
 			$\mathrm{found} = \mathrm{found} \lor \mathrm{f}$\;
 		}
	}
	\textcolor{orange}{\textbf{wait}} for the spawned tasks\;
	
    $\mathrm{Mutex}_{\mathrm{T_2}}.\mathrm{lock}()$\;   \label{algline:fgj_task_lock}
	$\Pi_{\mathrm{T_2}}.\mathrm{pop}()$\;
    \BlankLine
	\If(\Comment*[f]{Unblock vertices if a cycle was found}){$\mathrm{found}$} {RecursiveUnblock($\mathrm{v}$, $\mathrm{Blk}_{\mathrm{T_2}}$, $\mathrm{Blist}_{\mathrm{T_2}}$)\;}
	\Else {\lForEach{$\mathrm{u}$ : $\mathcal{N}\mathrm{(v)}$}{$\mathrm{Blist}_{\mathrm{T_2}}[\mathrm{u}] = \mathrm{Blist}_{\mathrm{T_2}}\mathrm{[u]} \cup \{\mathrm{v}\}$}}
    $\mathrm{Mutex}_{\mathrm{T_2}}.\mathrm{unlock}()$\; \label{algline:fgj_task_unlock}
	\Return $\mathrm{found}$\;
\caption{FGJ\_task ($\mathrm{v}$, $\mathrm{v_0}$, $\mathrm{d}$, $\mathrm{T_1}$)}
\label{algo:fgj_task}
\end{algorithm}

%% file: figures/copyOnSteal_example.tex
\begin{figure}[t]
	\centerline{
		\includegraphics[width=0.74\linewidth]{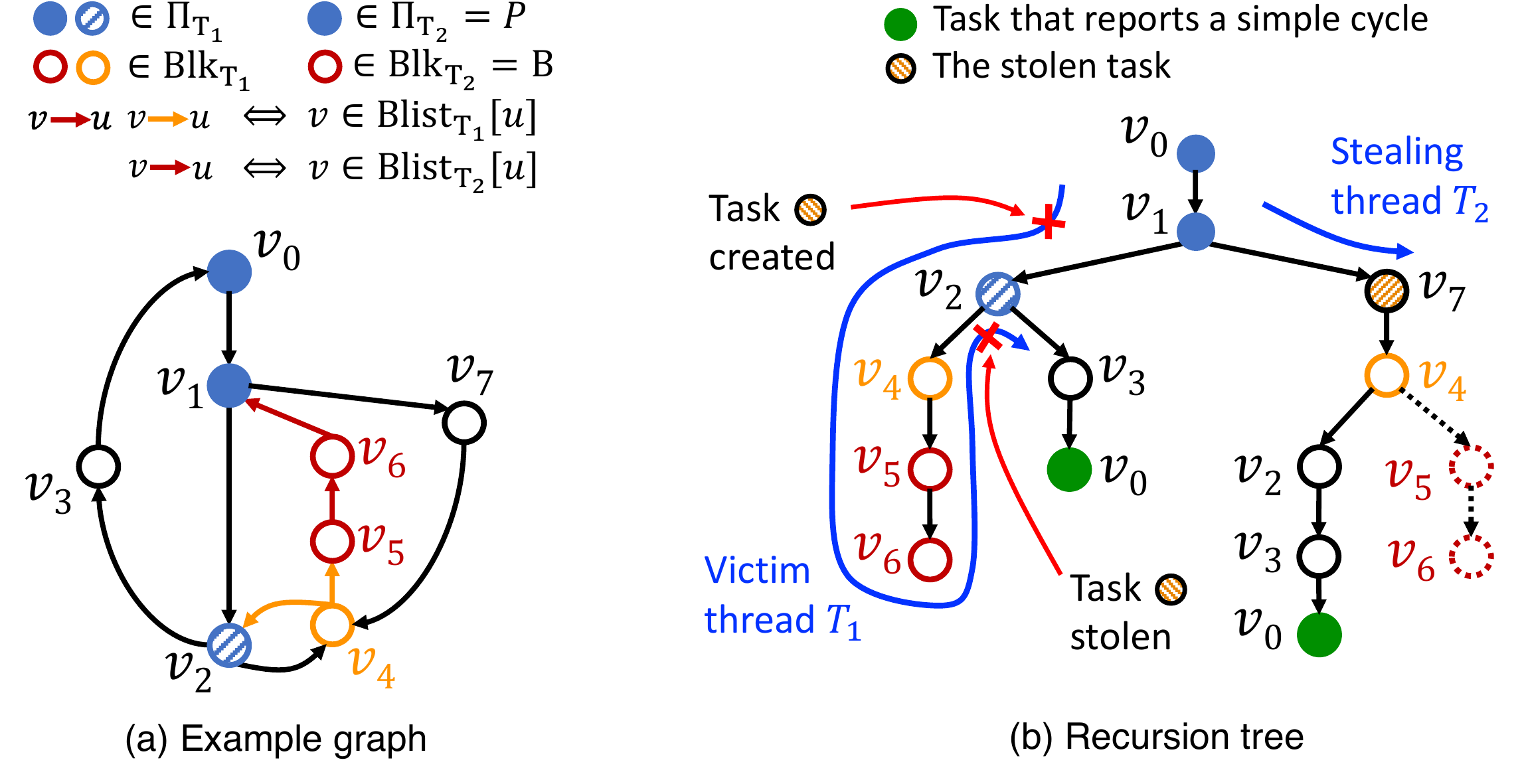}
	}
	\vspace{-.15in}
	\caption{
    (a) An example graph and (b) the recursion tree of our fine-grained Johnson algorithm when enumerating simple cycles that start from $v_0$.
        Here, $X_{T_i}$ denotes a data structure $X$ of the thread $T_i$.
	The thread $T_2$ can prune the dotted part of the tree by avoiding $v_5$ and $v_6$ that the thread $T_1$ has blocked after creating the task stolen by~$T_2$.
	}
	\label{fig:cos_example}
	\vspace{-.1in}
\end{figure}

%% file: algorithms/FGJ_copyOnSteal.tex
\begin{algorithm}[t]
\SetAlgoLined
\SetKwProg{Fn}{Function}{}{}
\SetKwProg{Proc}{Procedure}{}{}
\SetKwProg{Task}{\textcolor{orange}{Task}}{}{}
\SetKwInOut{InOut}{InOut}
\SetKwComment{Comment}{$\triangleright$\ }{}
\KwIn{$\mathrm{d}$ - the depth of the task executing this function}
\InOut{$\mathrm{T_1}$ - the victim thread\DontPrintSemicolon\;
\hspace*{10mm}$\mathrm{T_2}$ - the stealing thread}
$\mathrm{Mutex}_{\mathrm{T_1}}.\mathrm{lock}()$\; \label{algline:fgj_cos_lock}
$\{\Pi_{\mathrm{T_2}}, \mathrm{Blk}_{\mathrm{T_2}}, \mathrm{Blist}_{\mathrm{T_2}}\} = \mathrm{copy}\left(\{\Pi_{\mathrm{T_1}}, \mathrm{Blk}_{\mathrm{T_1}}, \mathrm{Blist}_{\mathrm{T_1}}\}\right)$\Comment*[r]{The data structures of $T_1$ are copied to $T_2$}
$\mathrm{Mutex}_{\mathrm{T_1}}.\mathrm{unlock}()$\; \label{algline:fgj_cos_unlock}
\While(\Comment*[f]{Copy-on-steal with recursive unblocking}){$\left|\Pi_{\mathrm{T_2}}\right| \geq \mathrm{d}$}{
    $\mathrm{u} = \Pi_{\mathrm{T_2}}.\mathrm{pop}()$\;
    RecursiveUnblock($\mathrm{u}$, $\mathrm{Blk}_{\mathrm{T_2}}$, $\mathrm{Blist}_{\mathrm{T_2}}$)\; \label{algline:fgj_cos_recunblock}
}
\caption{FGJ\_copyOnSteal ($\mathrm{d}$, $\mathrm{T_1}$, $\mathrm{T_2}$)}
\label{algo:fgj_cos}
\end{algorithm}

%% file: algorithms/FGJ.tex
\begin{algorithm}[t]
\SetAlgoLined
\SetKwProg{Fn}{Function}{}{}
\SetKwProg{Proc}{Procedure}{}{}
\SetKwProg{Task}{\textcolor{orange}{Task}}{}{}
\SetKwInOut{InOut}{InOut}
\SetKwComment{Comment}{$\triangleright$\ }{}

\KwIn{$\mathcal{G}$ - the input graph with vertices $\mathcal{V}$ and edges $\mathcal{E}$}
\textcolor{orange}{\textbf{parallel}} \ForEach{$v_0 \rightarrow v$ : $\mathcal{E}$}{
     $\mathrm{T_0}$ = the thread executing this loop iteration\Comment*[r]{$\mathrm{T_0}$ maintains $\Pi_{\mathrm{T_0}}$, $\mathrm{Blk}_{\mathrm{T_0}}$, $\mathrm{Blist}_{\mathrm{T_0}}$, and $\mathrm{Mutex}_{\mathrm{T_0}}$}
    $\Pi_{\mathrm{T_0}} = v_0$;\hspace*{2mm} $\mathrm{Blk}_{\mathrm{T_0}} = \varnothing$\;
    \lForEach{$u$ : $\mathcal{V}$} {$\mathrm{Blist}_{\mathrm{T_0}}[u] = \varnothing$}
    \textcolor{orange}{\textbf{spawn}} FGJ\_task($v$, $v_0$, $1$, $\mathrm{T_0}$)\Comment*[r]{Create a task}
}
\textcolor{orange}{\textbf{wait}} for all spawned tasks\;

\caption{FGJ $\left(\mathcal{G}(\mathcal{V}, \mathcal{E})\right)$}
\label{algo:fgj}
\end{algorithm}

%% file: sections/fgj-theory.tex
\subsection{Theoretical analysis}

We now show that the fine-grained parallel Johnson algorithm is scalable but not work-efficient.

\input{theorems/fgJohn_workEff}

The work inefficiency of our fine-grained parallel Johnson algorithm occurs if more than one thread performs the work the sequential Johnson algorithm would perform between the discovery of two cycles.
This behaviour can be illustrated using the graph from \fig{fig:fgj_example}a, which contains $c = 4$ cycles and $s = c \times 2^{m-1}+3$ maximal simple paths, each starting from vertex $v_0$.
When discovering each cycle, our fine-grained algorithm explores an infeasible region of the recursion tree, as shown in \fig{fig:fgj_example}b, in which the vertices $b_1,\ldots,b_m$ are visited.
If this infeasible region is explored using a single thread, each vertex $b_i$, with $i\in \{1,\ldots,m\}$, will be visited exactly once.
However, if $p$ threads are exploring the same infeasible region of the recursion tree, vertices $b_1,\ldots,b_m$ will be visited up to $p$ times because the threads are unaware of each other's blocked vertices.
In this case, the fine-grained parallel Johnson algorithm performs more work than necessary, and, thus, it is not work-efficient.
Additionally, each infeasible region of the recursion tree that visits vertices $b_1,\ldots,b_k$ can be executed by at most $s/c = 2^{m-1}$ threads because there are $2^{m-1}$ maximal simple paths that can be explored in each infeasible region.
In this case, each vertex $b_i$, with $i\in \{1,\ldots,m\}$, is visited up to $s$ times, and, thus, the fine-grained parallel Johnson algorithm behaves as the Tiernan algorithm (see Section~\ref{section:back_tiernan}).

\input{theorems/fgj_TinfLemma}

\input{theorems/fgJohn_scal}

For the fine-grained parallel Johnson algorithm to be scalable, it is sufficient for $c$ to increase sublinearly with $n$.
Even though this algorithm is scalable, a strong or weak scalability is not guaranteed due to the work inefficiency of this algorithm.
Nevertheless, our experiments show that this algorithm is strongly scalable in practice (see \fig{fig:scalfig_sc}).

\subsection{Summary}
Our relaxation of the strictly depth-first-search-based recursion-tree exploration reduces the pruning efficiency of the Johnson algorithm. 
In the worst case, the fine-grained parallel Johnson algorithm could perform as much work as the brute-force Tiernan algorithm does---i.e., $O(se)$. However, in practice, this worst-case scenario does not happen (see Section~\ref{sect:experiments}).
In addition, our fine-grained parallel Johnson algorithm can suffer from synchronisation issues in some rare cases (see Section~\ref{sect:experiments}) because our copy-on-steal mechanism can lead to long critical sections. In the next section, we introduce a fine-grained parallel algorithm that is scalable, work-efficient, and less prone to synchronisation issues.

%% file: theorems/fgJohn_workEff.tex
\begin{theorem}
The fine-grained parallel Johnson algorithm is not work-efficient.
\label{theorem:fgj_workEfficiency}
\vspace{-.07in}
\end{theorem}

\begin{proof}
According to Lemma~3 presented by Johnson~\cite{johnson_finding_1975}, a vertex cannot be unblocked more than once unless a cycle is found, and once a vertex is visited, it can be visited again only after being unblocked.
Thus, the Johnson algorithm visits each vertex and edge at most $c$ times.
In the fine-grained parallel Johnson algorithm executed using $p$ threads, each thread maintains a separate set of data structures used for managing blocked vertices.
Because the threads are unaware of each other's blocked vertices, each edge is visited at most $p c$ times, $c$ times by each thread.
Additionally, an edge cannot be visited more than $s$ times because each maximal simple path of a graph is explored by a different thread in the worst case, and during each simple path exploration, an edge is visited at most once. 
Therefore, the maximum number of times an edge can be visited by the fine-grained parallel Johnson algorithm is $\min\left\{pc, s\right\}$.
Because the algorithm executes in $O(n+e)$ time if there does not exist a cycle or a path in the input graph, the work performed by the fine-grained parallel Johnson algorithm is
\vspace{-.05in}
 \begin{equation}
 W_p(n) \in  O\left(n+e+\min\{pce, se\}\right).
 \end{equation} 
When $c > 0$, $p>1$, and $s>c$, the work performed by the fine-grained parallel Johnson algorithm $W_p(n)$ is greater than the execution time $T_1(n)$ of the sequential Johnson algorithm. 
Thus, this algorithm is not work-efficient.
\vspace{-.05in}
\end{proof}

%% file: theorems/fgJ_TinfLemma.tex
\begin{lemma}
The depth $T_{\infty}(n)$ of the fine-grained parallel Johnson algorithm is in $O(e)$.
\vspace{-.07in}
\label{lemma:fgJ_tinf}
\end{lemma}

\begin{proof}
The worst-case depth of this algorithm occurs when a thread performs copy-on-steal and explores a maximal simple path.
A thread explores such a path in $O(e)$ time because it visits at most $e$ edges.
As a result, $\Pi$ and $\mathit{Blk}$ contain at most $e$ vertices, and $\mathit{Blist}$ contains at most $e$ pairs of vertices.
Therefore, copy-on-steal requires $O(e)$ time to copy $\Pi$, $\mathit{Blk}$, and $\mathit{Blist}$, and to unblock vertices in $\mathit{Blk}$.
As a result, the depth of this algorithm is $T_{\infty}(n) \in O(e)$.
\vspace{-.05in}
\end{proof}

%% file: theorems/fgJohn_scal.tex
\begin{theorem}
\label{theorem:fgJohnScal}
The fine-grained parallel Johnson algorithm is scalable when $\lim\limits_{n \to \infty} c = \infty$.
\vspace{-.08in}
\end{theorem}

\begin{proof}
For this algorithm, $T_1(n) \in O(n+e+ec)$ and $T_{\infty}(n) \in O(e)$ (see Lemma~\ref{lemma:fgJ_tinf}).
Given $e < n^2$ and our assumption that $\lim\limits_{n \to \infty} c = \infty$, we have $\lim\limits_{n\to\infty}\dfrac{T_{\infty}(n)}{T_1(n)}  =  \lim\limits_{n\to\infty} \dfrac{e}{n+e+ec} = 0$.
Thus, this algorithm is scalable based on Definition~\ref{def:scalability}.
\vspace{-.07in}
\end{proof}

%% file: sections/fineGrainedRT.tex
\section{Fine-grained parallel Read-Tarjan}
\label{sect:tpReadTarjan}

In this section, we first introduce several optimisations that reduce the number of unnecessary vertex visits performed by the sequential Read-Tarjan algorithm.
Then, we present our fine-grained parallel Read-Tarjan algorithm that includes these optimisations.
Finally, we show that our parallel algorithm is work-efficient and strongly-scalable.

\subsection{Improvements to the pruning efficiency}
\label{sect:impRT}

To improve the pruning efficiency of the sequential Read-Tarjan algorithm, we include the following optimisations:

\input{figures/fgrt_example}

\textbf{(i) Blocked vertex set forwarding} enables a recursive call of the Read-Tarjan algorithm to reuse vertices blocked by its parent call, resulting in fewer vertex visits.
The original Read-Tarjan algorithm discards blocked vertices after each recursive call~\cite{read_bounds_1975}, even though this information could be reused later.
In this optimisation, the algorithm forwards the blocked vertices $\mathit{Blk}$ of a recursive call to its child recursive calls, preventing those child calls from unnecessarily visiting the vertices in $\mathit{Blk}$ again.
For example, in \fig{fig:fgrt_example}, the vertex $v_8$ is blocked the first time the algorithm visits $v_8$ while exploring the path extension $E_1$. 
This optimisation prevents the algorithm from visiting $v_8$ again when exploring the same extension $E_1$ or another extension $E_3$ that branches from $E_1$.
As a result of this optimisation, the algorithm can avoid the dotted part of the recursion tree.

\textbf{(ii) Path extension forwarding} prevents recomputation of the path extension $E$ found by a parent recursive call by forwarding this path extension to its child recursive call.
In this way, each child recursive call performs one fewer DFS invocation than the original Read-Tarjan algorithm~\cite{read_bounds_1975}.

\input{algorithms/FGRT_DFS}

\textbf{(iii) Blocking on a successful DFS} is another mechanism for discovering vertices to be blocked.
As a reminder, the Read-Tarjan algorithm searches for path extensions using a DFS.
In the original algorithm, a vertex is blocked only if it is visited during an unsuccessful DFS invocation, which fails to discover a path extension.
However, successful DFS invocations could also visit some vertices that have all their neighbours blocked.
Such vertices cannot lead to the discovery of new cycles and, thus, can also be blocked.
The pseudocode of the DFS function that includes this optimisation is given in Algorithm~\ref{algo:fgrt_dfs}.
In our example given in \fig{fig:fgrt_example}, a successful DFS invoked from $v_3$ finds a path extension $E_3$ and discovers that the only neighbour $v_8$ of $v_7$ is blocked.
The algorithm then blocks $v_7$, which enables it to avoid visiting $v_7$ again when exploring $E_3$. 
Therefore, fewer vertices are visited during the execution of the algorithm.

\subsection{Fine-grained parallelisation}
\label{sect:fgrt_algo}

Although the optimisations presented in Section~\ref{sect:impRT} eliminate some of the redundant work performed by the Read-Tarjan algorithm, this algorithm typically performs more work than the Johnson algorithm (see Section~\ref{sect:read_tarjan}). 
However, this redundancy makes it possible to parallelise the Read-Tarjan algorithm in a scalable and work-efficient manner.

Because the Read-Tarjan algorithm allocates a new $\mathit{Blk}$ set for each path extension exploration, a recursive call can explore different path extensions in an arbitrary order.
In addition, discovery of a new path extension $E$ results in the invocation of a single recursive call, and these calls can be executed in an arbitrary order.
As a result, several threads can concurrently explore different paths of the same recursion tree constructed by the Read-Tarjan algorithm for a given starting edge. 
There are neither data dependencies nor ordering requirements between different calls, apart from those that exist between a parent and a child.
To exploit the parallelism available during the recursion tree exploration, we execute each path extension exploration in each recursive call as a separate task, all of which can be independently executed.
Examples of such tasks are shown in \fig{fig:fgrt_example}.
We refer to the resulting algorithm as the \textit{fine-grained parallel Read-Tarjan} algorithm.

\input{algorithms/FGRT_task}

\input{algorithms/FGRT}

Our implementation shown in Algorithm~\ref{algo:fgrt_task} performs only a single path extension exploration in a recursive call and uses all the optimisations we introduced in Section~\ref{sect:impRT}.
We execute each such recursive call as a separate task using a dynamic thread scheduling framework (see Section~\ref{sect:tlp}).
To find all cycles of a graph, we execute a \textit{parallel for} loop iteration for each edge $v_0 \rightarrow v$ that uses Algorithm~\ref{algo:fgrt_dfs} to search for a path extension $E$ from $v$ to $v_0$, as shown in Algorithm~\ref{algo:fgrt}
If such $E$ exists, a task is created using $v$, $v_0$, and $E$ as its input parameters.
This task then recursively creates new tasks, as shown in lines~\ref{algline:fgrt_task_spawn} and~\ref{algline:fgrt_task_spawn2} of Algorithm~\ref{algo:fgrt_task}, until all cycles that start with the edge $v_0 \rightarrow v$ have been discovered.

To prevent different threads from concurrently modifying $\Pi$ and $\mathit{Blk}$, each task allocates and maintains its own $\Pi$ and $\mathit{Blk}$ sets. A task can receive a copy of $\Pi$ and $\mathit{Blk}$ directly from its parent task at the time of task creation.
However, it is possible to minimise the copy overheads by copying these sets only when a task is stolen.
For this purpose, we use the copy-on-steal with complete unblocking approach described in Section~\ref{sect:fgj_copyOnSteal}, which has shorter critical sections than the copy-on-steal with recursive unblocking approach used by our fine-grained parallel Johnson algorithm.

\input{sections/fgrt-theory}

%% file: figures/fgrt_example.tex
\begin{figure}[t]
	\centerline{
		\includegraphics[width=0.75\linewidth]{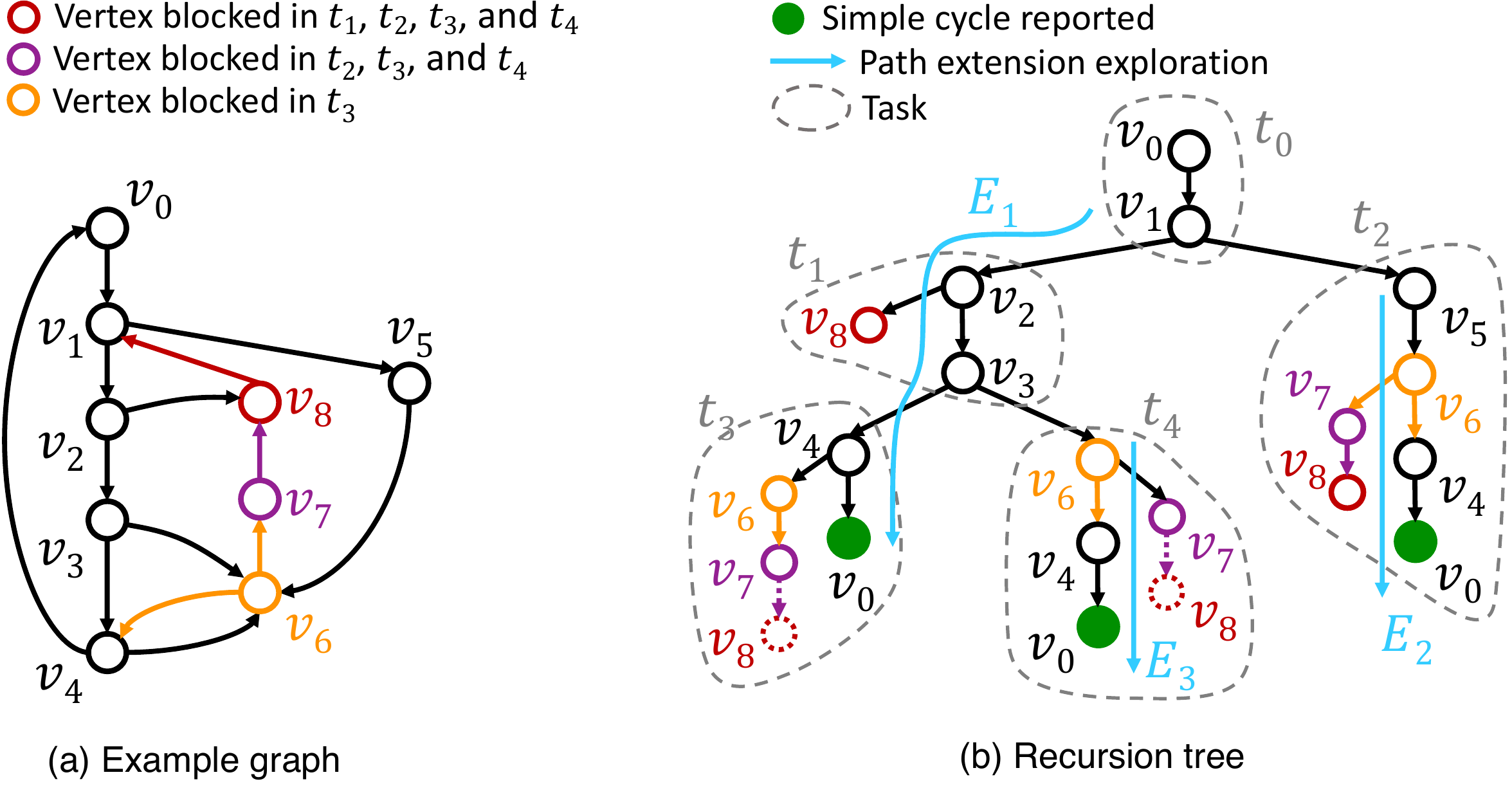}
	}
	\vspace{-.15in}
	\caption{
    (a) An example graph and (b) the recursion tree of our fine-grained parallel Read-Tarjan algorithm when enumerating cycles that start from $v_0$.
    The nodes of the recursion tree represent the recursive calls of the depth-first search.
        Tasks shown in (b) can be executed independently of each other.
	}
	\label{fig:fgrt_example}
	\vspace{-.1in}
\end{figure}

%% file: algorithms/FGRT_DFS.tex
\begin{algorithm}[t]
\SetAlgoLined
\SetKwProg{Fn}{Function}{}{}
\SetKwInOut{InOut}{InOut}
\SetKwComment{Comment}{$\triangleright$\ }{}
    \KwIn{$\mathrm{u}$ - the current vertex, $\mathrm{v_0}$ - the starting vertex}
    \InOut{$\mathrm{Blk}$ - blocked vertices\DontPrintSemicolon\;
    \hspace*{10mm}$\mathrm{Vis}$ - vertices visited during the DFS}
    \KwOut{$\mathrm{E}$ - the resulting path extension from $\mathrm{u}$ to $\mathrm{v_0}$}
    
    \lIf{$\mathrm{u} = \mathrm{v_0}$}{\KwRet{$\mathrm{u}$}}
    
    $\mathrm{Vis} = \mathrm{Vis} \cup \{\mathrm{u}\}$\;
    $\mathrm{block} = \mathit{true}$\;
    
    \ForEach{$\mathrm{w}$ : $\mathcal{N}\mathrm{(u)}$  $\mathbf{s.t.}$ $\mathrm{w}.\mathrm{id} > \mathrm{v_0}.\mathrm{id}$}{
 		\If{$\mathrm{w} = \mathrm{v_0}$}{\KwRet{$\mathrm{u} \rightarrow \mathrm{w}$}}
 		\ElseIf{$\mathrm{w} \notin \mathrm{Blk} \land \mathrm{w} \notin \mathrm{Vis}$}{
 			$\mathrm{E}$ = FGRT\_DFS($\mathrm{w}$, $\mathrm{v_0}$ $\mathrm{Blk}$, $\mathrm{Vis}$)\Comment*[r]{Recursively search for the path extension $E$}
 			\If{$\mathrm{E} \neq \varnothing$} {\KwRet{$\mathrm{E}.\mathrm{push\_front}(\mathrm{u})$}\;}
 		}
 		\If{$\mathrm{w} \notin \mathrm{Blk}$}{ $\mathrm{block} = \mathit{false}$}
    }
    \lIf(\Comment*[f]{Blocking on a successful DFS}){$\mathrm{block}$}{$\mathrm{Blk} = \mathrm{Blk} \cup \{\mathrm{u}\}$} \label{algline:fgrt_dfs_blocking}
    \KwRet{$\varnothing$}\;

\caption{FGRT\_DFS($\mathrm{u}$, $\mathrm{v_0}$, $\mathrm{Blk}$, $\mathrm{Vis}$)}
\label{algo:fgrt_dfs}
\end{algorithm}

%% file: algorithms/FGRT_task.tex
\begin{algorithm}[t]
\SetAlgoLined
\SetKwProg{Fn}{Function}{}{}
\SetKwProg{Proc}{Procedure}{}{}
\SetKwInOut{InOut}{InOut}
\SetKwComment{Comment}{$\triangleright$\ }{}

    \KwIn{$\mathrm{v}$ - the current vertex, $\mathrm{v_0}$ - the starting vertex\DontPrintSemicolon\;
    \hspace*{10mm}$\mathrm{E}$ - the path extension from $\mathrm{v}$ to $\mathrm{v_0}$\DontPrintSemicolon\;
    \hspace*{10mm}$\mathrm{d}$ - the depth of this task}
    \InOut{$\mathrm{T_1}$ - the thread that created this task\Comment*[f]{$\mathrm{T_1}$ maintains $\Pi_{\mathrm{T_1}}$ and $\mathrm{Blk}_{\mathrm{T_1}}$}}
    \BlankLine
    $\mathrm{T_2}$ = the thread executing this task\Comment*[r]{$\mathrm{T_2}$ maintains $\Pi_{\mathrm{T_2}}$ and $\mathrm{Blk}_{\mathrm{T_2}}$} \label{algline:fgrt_task_cos1}
    \If(\Comment*[f]{Check if this task is stolen}){$\mathrm{T_1} \neq \mathrm{T_2}$}{ 
        $\{\Pi_{\mathrm{T_2}}, \mathrm{Blk}_{\mathrm{T_2}} \} = \mathrm{copy}\left(\{\Pi_{\mathrm{T_1}}, \mathrm{Blk}_{\mathrm{T_1}}\}\right)$\Comment*[r]{Operations on $\Pi$ and $\mathrm{Blk}$ are thread-safe}} \label{algline:fgrt_task_cos2}
    \lWhile{$\Pi_{\mathrm{T_2}}.\mathrm{back}() \neq \mathrm{v}$}{$\Pi_{\mathrm{T_2}}.\mathrm{pop}()$} \label{algline:fgrt_task_rem1}
    Remove vertices from $\mathrm{Blk}_{\mathrm{T_2}}$ inserted at depth $\mathrm{d}^{\prime} \geq \mathrm{d}$;\label{algline:fgrt_task_rem2}
    
     $\mathrm{found} = \mathit{false}$\;
     \While(\Comment*[f]{Exploration of the path extension $\mathrm{E}$}){$\mathrm{E} \neq \varnothing$}{
         $\mathrm{v} = \mathrm{E}.\mathrm{pop\_front}$()\;
         $\Pi_{\mathrm{T_2}} = \Pi_{\mathrm{T_2}}.\mathrm{push}(\mathrm{v})$; $\mathrm{Blk}_{\mathrm{T_2}} = \mathrm{Blk}_{\mathrm{T_2}} \cup \{\mathrm{v}\}$\;
         \ForEach{$\mathrm{u}$ : $\mathcal{N}\mathrm{(v)}$ $\mathbf{s.t.}$ $\mathrm{u}.\mathrm{id} > \mathrm{v_0}.\mathrm{id}$} {
             \If{$\mathrm{u} \neq \mathrm{E}.\mathrm{front}()\land  \mathrm{u} \notin \mathrm{Blk}_{\mathrm{T_2}}$}{
                 $\mathrm{E}^{\prime}$ = FGRT\_DFS($\mathrm{u}$, $\mathrm{v_0}$, $\mathrm{Blk}_{\mathrm{T_2}}$, $\mathrm{Vis} = \varnothing$)\Comment*[r]{Find an alternate path extension $\mathrm{E}^{\prime}$}
                 \If{$\mathrm{E}^{\prime} \neq \varnothing$}{\textcolor{orange}{\textbf{spawn}} FGRT\_task($\mathrm{v}$, $\mathrm{v_0}$, $\mathrm{E}^{\prime}$, $\mathrm{d+1}$, $\mathrm{T_2}$)\Comment*[r]{Create a child task}  \label{algline:fgrt_task_spawn}
                    $\mathrm{found} = \mathit{true}$\;
                 }
                 \lElse{$\mathrm{Blk}_{\mathrm{T_2}} = \mathrm{Blk}_{\mathrm{T_2}} \cup \mathrm{Vis}$} \label{algline:fgrt_task_blk}
             }
         }
         
         \lIf{$\mathrm{found}$}{\textbf{break}}
     }
     
 	\lIf{$\mathrm{E} = \varnothing$}{   \label{algline:fgrt_goto} report cycle $\Pi_{\mathrm{T_2}}$}
 	\lElse(\Comment*[f]{Create a child task}){\textcolor{orange}{\textbf{spawn}} FGRT\_task($\mathrm{v}$, $\mathrm{v_0}$, $\mathrm{E}$, $\mathrm{d+1}$, $\mathrm{T_2}$)} \label{algline:fgrt_task_spawn2}

\caption{FGRT\_task($\mathrm{v}$, $\mathrm{v_0}$, $\mathrm{E}$, $\mathrm{d}$, $\mathrm{T_1}$)}
\label{algo:fgrt_task}
\end{algorithm}

%% file: algorithms/FGRT.tex
\begin{algorithm}[t]
\SetAlgoLined
\SetKwProg{Fn}{Function}{}{}
\SetKwProg{Proc}{Procedure}{}{}
\SetKwProg{Task}{\textcolor{orange}{Task}}{}{}
\SetKwInOut{InOut}{InOut}
\SetKwComment{Comment}{$\triangleright$\ }{}

\KwIn{$\mathcal{G}$ - the input graph with vertices $\mathcal{V}$ and edges $\mathcal{E}$}
 \textcolor{orange}{\textbf{parallel}} \ForEach{$\mathrm{v_0} \rightarrow \mathrm{v}$ : $\mathcal{E}$}{
     $\mathrm{T_0}$ = the thread executing this loop iteration\Comment*[r]{$\mathrm{T_0}$ maintains $\Pi_{\mathrm{T_0}}$ and $\mathrm{Blk}_{\mathrm{T_0}}$}
     $\Pi_{\mathrm{T_0}} = \mathrm{v_0}$; $\mathrm{Blk}_{\mathrm{T_0}} = \varnothing$\Comment*[r]{Operations on $\Pi$ and $\mathrm{Blk}$ are thread-safe}
     $\mathrm{E}$ = FGRT\_DFS($\mathrm{v}$, $\mathrm{v_0}$, $\mathrm{Blk}_{\mathrm{T_0}}$, $\mathrm{Vis}  = \varnothing$)\;
     \lIf(\Comment*[f]{Create a task}){$\mathrm{E} \neq \varnothing$}{\textcolor{orange}{\textbf{spawn}} FGRT\_task($\mathrm{v}$, $\mathrm{v_0}$, $\mathrm{E}$, $1$, $\mathrm{T_0}$)}
}
\textcolor{orange}{\textbf{wait}} for all spawned tasks\;

\caption{FGRT $\left(\mathcal{G}(\mathcal{V}, \mathcal{E})\right)$}
\label{algo:fgrt}
\end{algorithm}

%% file: sections/fgrt-theory.tex
\subsection{Theoretical analysis}
\label{sect:theoryRT}
We now show that the fine-grained parallel Read-Tarjan algorithm is both work-efficient and 
 strongly scalable.

\input{theorems/fgRT_workEff}

The work-efficiency of our fine-grained parallel Read-Tarjan algorithm can be demonstrated using the example given in \fig{fig:fgj_example}a.
In this example, the threads of this algorithm independently explore four different path extensions $E_{i} = v_{1} \rightarrow u_{i} \rightarrow v_2 \rightarrow v_0$, with $i\in \{1\ldots4\}$.
A thread exploring a path extension $E_{i}$ invokes a DFS from $v_2$, which explores vertices $b_1,\ldots,b_m$ at most once and fails to find any other path extension.
Therefore, the amount of work the fine-grained parallel Read-Tarjan algorithm performs does not increase compared to its single-threaded execution.

\input{theorems/fgRT_TinfLemma}

The worst-case depth of our algorithm can be observed when this algorithm is executed on the graph given in \fig{fig:wc_example}a. 
This graph has $c = 2^{n-2}$ cycles and the length of its longest cycle $v_0 \rightarrow \ldots v_{n-1} \rightarrow v_0$ is $n$.
The algorithm creates a task for each vertex of the cycle and performs a successful DFS in each such call, which leads to $T_{\infty} \in O(ne)$.

\input{theorems/fgRT_strongScaling}

As shown in Table~\ref{tab:workDepth}, our fine-grained parallel Read-Tarjan algorithm has a higher depth than our fine-grained parallel Johnson algorithm, introduced in Section~\ref{sect:tpJohnson}.
Nevertheless, the former algorithm is strongly scalable when $c$ grows superlinearly with $n$, whereas strong scalability cannot be guaranteed for the latter algorithm.

\subsection{Summary}

The work of our fine-grained parallel Read-Tarjan algorithm does not increase after fine-grained parallelisation.
This parallel algorithm performs $W_p (n) \in O(n+e+ec)$ work: the same as the work performed by its serial version.
Our optimisations introduced in Section~\ref{sect:impRT} do not reduce the work $W_p (n)$ performed by our parallel algorithm in the worst case.
However, these optimisations significantly improve its performance in practice (see Section~\ref{sect:exp_fgrt_imp}).
In addition, the synchronisation overheads of the fine-grained parallel Read-Tarjan algorithm are not as significant as those of the fine-grained Johnson algorithm because of its shorter critical sections.
Furthermore, this algorithm is the only asymptotically-optimal parallel algorithm for cycle enumeration for which we are able to prove strong scalability.

%% file: theorems/fgRT_workEff.tex
\begin{theorem}
The fine-grained parallel Read-Tarjan algorithm is work-efficient.
\vspace{-.03in}
\label{theorem:fgRT_workEfficiency}
\end{theorem}

\begin{proof}
Because each task of our fine-grained parallel Read-Tarjan algorithm either discovers a cycle or creates at least two child tasks, our algorithm is executed using $O(c)$ tasks.
Each task performs several unsuccessful DFS invocations and one successful DFS per each child task it creates.
All unsuccessful DFS invocations explore at most $e$ edges in total because they share the same set of blocked vertices.
In the worst case, each edge is visited twice per task, once by a successful DFS and once by one of the unsuccessful DFS invocations.
Thus, this algorithm performs $O(e)$ work per task.
Because this algorithm performs $O\left(n+e\right)$ work if there are no cycles in the graph, the total amount of work this algorithm performs is $W_p(n) = O\left(n+e+ec\right)$.
Hence, this algorithm is work-efficient based on Definition~\ref{def:workEfficiency}.
\vspace{-.07in}
\end{proof}

%% file: theorems/fgRT_TinfLemma.tex
\begin{lemma}
The depth $T_{\infty}(n)$ of the fine-grained parallel Read-Tarjan algorithm is in $O(ne)$.
\vspace{-.07in}
\label{lemma:fgRT_tinf}
\end{lemma}

\begin{proof}
In the worst case, a thread executing this algorithm creates a task for each vertex of its longest simple cycle, which has a length of at most $n$.
Before invoking its first child task, a task executes a sequence of unsuccessful DFS invocations in $O(e)$ and a successful DFS invocation also in $O(e)$.
Thus, the depth of this algorithm is $O\left(ne\right)$.
\vspace{-.05in}
\end{proof}

%% file: theorems/fgRT_strongScaling.tex
\begin{theorem}
The fine-grained parallel Read-Tarjan algorithm is strongly scalable when $\lim\limits_{n \to \infty} c/n = \infty$.
\vspace{-.07in}
\label{theorem:fgRT_scalability}
\end{theorem}

\begin{proof}
Because the fine-grained parallel Read-Tarjan algorithm is work-efficient, we can apply Brent's rule~\cite{brent_parallel_1974}:
\begin{equation}
\dfrac{T_1(n)}{p} \leq T_p(n) \leq \dfrac{T_1(n)}{p} + T_{\infty}(n).
\end{equation}
Substituting $T_1(n)$ with $O(n+e+ec)$ and $T_{\infty}(n)$ with $O(ne)$ (see Lemma~\ref{lemma:fgRT_tinf}), for a positive constant $C_0$, it holds that
\begin{equation}
\slfrac{1}{\left(\dfrac{1}{p} + \mathit{C_0}\dfrac{n}{c}\right)} < \slfrac{1}{\left(\dfrac{1}{p} + \dfrac{T_{\infty}(n)}{T_1(n)}\right)} \leq \dfrac{T_1(n)}{T_{p}(n)} \leq p.
\end{equation}
Given that $\lim\limits_{n \to \infty} c/n = \infty$, there exist $n_0>0, C_1>0$ such that if $n > n_0$, then $c/n > C_1 p$.
Thus, for every $n > n_0$, it holds that $k p \leq \tfrac{T_1(n)}{T_{p}(n)} \leq p$, where $k = C_1/(C_0+C_1) < 1$.
As a result, $\tfrac{T_1(n)}{T_{p}(n)}  = \Theta(p)$, which, based on Definition~\ref{def:strongScalability}, completes the proof.
\vspace{-.05in}
\end{proof}

%% file: sections/timeLenConstraints.tex
\section{Parallelising constrained cycle search}
\label{sect:lcycle}

This section describes the methods for adapting our parallel algorithms to search for simple cycles under various constraints.
Because state-of-the-art algorithms for temporal and hop-constrained cycle enumeration are extensions of the Johnson algorithm~\cite{kumar_2scent_2018, peng_towards_2019}, our parallelisation approach described in Section~\ref{sect:tpJohnson} is also applicable to these algorithms.
In this section, we describe the changes to the fine-grained parallel Johnson algorithm needed for enumeration of temporal and hop-constrained cycles.
We also introduce modifications to the cycle enumeration algorithms required for finding time-window-constrained cycles.

\subsection{Cycles in a time window}
\label{sect:twCycles}

Cycle enumeration algorithms require minimal modifications to support time-window constraints.
Such constraints restrict the search for simple, temporal, and hop-constrained cycles to those that occur within a time window of a given size $\delta$, as illustrated in \fig{fig:time-window}.
To find time-window-constrained cycles that start with an edge that has timestamp $t_0$, only the edges with timestamps that belong to the time window $\left[t_0 : t_0 + \delta \right]$ are visited.
To avoid reporting the same cycle several times, another edge with the same timestamp $t_0$ is visited only if the source vertex of that edge has an ID that is smaller than the ID of the vertex from which the search for cycles was started.
Overall, imposing time-window constraints reduces the number of cycles discovered, which results in a more tractable problem.

Because temporal graphs may contain parallel edges, several simple cycles that contain the same sequence of vertices may exist in a temporal graph (see \fig{fig:time-window}a).
As a result, such cycles may be discovered simultaneously, which could accelerate the search for cycles in temporal graphs.
For that purpose, we use a method similar to Kumar and Calders~\cite{kumar_2scent_2018}, in which several simple cycles that contain the same sequence of vertices and whose edges belong to the same time window $\left[t_0 : t_0 + \delta \right]$ are grouped into a single cycle bundle and explored together.
Once a cycle bundle is discovered, the cycles that belong to that bundle can be extracted and reported.

A strongly-connected component (SCC) can be used to reduce the number of vertices visited during the search for time-window-constrained cycles.
The search for cycles that start with the edge $\varepsilon$ can be limited to use only the vertices from the SCC that contains $\varepsilon$~\cite{johnson_finding_1975}.
In the case of time-window-constrained cycles, we compute an SCC for $\varepsilon$ using only the edges with timestamps that belong to $\left[t_0 : t_0 + \delta \right]$, where $t_0$ is the timestamp of $\varepsilon$.
Because an SCC can be computed independently for each edge in $O(e)$ time~\cite{fleischer_identifying_2000}, our fine-grained parallel algorithms remain scalable.

\subsection{Temporal cycles}
\label{sect:temporalCycle}

To efficiently enumerate temporal cycles, the 2SCENT algorithm~\cite{kumar_2scent_2018} replaces the set of blocked vertices $\mathit{Blk}$ in the Johnson algorithm with \textit{closing times}.
The closing time $\mathit{ct}$ of a vertex $v$ indicates that the outgoing temporal edges of $v$ with a timestamp greater than or equal to $\mathit{ct}$ cannot participate in a temporal cycle and are therefore blocked.
Increasing the closing time of $v$ to a new value $\mathit{ct}^{\prime}$ unblocks the blocked outgoing edges of $v$ that have timestamps smaller than $\mathit{ct}^{\prime}$.
This operation triggers the recursive unblocking procedure that unblocks the incoming edges of $v$ with a timestamp smaller than the maximal timestamp among the unblocked outgoing edges of $v$.
This process is repeated for every vertex with unblocked outgoing edges.

\input{algorithms/FGJ_copyOnSteal_temp}

Because the backtracking phase of 2SCENT is based on the Johnson algorithm, it can be parallelised using our fine-grained approach described in Section~\ref{sect:tpJohnson}.
For this purpose, we use our copy-on-steal mechanism with recursive unblocking, introduced in Section~\ref{sect:fgj_copyOnSteal}, which enables a thread to maintain its own set of data structures used for recursion tree pruning.
However, this mechanism is not directly applicable in this case because the recursive unblocking procedure of 2SCENT requires the new closing time for a vertex as a parameter in addition to the vertex itself.
For this reason, an additional data structure called $\mathit{PrevLocks}$ is used alongside the current path $\Pi$ that records the closing time that each vertex $v$ had before it was added to $\Pi$.
Copy-on-steal then performs the recursive unblocking procedure for each vertex $v$ removed from $\Pi$ using the original closing time of the vertex $v$ obtained from $\mathit{PrevLocks}$, as shown in Algorithm~\ref{algo:fgj_cos_mod}.
We refer to the resulting algorithm as the \textit{fine-grained parallel temporal Johnson} algorithm.

\input{figures/temp_cos_example}

The aforementioned modification to the copy-on-steal with recursive unblocking approach also enables a thread of our fine-grained parallel algorithm to reuse the edges blocked by another thread.
This behaviour can be observed in the example shown in \fig{fig:temp_cos_example}, where the thread $T_2$ steals the task indicated in \fig{fig:temp_cos_example}b from the thread $T_1$.
Copy-on-steal executed by $T_2$ invokes recursive unblocking that restores the closing time of $v_3$ to its original value of $9$ obtained from $\mathit{PrevLocks}$.
Note that this original closing time of $v_3$ was previously set by $T_1$ while exploring the path $v_0\rightarrow v_1 \rightarrow v_3$.
The recursive unblocking that $T_2$ invokes for $v_3$ unblocks only the edge $v_6 \rightarrow v_3$ because it is the only incoming edge of $v_3$ with a timestamp smaller than the closing time $9$ of the vertex $v_3$.
Without recording the previous closing times, $T_2$ could instead unblock all incoming edges of $v_3$ by invoking recursive unblocking for $v_3$ with a closing time $\infty$, which also unblocks the edges $v_6 \rightarrow v_7$ and $v_7 \rightarrow v_3$.
However, because there is no temporal cycle that contains these two edges and starts with $v_0$, $T_2$ would unnecessarily visit them in this case.
Thus, restoring the closing time of $v_3$ to its original value $9$ prevents $T_2$ from performing this redundant~work.

We also adapt the Read-Tarjan algorithm and its fine-grained and coarse-grained versions to enumerate temporal cycles using closing times. 
The necessary changes to the algorithm are trivial, and we omit discussing them for~brevity.

To reduce the number of vertices visited during the search for temporal cycles, we use a method similar to the SCC-based technique discussed in Section~\ref{sect:twCycles}.
Instead of computing an SCC for each edge $\varepsilon$, we compute a \textit{cycle-union} that represents an intersection of temporal ancestors and temporal descendants of $\varepsilon$.
The temporal descendants and the temporal ancestors of $\varepsilon$ are the vertices that belong to the temporal paths in which $\varepsilon$ is the first edge and the last edge, respectively.
Defined as such, a cycle-union contains only the vertices that participate in temporal cycles that have $\varepsilon$ as their starting edge.
Thus, the search for temporal cycles that start with $\varepsilon$ can be limited to only those vertices.

\subsection{Hop-constrained cycles}
\label{sect:lcCycles}

An efficient algorithm for enumerating hop-constrained cycles and paths, called BC-DFS~\cite{peng_towards_2019}, replaces the set of blocked vertices $\mathit{Blk}$ in the Johnson algorithm with \emph{barriers}.
A barrier value $\mathit{bar}$ of a vertex $v$ indicates that the starting vertex $v_0$ of a cycle cannot be reached within $\mathit{bar}$ hops from $v$.
As a result, $v$ is blocked if the length of the current path $\Pi$ when the algorithm attempts to visit $v$ is greater than or equal to $L - \mathit{bar}$, where $L$ is the hop constraint.
BC-DFS modifies the recursive unblocking of the Johnson algorithm to reduce the barrier $\mathit{bar}$ of $v$ to a specified value $\mathit{bar}^{\prime} < \mathit{bar}$.
This procedure also sets the barrier of any vertex $u$ that can reach $v$ in $k$ hops to $\mathit{bar}^{\prime}+k$ if the previous barrier of $u$ was greater than $\mathit{bar}^{\prime}+k$.
Maintaining barriers in such a way minimises redundant vertex visits when searching for hop-constrained cycles.

\input{figures/hc-cycle_cos_example}

To parallelise BC-DFS in a fine-grained manner, we use the same technique as that used for fine-grained parallelisation of the Johnson algorithm (Section~\ref{sect:tpJohnson}) and the 2SCENT algorithm (Section~\ref{sect:temporalCycle}).
In this case, threads exploring a recursion tree of BC-DFS maintain separate data structures, such as the current path $\Pi$ and barrier values for each vertex, and use the copy-on-steal with the recursive unblocking approach to copy these data structures among threads.
Similarly to our algorithm from Section~\ref{sect:temporalCycle}, each thread also maintains a data structure $\mathit{PrevLocks}$ that records the original barrier value of each vertex $v$ from $\Pi$.
When a thread steals a task, it performs a recursive unblocking procedure for each vertex $v$ removed from $\Pi$ using its original barrier value obtained from $\mathit{PrevLocks}$, as shown in Algorithm~\ref{algo:fgj_cos_mod}.
This procedure reduces the barrier value of the vertices that can reach $v$, enabling the stealing thread to visit those vertices.
We refer to the resulting algorithm as the \textit{fine-grained parallel hop-constrained Johnson} algorithm.

The modified copy-on-steal with recursive unblocking approach given in Algorithm~\ref{algo:fgj_cos_mod} enables a stealing thread of the aforementioned fine-grained parallel algorithm to reuse barriers discovered by other threads.
This behaviour can be observed in the example given in \fig{fig:hc-cycle_cos_example}.
In that example, the thread $T_1$ first visits the vertices $v_2, v_6, v_7, v_8$ and sets the barrier value of each visited vertex to $L - |\Pi| + 1$ (values in red shown in \fig{fig:hc-cycle_cos_example}a) because it was not able to find a cycle of length $L=6$~\cite{peng_towards_2019}.
Here, $|\Pi|$ denotes the length of $\Pi$ at the moment of exploration of each vertex.
When the thread $T_2$ steals the task indicated in \fig{fig:hc-cycle_cos_example}b from $T_1$, the copy-on-steal mechanism executed by $T_2$ performs a recursive unblocking of the vertex $v_1$ using the original barrier value $0$ of $v_1$ obtained from $\mathit{PrevLocks}$.
This recursive unblocking reduces the barrier value of $v_2$ from $4$ to $1$, which enables $T_2$ to find the cycle that contains $v_2$.
The barrier values of the vertices $v_6$, $v_7$, and $v_8$ are not modified, and, thus, the thread $T_2$ avoids visiting these vertices unnecessarily.

\subsection{Summary}
In this section, we described a method to adapt the cycle enumeration algorithms, such as our fine-grained algorithms introduced in Sections~\ref{sect:tpJohnson} and~\ref{sect:tpReadTarjan}, to search for cycles under time window constraints.
In addition, we introduced a modified version of our copy-on-steal with recursive unblocking approach, introduced in Section~\ref{sect:tpJohnson}, that supports fine-grained parallelisation of temporal and hop-constrained cycle enumeration algorithms~\cite{kumar_2scent_2018, peng_towards_2019} derived from the Johnson algorithm.
As a result, our fine-grained parallel algorithms can enumerate cycles under time-window, temporal, and hop~constraints.

%% file: algorithms/FGJ_copyOnSteal_temp.tex
\begin{algorithm}[t]
\SetAlgoLined
\SetKwProg{Fn}{Function}{}{}
\SetKwProg{Proc}{Procedure}{}{}
\SetKwProg{Task}{\textcolor{orange}{Task}}{}{}
\SetKwInOut{InOut}{InOut}
\SetKwComment{Comment}{$\triangleright$\ }{}
\KwIn{$\mathrm{d}$ - the depth of the task executing this function}
\InOut{$\mathrm{T_1}$ - the victim thread\DontPrintSemicolon\;
\hspace*{10mm}$\mathrm{T_2}$ - the stealing thread}
$\mathrm{Mutex}_{\mathrm{T_1}}.\mathrm{lock}()$\;
$\{\Pi_{\mathrm{T_2}}, \mathrm{Blk}_{\mathrm{T_2}} , \mathrm{PrevLocks}_{\mathrm{T_2}}\} = \mathrm{copy}\left(\{\Pi_{\mathrm{T_1}}, \mathrm{Blk}_{\mathrm{T_1}}, \mathrm{PrevLocks_{\mathrm{T_1}}}\}\right)$\Comment*[r]{$\mathrm{Blk}$ contains closing times or barriers}
$\{\mathrm{Blist}_{\mathrm{T_2}}\} = \mathrm{copy}\left(\{\mathrm{Blist}_{\mathrm{T_1}}\}\right)$\Comment*[r]{$\mathrm{Blist}$ is not used for hop-constrained cycles}
$\mathrm{Mutex}_{\mathrm{T_1}}.\mathrm{unlock}()$\;
\While(\Comment*[f]{Copy-on-steal with recursive unblocking}){$\left|\Pi_{\mathrm{T_2}}\right| > \mathrm{d}$}{
    $\mathrm{u} = \Pi_{\mathrm{T_2}}.\mathrm{pop}()$\;
    $\mathrm{lock} = \mathrm{PrevLocks}_{\mathrm{T_2}}.\mathrm{pop}()$\;
    RecursiveUnblock($\mathrm{u}$, $\mathrm{lock}$, $\mathrm{Blk}_{\mathrm{T_2}}$, $\mathrm{Blist}_{\mathrm{T_2}}$)\;
}
\caption{CFGJ\_copyOnSteal($\mathrm{d}$, $\mathrm{T_1}$, $\mathrm{T_2}$)}
\label{algo:fgj_cos_mod}
\end{algorithm}

%% file: figures/temp_cos_example.tex
\begin{figure}[t]
	\centerline{
		\includegraphics[width=0.82\linewidth]{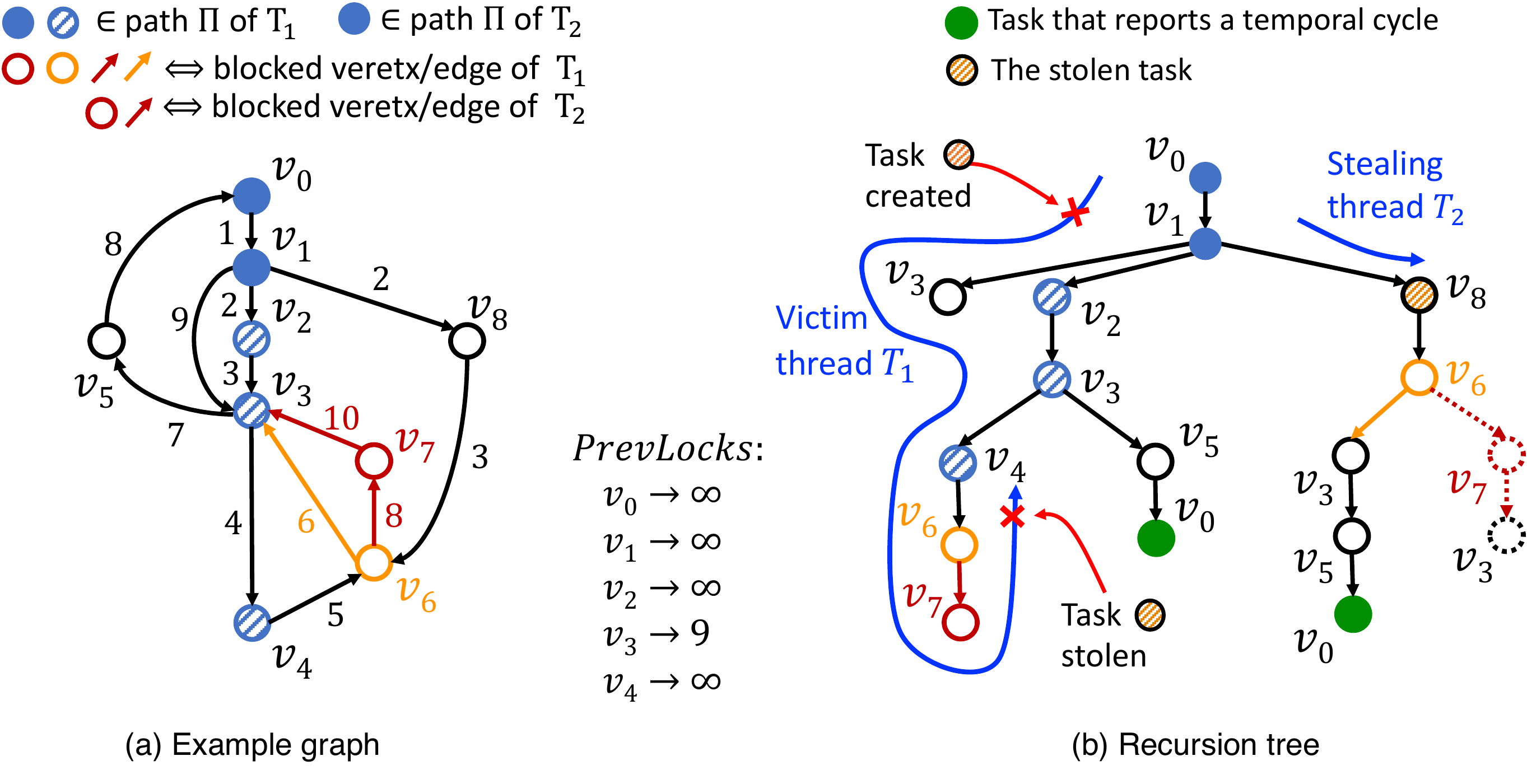}
	}
	\vspace{-.15in}
	\caption{(a) An example graph and (b) the recursion tree of our fine-grained temporal Johnson algorithm when enumerating temporal cycles that start from $v_0$. 
	The thread $T_2$ can avoid the dotted part of the tree by reusing the blocked edges $v_6 \rightarrow v_7$ and $v_7 \rightarrow v_3$ discovered by $T_1$.
	}
	\label{fig:temp_cos_example}
	\vspace{-.15in}
\end{figure}

%% file: figures/hc-cycle_cos_example.tex
\begin{figure}[t]
	\centerline{
		\includegraphics[width=0.78\linewidth]{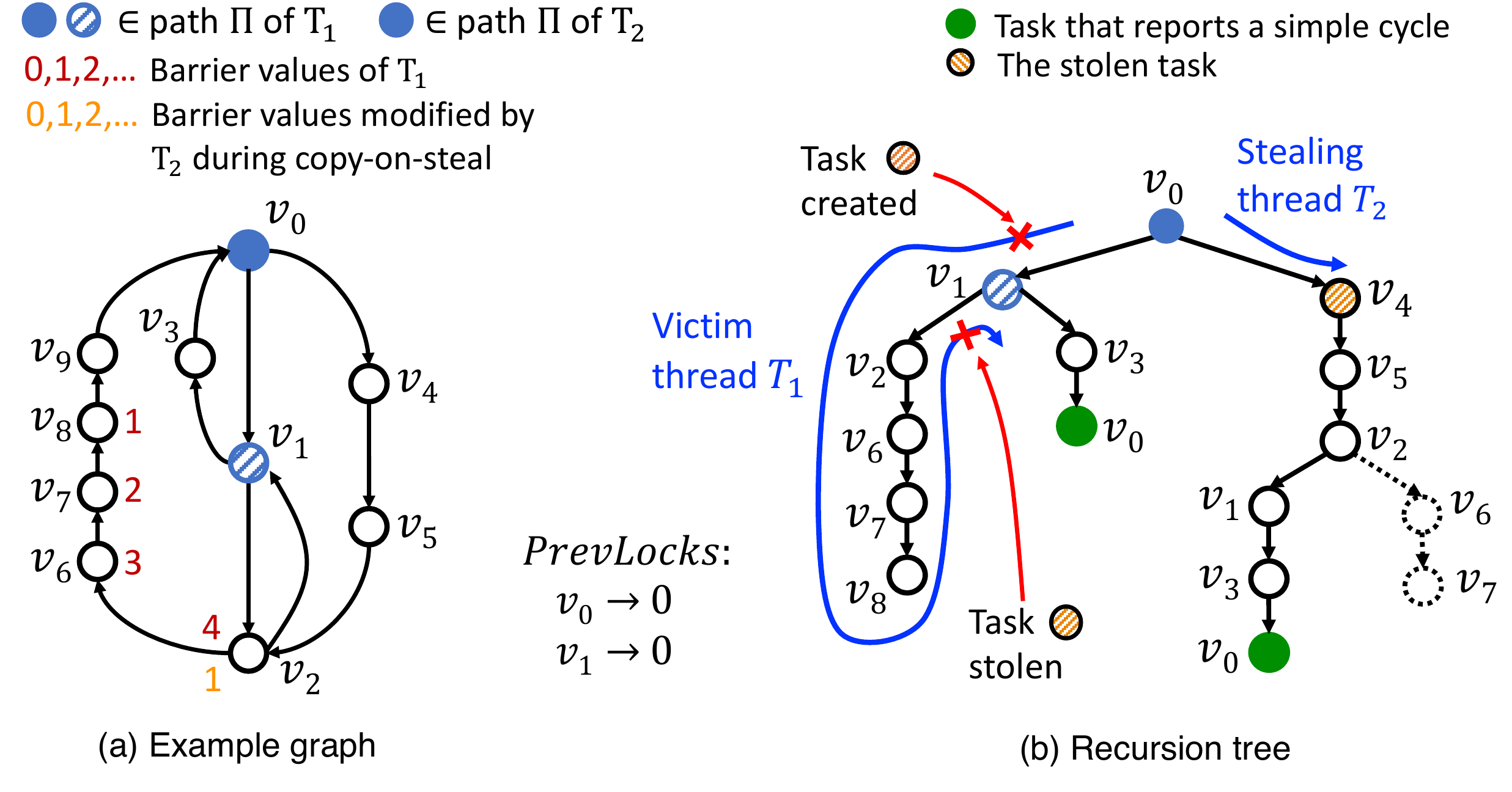}
	}
	\vspace{-.15in}
	\caption{(a) An example graph and (b) the recursion tree of our fine-grained hop-constrained Johnson algorithm when enumerating cycles of length $L=6$ that start from $v_0$.
	Barrier values of unmarked vertices are $0$.
	Copy-on-steal enables the thread $T_2$ to reuse barriers discovered by the thread $T_1$ and to avoid exploring the dotted part of the tree.
	}
	\label{fig:hc-cycle_cos_example}
	\vspace{-.1in}
\end{figure}

%% file: sections/experiments.tex
\section{Experimental evaluation}
\label{sect:experiments}

This section evaluates the performance of our fine-grained parallel algorithms for simple, temporal, and hop-constrained cycle enumeration\footnote{The open-source implementations of our algorithms are maintained here: \url{https://github.com/IBM/parallel-cycle-enumeration}.}.
As Table~\ref{tab:relWork} shows, we are the only ones to offer fine-grained parallel versions of the asymptotically-optimal cycle enumeration algorithms, such as the Johnson and the Read-Tarjan algorithms. 
However, the methods covered in Table~\ref{tab:relWork} can be parallelised using the coarse-grained approach covered in Section~\ref{sect:vertEdgePar}. 
Thus, we use the coarse-grained approach as our main comparison~point.

\input{tables/hwPlatforms}

\input{tables/tempgraphs}

The experiments are performed using two different clusters: Intel\footnote{Intel and Intel Xeon are trademarks or registered trademarks of Intel Corporation or its subsidiaries in the United States and other countries.} KNL~\cite{sodani_knights_2015} and Intel Xeon Skylake~\cite{gcloud_n1}.
The details of these two clusters are given in Table~\ref{tab:platform}.
We developed our code on the Intel KNL cluster and ran most of the analyses there; yet, for completeness, we also provide the comparisons to competing implementations on the Intel Xeon Skylake cluster available in Google Cloud's Compute Engine~\cite{gcloud_n1}.
Scalability experiments are conducted on the Intel KNL cluster. 
In these experiments, the data points that use $64$ threads or less are executed on a single Intel KNL processor; two processors are used to execute the data points that use $128$ threads; and all four processors are used otherwise.
Furthermore, we use more than one thread per core only if the number of threads used is greater than $256$.

We use the \textit{Threading Building Blocks} (TBB)~\cite{kukanov_foundations_2007} library to parallelise the algorithms on a single processor.
We distribute the execution of the algorithms across multiple processors using the Message Passing Interface (MPI)~\cite{mpi_1993}.
When using distributed execution, each processor stores a copy of the input graph in its main memory and searches for cycles starting from a different set of graph edges.
The starting edges are divided among the processors such that when the edges are ordered in the ascending order of their timestamps, $k$ consecutive edges in that order are assigned to $k$ different processors.
Each processor then uses its own dynamic scheduler to balance the workload across its hardware threads.
In this setup, workload imbalance across processors may still occur, but its impact is limited in our experiments because we use at most five processors.

We perform the experiments using the graphs listed in Table~\ref{tab:dataset}.
The TR, FR, and MS graphs are from \emph{Harvard Dataverse}~\cite{jankowski_spreading_2017}, the NL graph is from \emph{Konect}~\cite{kunegis_konect_2013}, the AML graph is from the \emph{AML-Data} repository~\cite{amldata}, and the rest are from \emph{SNAP}~\cite{snapnets}.
Except for BA and BO, all of the graphs have parallel edges, as shown in Table~\ref{tab:dataset}.
To make cycle enumeration problems tractable, we use time-window constraints in all of our experiments.
The time window sizes used in our experiments are given in the figures next to the graph names.
We stop the execution of an algorithm if it takes more than $24h$ on the Intel KNL cluster or more than $6h$ on the Intel Xeon Skylake cluster.

\subsection{Temporal cycle enumeration}

\input{figures/temporal-cycle_comparison}

\input{figures/tw-effect}

The goal of a temporal cycle enumeration problem is to find all simple cycles with edges ordered in time. 
Here, we evaluate the performance of our fine-grained parallel algorithms for this problem introduced in Section~\ref{sect:temporalCycle}. 
Our main comparison points are the coarse-grained parallel versions of the temporal Johnson and temporal Read-Tarjan algorithms.
We refer to the backtracking phase of the state-of-the-art 2SCENT algorithm~\cite{kumar_2scent_2018} for temporal cycle enumeration as the temporal Johnson algorithm and parallelise it in a coarse-grained manner for the experiments.
We do not parallelise the entire 2SCENT algorithm because the preprocessing phase of 2SCENT is strictly sequential and has a time complexity in the order of the complexity of its backtracking phase.
We also provide direct comparisons with the 2SCENT algorithm.

\fig{fig:temporal-cycle_comparison} shows that our fine-grained parallel algorithms achieve an order of magnitude speedup compared to the coarse-grained algorithms on the Intel KNL cluster.
For the NL graph, this speedup reaches up to $40\times$.
Because the Intel Xeon Skylake cluster contains fewer physical cores than the Intel KNL cluster, the speedup between our fine-grained and the coarse-grained parallel Johnson algorithms is smaller on the former cluster.
As can be observed in \fig{fig:tw-effect}, this speedup increases as we increase the time window size used in the algorithms.
Note that enumerating cycles in longer time windows is more challenging because longer time windows contain a larger number of cycles.

\input{figures/numCycles_temp}
\input{figures/scalfig}

\fig{fig:numcyc_temp} shows the number of temporal cycles enumerated in the experiments shown in \fig{fig:temporal-cycle_comparison} and their frequency distribution for the given cycle length.
The execution time of the cycle enumeration algorithms typically depends on the number of cycles discovered.
However, due to the existence of parallel edges, many cycles may consist of the same sequence of vertices and can be explored simultaneously by grouping such cycles into a cycle bundle~\cite{kumar_2scent_2018}.
For example, in the cases of the CO, TR, and MS graphs, a cycle bundle discovered by our algorithms contains more than 10 M cycles on average.
For this reason, despite discovering several orders of magnitude more cycles in the CO, TR, and MS graphs than in the other graphs, the execution time of our fine-grained algorithms on the CO, TR, and MS graphs is comparable to their execution time on the other graphs.
In addition, in the cases of BA, BO, FR, and NL, where one temporal cycle per cycle bundle is discovered, our fine-grained algorithms are more time-consuming on the NL graph because more cycles are discovered in the NL graph than in the BA, BO, and FR graphs.
Thus, the execution time of our fine-grained algorithms depends more on the number of cycle bundles explored than the number of cycles.

The scalability evaluation of the parallel temporal cycle enumeration algorithms is given in \fig{fig:scalfig}.
We also report the performance of the sequential 2SCENT algorithm in the same figure.
The performance of our fine-grained parallel algorithms improves linearly until $256$ threads, after which it becomes sublinear due to simultaneous multithreading.
As a result, our fine-grained versions of the Johnson and the Read-Tarjan algorithms reach $435\times$ and $470\times$ speedups, respectively, compared to their serial versions. 
Additionally, when using $1024$ threads, our fine-grained Johnson algorithm is on average $260\times$ faster than 2SCENT when 2SCENT completes in $24$ hours.
On the other hand, the coarse-grained Johnson algorithm does not scale as well as the fine-grained algorithms.
As a result, the performance gap between the fine-grained and the coarse-grained algorithms increases as we increase the number of threads.

Overall, the fastest algorithm for temporal cycle enumeration that we tested is our fine-grained Johnson algorithm, which is, on average, $60\%$ faster than our fine-grained Read-Tarjan algorithm.
When using $1024$ threads, both fine-grained algorithms are an order of magnitude faster than their coarse-grained counterparts.
Moreover, our fine-grained parallel algorithms, executed on the Intel KNL cluster using 1024 threads, are two orders of magnitude faster than the state-of-the-art algorithm 2SCENT~\cite{kumar_2scent_2018}.

\subsection{Hop-constrained cycle enumeration}

In hop-constrained cycle enumeration, we search for all simple cycles in a graph that are shorter than the specified hop constraint.
Here, we compare our fine-grained parallel hop-constrained Johnson algorithm, introduced in Section~\ref{sect:lcCycles}, with the state-of-the-art algorithms BC-DFS and JOIN~\cite{peng_towards_2019} for this problem.
For this evaluation, we parallelised BC-DFS and JOIN in the coarse-grained manner.
Because adapting the Read-Tarjan algorithm to enumerate hop-constrained cycles is not trivial, we do not report the performance of the fine-grained and coarse-grained versions of this algorithm.
We also omit the performance results for the MS graph because our fine-grained algorithm did not finish under $12h$ when using the smallest time window size.

\fig{fig:hc-cycle_comparison} shows that our fine-grained parallel algorithm is, on average, more than $10\times$ faster than the coarse-grained parallel BC-DFS algorithm for the two largest hop constraints tested.
When using the hop-constraint that is less than or equal to ten, the coarse-grained parallelisation approach is able to achieve workload balance across cores, and thus the performance of this approach is similar to that of our fine-grained approach in this case.
As we increase the hop constraint, the probability of encountering deeper recursion trees also increases.
Exploring such trees using the coarse-grained approach leads to workload imbalance (see Section~\ref{sect:vertEdgePar}).
Our fine-grained algorithm is designed to resolve this problem by exploring a recursion tree using several threads.
Therefore, increasing the hop constraint increases the speedup of our fine-grained algorithm with respect to the coarse-grained algorithm.

According to \fig{fig:numcyc_hc}, the number of cycles increases exponentially as the hop constraint is increased.
Thus, increasing the hop constraint could lead to an exponential increase in the execution time of our fine-grained parallel algorithm for hop-constrained cycle enumeration, which can be observed in \fig{fig:hc-cycle_comparison}.
Note that Figs.~\ref{fig:numcyc_temp} and~\ref{fig:numcyc_simple} indicate that the frequency distributions of the cycles have a bell shape.
As a result, the increase in the number of cycles with increasing hop constraints shown in \fig{fig:numcyc_hc} may not be exponential when the hop constraint is increased beyond $20$.

\input{figures/hc-cycle_comparison}

\input{figures/numCycles_hc-cycle}

\input{figures/scaling_hc}

When the hop constraint is set to $20$, our fine-grained parallel algorithm is, on average, $10\times$ faster than the coarse-grained parallel JOIN algorithm, as shown in \fig{fig:hc-cycle_comparison}.
Although the latter algorithm can be competitive with our fine-grained algorithm, it can also suffer from long execution times, such as in the cases of the AU, NL, and AML graphs.
The reason for these long execution times is the fact that the JOIN algorithm might temporarily construct many non-simple cycles while searching for simple cycles.
Because this algorithm constructs cycles by combining simple paths, it is not guaranteed that each combination results in a simple cycle.
The overhead of combining paths can dominate the execution time of JOIN if this algorithm constructs orders of magnitude more non-simple cycles than simple cycles.
For instance, this situation occurs in the case of AU and hop constraint of $20$, where JOIN discovers $600\times$ more non-simple cycles than simple cycles.
As a result, the speedup of our fine-grained algorithm compared to the coarse-grained JOIN algorithm can reach up to three orders of magnitude.

\fig{fig:scalfig_hc} shows that the speedup of our fine-grained parallel Johnson algorithm with respect to the coarse-grained parallel BC-DFS can be increased by using more threads.
The performance of our fine-grained parallel algorithm scales linearly with the number of threads, whereas the scaling of the coarse-grained parallel BC-DFS eventually slows down.
Thus, in addition to being, on average, an order of magnitude faster than the coarse-grained parallel BC-DFS, our fine-grained algorithm is also more scalable.

\subsection{Simple cycle enumeration}

Here, we evaluate our fine-grained parallel algorithms for simple cycle enumeration.
The computational complexity of simple cycle enumeration is higher than the complexity of temporal and hop-constrained cycle enumeration because simple cycle enumeration does not impose temporal ordering or hop constraints. 
The only constraint we impose is the time-window constraint.
Because the complexity of enumerating simple cycles is higher, we use smaller time windows compared to the cases of temporal and hop-constrained cycle enumeration.
We use the coarse-grained parallel versions of the Johnson and the Read-Tarjan algorithms as our main comparison points.
We do not report the results for the MS graph because our algorithms did not finish in $12h$ even if we set the time window to one second.
We also provide a comparison with the Tiernan algorithm~\cite{tiernan_efficient_1970} parallelised in a fine-grained manner, which is a more efficient version of the previous algorithm by Qing et al.~\cite{nah_efficient_2020} (see Table~\ref{tab:relWork} and Section~\ref{sect:related_work}).
We parallelise the Tiernan algorithm in a fine-grained manner by wrapping each recursive call of this algorithm into a task and by using a dynamic task scheduler to balance the workload across the threads.
Note that the algorithm by Qing et al.~\cite{nah_efficient_2020} uses a static load balancing mechanism, which makes it less efficient than our fine-grained parallelisation of the Tiernan algorithm.

\input{figures/simple-cycle_comparison}
\input{figures/scalfig_sc}

As we can see in \fig{fig:simple-cycle_comparison}, our fine-grained parallel algorithms show an order of magnitude average speedup compared to coarse-grained parallel algorithms on two different platforms.
The reason for this speedup is better scalability of our fine-grained algorithms, which we demonstrate in \fig{fig:scalfig_sc}.
Similarly to the cases of temporal and hop-constrained cycle enumeration (see Figs.~\ref{fig:scalfig} and~\ref{fig:scalfig_hc}), our fine-grained parallel algorithms scale linearly with the number of physical cores used whereas the coarse-grained parallel Johnson algorithm does not scale as well.
Thus, the speedup between the fine-grained and the coarse-grained algorithms increases by utilising more threads.

\fig{fig:numcyc_simple} shows the number of simple cycles enumerated in the experiments shown in \fig{fig:simple-cycle_comparison} and their frequency distribution for the given cycle length.
Similarly to temporal cycle enumeration, the execution time mainly depends on the number of cycle bundles explored rather than on the number of cycles enumerated.
For example, each cycle bundle explored in the BA, BO, SU, FR, and NL graphs contains only two or fewer simple cycles on average, and the execution time of our fine-grained Johnson algorithm is the longest for the NL graph, which also has the most reported cycles (see \fig{fig:numcyc_simple}).
Furthermore, despite BA and BO having similar sizes, the execution time of our fine-grained Johnson algorithm is an order of magnitude longer for BO than for BA.
The reason for this difference is that more cycles were discovered in the BO graph than in the BA graph for the time window sizes given in \fig{fig:simple-cycle_comparison}.
As a result, the execution time of the simple cycle enumeration can significantly vary, even for graphs of similar sizes.

\fig{fig:simple-cycle_comparison_tiernan} presents the comparison of our fine-grained parallel Johnson and Read-Tarjan algorithms with the Tiernan algorithm~\cite{tiernan_efficient_1970} parallelised in a fine-grained manner.
Our fine-grained parallel Johnson algorithm is up to $7\times$ faster than the fine-grained parallel Tiernan algorithm.
The main reason for this performance gap is that the Tiernan algorithm performs more redundant work than the Johnson algorithm (see Section~\ref{section:back_tiernan}).
The fine-grained parallel Read-Tarjan algorithm is slower in the case of NL than the fine-grained parallel Tiernan algorithm because the redundant work performed by the algorithms is limited, and the Tiernan algorithm performs less work per visited edge than the Read-Tarjan algorithm.
However, the fine-grained parallel Read-Tarjan algorithm can be up to $5.3 \times$ faster than the fine-grained parallel Tiernan algorithm for other benchmarks.
Therefore, our fine-grained parallel Johnson and Read-Tarjan algorithms are preferable to the parallel formulation of the Tiernan algorithm, such as the algorithm by Qing et al.~\cite{nah_efficient_2020}.

\input{figures/numCycles_simple}

\input{figures/simple-cycle_comparison_tiernan}

The synchronisation overheads caused by recursive unblocking of our fine-grained parallel Johnson algorithm (see Section~\ref{sect:fgj_copyOnSteal}) are visible only in the case of AML.
In this case, the fine-grained parallel Johnson algorithm performs $60\%$ fewer edge visits than the fine-grained parallel Read-Tarjan; however, it is $25\%$ slower. 
These synchronisation overheads can be explained by a very low cycle-to-vertex ratio.
Because a vertex is blocked if it cannot take part in a cycle, the probability of a vertex being blocked is higher when the cycle-to-vertex ratio is lower (see Table~\ref{tab:dataset} and \fig{fig:numcyc_simple}).
In consequence, more vertices are unblocked during the recursive unblocking of the fine-grained parallel Johnson algorithm, leading to longer critical sections and more contention on the locks. 
Nevertheless, our fine-grained parallel Johnson algorithm achieves a good trade-off between pruning efficiency and lock contention in most cases.

Overall, our fine-grained parallel Johnson and fine-grained parallel Read-Tarjan algorithms have comparable performance, as shown in \fig{fig:simple-cycle_comparison}.
Although the former algorithm is slightly faster, it can suffer from synchronisation overheads in some cases.
Nevertheless, both parallel algorithms achieve linear scaling with the number of physical cores used and achieve, on average, more than $10\times$ speedup with respect to coarse-grained parallel versions of the algorithms.
These conclusions also hold in the cases of temporal and hop-constrained cycle enumeration.

\subsection{Improvements to the Read-Tarjan algorithm}
\label{sect:exp_fgrt_imp}

\fig{fig:rt-opt-exp} shows the effect of our pruning improvements, introduced in Section~\ref{sect:impRT}, on the performance of our fine-grained Read-Tarjan algorithm.
The experiments are performed using a single Intel KNL processor using $256$ threads.
Note that using one processor instead of the entire cluster results in longer execution times, but it enables us to eliminate the effect of workload imbalance across processors in this experiment.
The execution time of the fine-grained parallel Read-Tarjan algorithm decreases after activating each optimisation because fewer redundant vertex and edge visits are performed during the execution of this algorithm.
When all optimisations are enabled, the average speedup of our algorithm for simple cycle enumeration compared to its unoptimised version is $2\times$.
In the case of temporal cycle enumeration, the average speedup increases to $3.4\times$.
As a result, our pruning improvements enable the fine-grained parallel Read-Tarjan algorithm to be competitive with the fine-grained parallel Johnson algorithm.

\input{figures/rt-opt-exp}

%% file: tables/hwPlatforms.tex
\begin{table}[t]
\centering
\caption{Hardware platforms used in the experiments. Here, $\mathrm{P}$, $\mathrm{C/P}$, and $\mathrm{T/C}$ represent the number of processors, the number of cores per processor, and the number of hardware threads per core, respectively.}
\vspace{-.10in}
\addtolength{\tabcolsep}{-1pt}
\begin{tabular}{l|cc}
\textbf{platform}    & Intel KNL~\cite{sodani_knights_2015} & Intel Xeon Skylake~\cite{gcloud_n1}  \\ \hline
\textbf{$\mathbf{P \times C/P \times T/C}$}   & $4\times64\times4$                                  & $5\times48\times2$                           \\
\textbf{Total no. threads}   & $1024$                                   & $480$                           \\
\textbf{Frequency}      & $1.3$ GHz                              & $2$ GHz                       \\
\textbf{Memory per proc.}      & $110$ GB                               & $360$ GB                        \\
\textbf{L1d/L2/L3 cache}   & $32$ KB/$512$ KB/none                     & $32$ KB/$1$ MB/$38.5$ MB    \\ \hline
\end{tabular}
\label{tab:platform}
% \vspace{-.1in}
\end{table}

%% file: tables/tempgraphs.tex
\begin{table}[t]
\centering
\small 
\caption{Temporal graphs used in the experiments.
In this table, $\Delta_{\mathit{avg}}$ and $\Delta_{\mathit{max}}$ refer to the average and maximum values of $\Delta$, respectively, where $\Delta$ is the number of outgoing edges of a vertex, i.e., vertex degree.
Similarly, $P_{\mathit{avg}}$ and $P_{\mathit{max}}$ refer to the average and maximum values of $P$, respectively, where $P$ represents the number of parallel edges for a given source and destination vertex. 
Time span refers to the difference between the maximum and minimum timestamps in a graph.
}
\vspace{-.1in}

\begin{tabular}{ll|ccccccc}
graph           & abbr. & $n$          & $e$       & $\Delta_{\mathit{avg}}$ & $\Delta_{\mathit{max}}$ & $P_{\mathit{avg}}$ & $P_{\mathit{max}}$     &  Time span [days] \\ \hline
bitcoinalpha & BA    & 3.3 k & 24 k  & 7.4        & 490        & 1.0   & 1.0     & 1901    \\
bitcoinotc   & BO    & 4.8 k & 36 k  & 7.4        & 763        & 1.0   & 1.0     & 1903    \\
CollegeMsg   & CO    & 1.3 k & 60 k  & 44.3       & 1.1 k      & 2.9   & 98    & 193     \\
email-Eu-core& EM    & 824   & 332 k & 403.3      & 9.8 k      & 13.3  & 2.8 k & 803     \\
mathoverflow & MO    & 16 k  & 390 k & 23.7       & 4.5 k      & 1.7   & 225   & 2350    \\
transactions & TR    & 83 k  & 530 k & 6.4        & 1.7 k      & 1.5   & 290   & 1803    \\
higgs-activity & HG    & 278 k & 555 k & 2.0      & 655        & 1.2   & 95    & 6.0       \\
askubuntu    & AU    & 102 k & 727 k & 7.1        & 7.6 k      & 1.3   & 154   & 2613    \\
superuser    & SU    & 138 k & 1.1 M & 8.1        & 26 k       & 1.3   & 78    & 2773    \\
wiki-talk    & WT    & 140 k & 6.1 M & 43.7       & 233 k      & 1.9   & 1.1 k & 2277    \\
friends2008  & FR    & 481 k & 12 M  & 25.5       & 9.1 k      & 1.0   & 6.0     & 1826    \\
wiki-dynamic-nl           & NL    & 1.0 M   & 20 M  & 19.5       & 30 k  & 1.5   & 352   & 3602    \\
messages     & MS    & 313 k & 26 M  & 83.4       & 48 k       & 4.3   & 10 k  & 1880    \\
AML-Data     & AML   & 10 M  & 34 M  & 3.4        & 26 k       & 4.6   & 96    & 30      \\
stackoverflow& SO    & 2.0 M & 48 M  & 24.2       & 72 k       & 1.3   & 594   & 2774    \\ \hline
\end{tabular}
\label{tab:dataset}
% \vspace{-.1in}
\end{table}

%% file: figures/temporal-cycle_comparison.tex
\begin{figure*}[t!]
	\centerline{
    \includegraphics[width=1\linewidth]{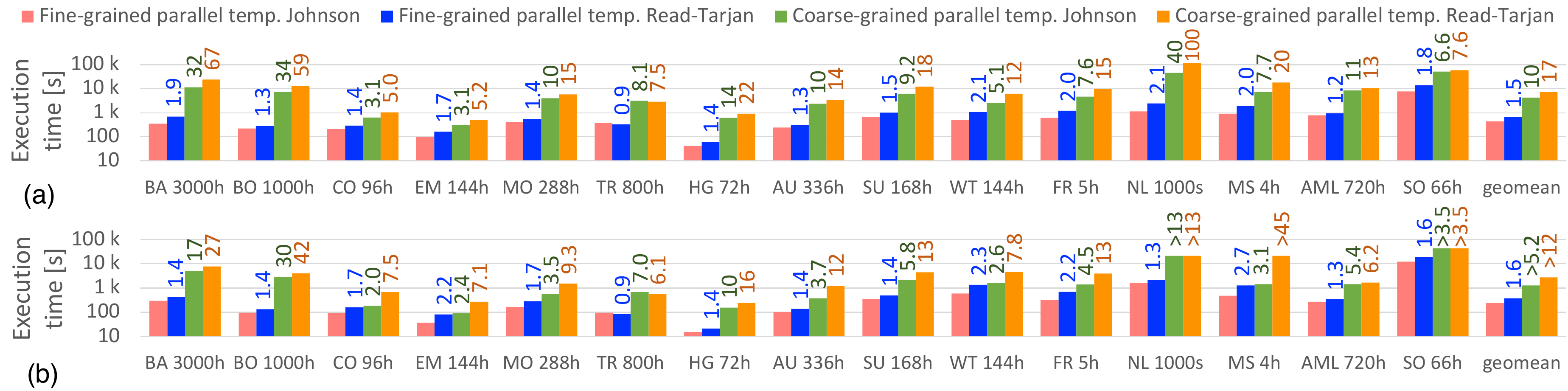}
    }
	\vspace{-.1in}
	\caption{Performance of parallel algorithms for temporal cycle enumeration on (a) the Intel KNL cluster using $1024$ threads and (b) the Intel Xeon Skylake cluster using $480$ threads.
    The values above the bars show the execution time of each algorithm relative to that of our fine-grained parallel temporal Johnson for the same benchmark.
    The values that contain the symbol $>$ represent the experiments that did not finish within the given time limit.
	}
	\vspace{-.1in}
	\label{fig:temporal-cycle_comparison}
\end{figure*}

%% file: figures/tw-effect.tex
\begin{figure*}[t]
	\centerline{
		\includegraphics[width=1\linewidth]{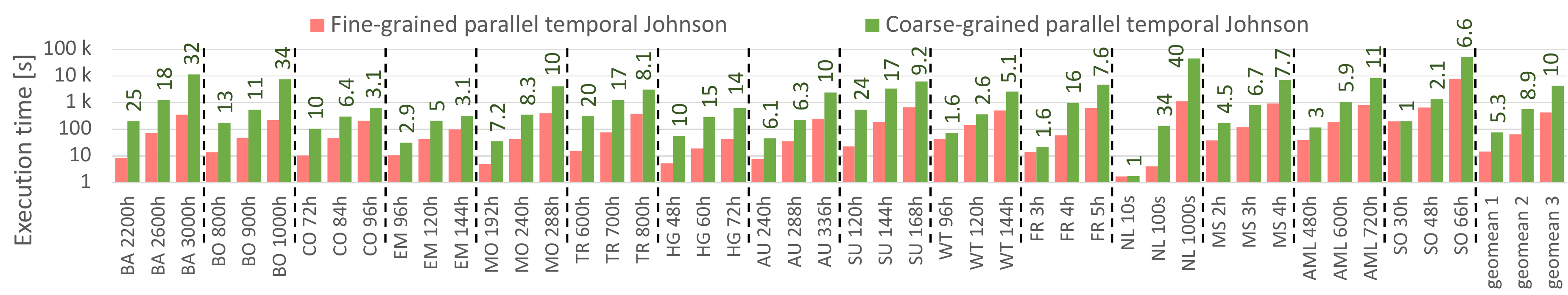}
	}
	\vspace{-.15in}
	\caption{
 Longer time windows increase the performance gap between the algorithms. 
 The algorithms are executed on the Intel KNL cluster using $1024$ threads.
 The numbers above the bars show the execution times of the coarse-grained algorithm relative to that of the fine-grained algorithm.
	}
	\label{fig:tw-effect}
	\vspace{-.1in}
\end{figure*}

%% file: figures/numCycles_temp.tex
\begin{figure*}[t]
	\centerline{
		\includegraphics[width=1\linewidth]{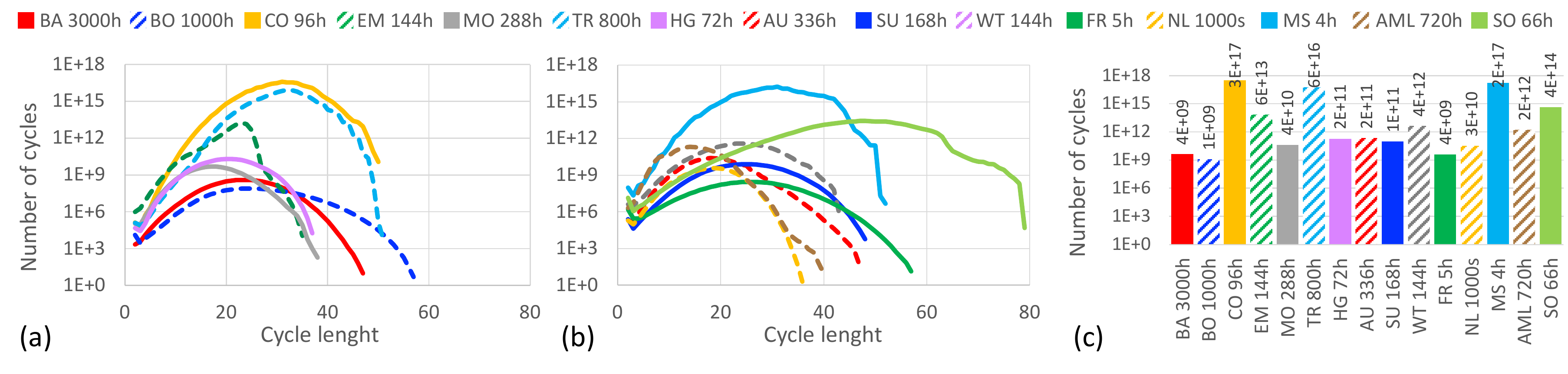}
	}
	\vspace{-.1in}
	\caption{(a), (b) Frequency distribution of temporal cycles for different cycle lengths and (c) the total number of temporal cycles discovered during the experiments shown in \fig{fig:temporal-cycle_comparison}. The number of temporal cycles discovered is several orders of magnitude greater than the number of vertices or edges of a graph.
	}
	\label{fig:numcyc_temp}
	\vspace{-.1in}
\end{figure*}

%% file: figures/scalfig.tex
\begin{figure*}[t]
	\centerline{
		\includegraphics[width=0.92\linewidth]{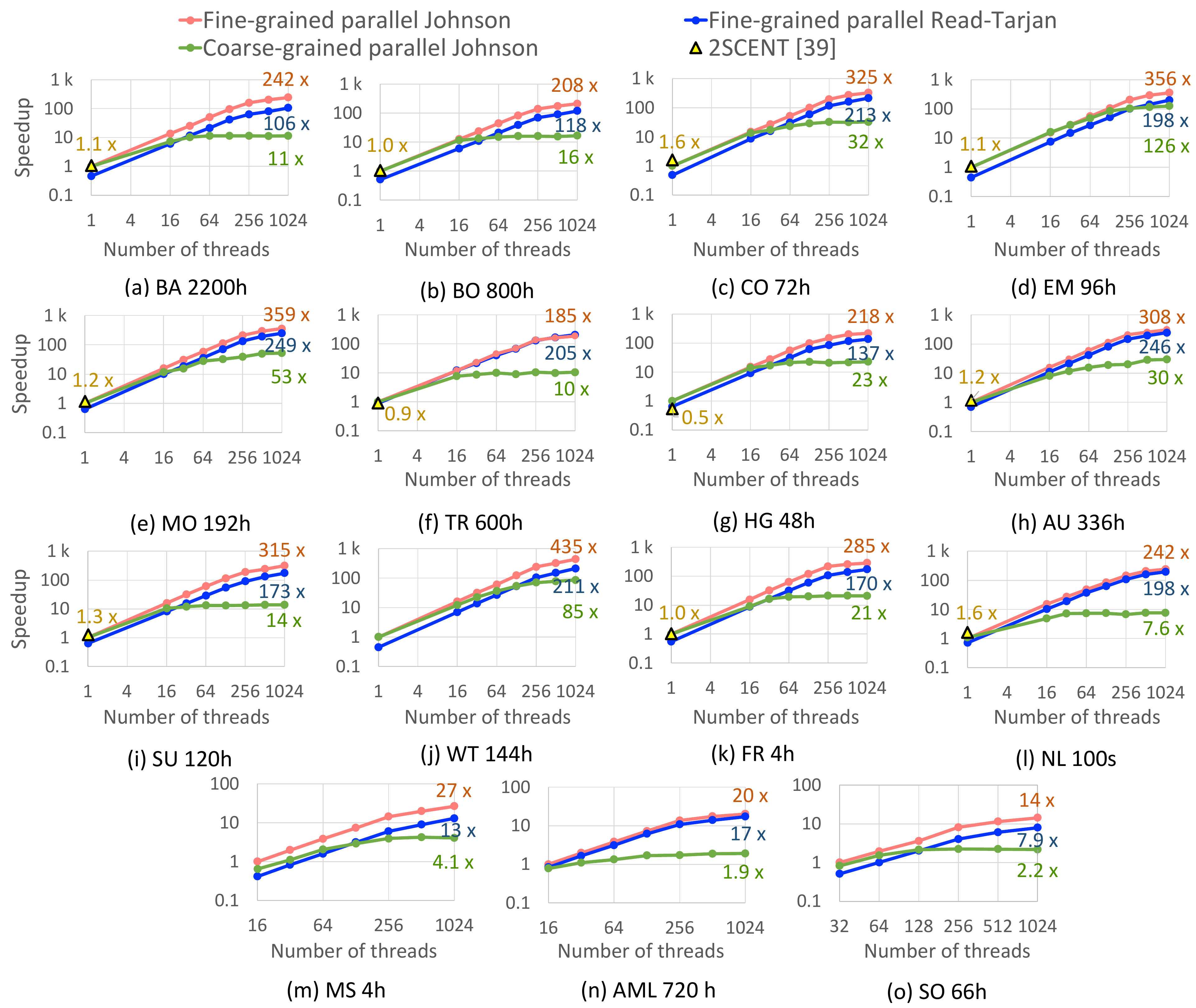}
	}
	\vspace{-.1in}
	\caption{Scalability evaluation of parallel temporal cycle enumeration algorithms executed on the Intel KNL cluster. The baseline is our fine-grained parallel temporal Johnson algorithm.
	The relative performance of 2SCENT~\cite{kumar_2scent_2018} is shown when it completes in $24$ hours. Note that the 2SCENT implementation is single-threaded and the single-threaded execution results are not available for all~graphs.
	}
	\label{fig:scalfig}
	\vspace{-.15in}
\end{figure*}

%% file: figures/hc-cycle_comparison.tex
\begin{figure*}[t!]
       	\centerline{
            \includegraphics[width=1\linewidth]{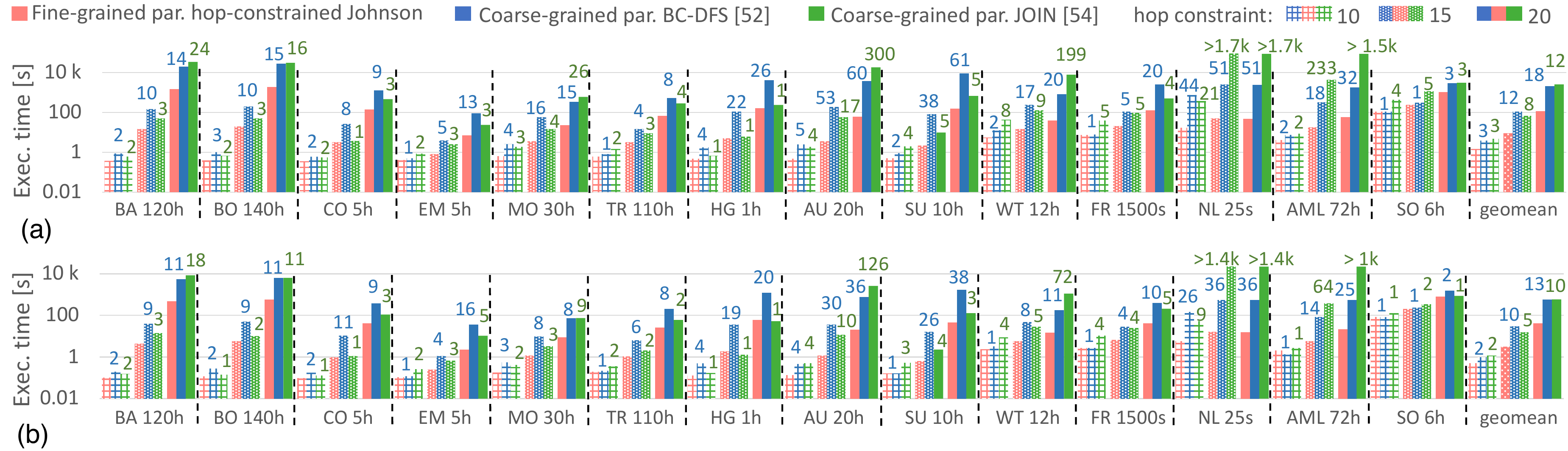}
        }
	\vspace{-.05in}
	\caption{Performance of parallel algorithms for hop-constrained simple cycle enumeration on (a) the Intel KNL cluster using $1024$ threads and (b) the Intel Xeon Skylake cluster using $480$ threads.
    The values above the bars show the execution time of the coarse-grained parallel algorithms relative to that of our fine-grained parallel algorithm.  
    The values that contain the symbol $>$ represent the experiments that did not finish within the given time limit.
    Larger hop constraints increase the performance gap between the two~algorithms.
	}
	\label{fig:hc-cycle_comparison}
	\vspace{-.1in}
\end{figure*}

%% file: figures/numCycles_hc-cycle.tex
\begin{figure*}[t]
	\centerline{
		\includegraphics[width=1\linewidth]{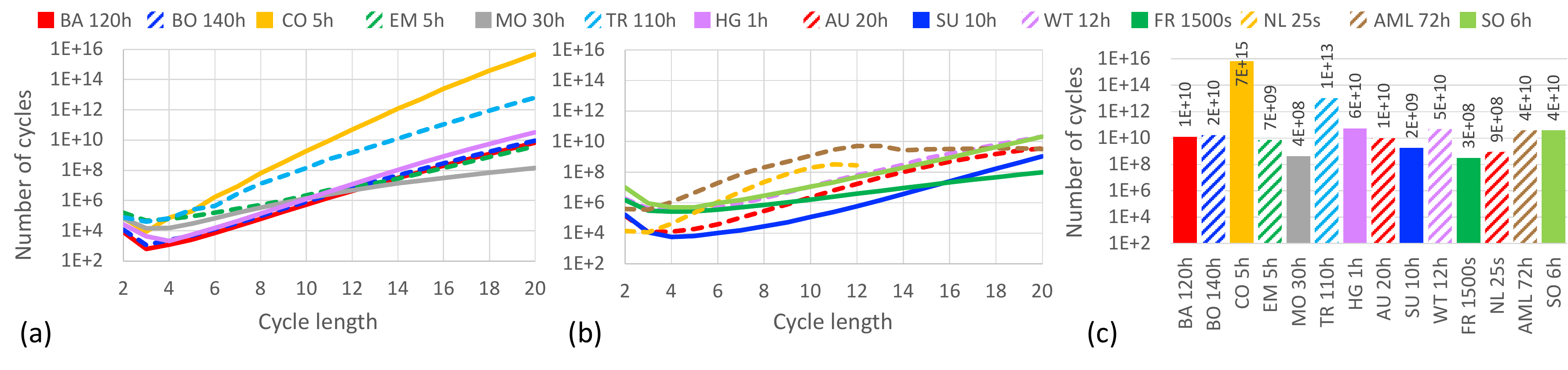}
	}
	\vspace{-.1in}
	\caption{(a), (b) Frequency distribution of hop-constrained cycles for different cycle lengths and (c) the total number of hop-constrained cycles discovered during the experiments shown in \fig{fig:hc-cycle_comparison} using the hop-constraint of $20$. In most cases, the number of cycles increases exponentially with hop-constraint.
	}
	\label{fig:numcyc_hc}
	\vspace{-.1in}
\end{figure*}

%% file: figures/scaling_hc.tex
\begin{figure*}[t]
	\centerline{
		\includegraphics[width=0.87\linewidth]{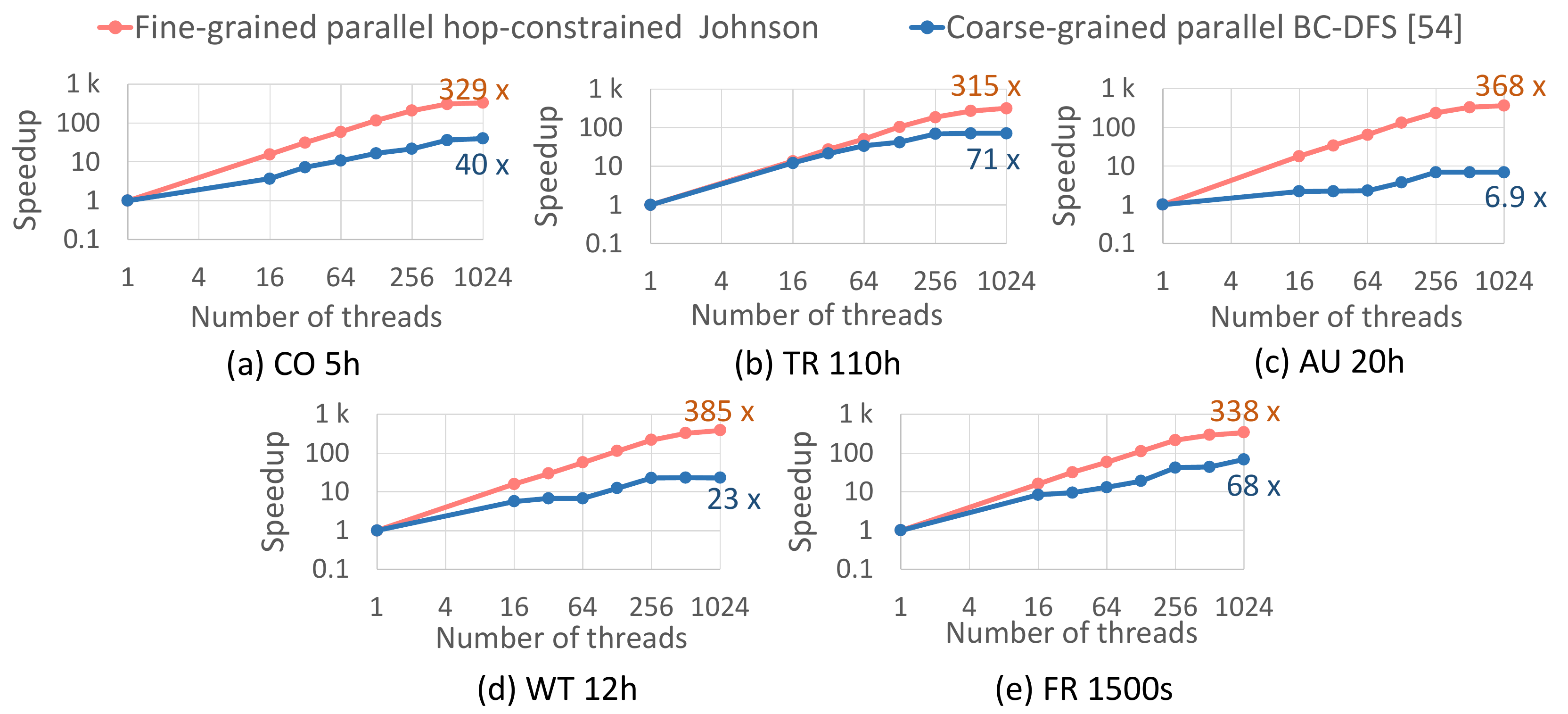}
	}
	\vspace{-.05in}
	\caption{Scalability evaluation of parallel hop-constrained cycle enumeration algorithms executed on the Intel KNL cluster using the hop constraint of $15$. The speedup values are relative to the single-threaded execution of BC-DFS.
 Evaluation on other graphs is omitted for brevity.
	}
	\label{fig:scalfig_hc}
	\vspace{-.1in}
\end{figure*}

%% file: figures/simple-cycle_comparison.tex
\begin{figure*}[t!]
 	\centerline{
            \includegraphics[width=1\linewidth]{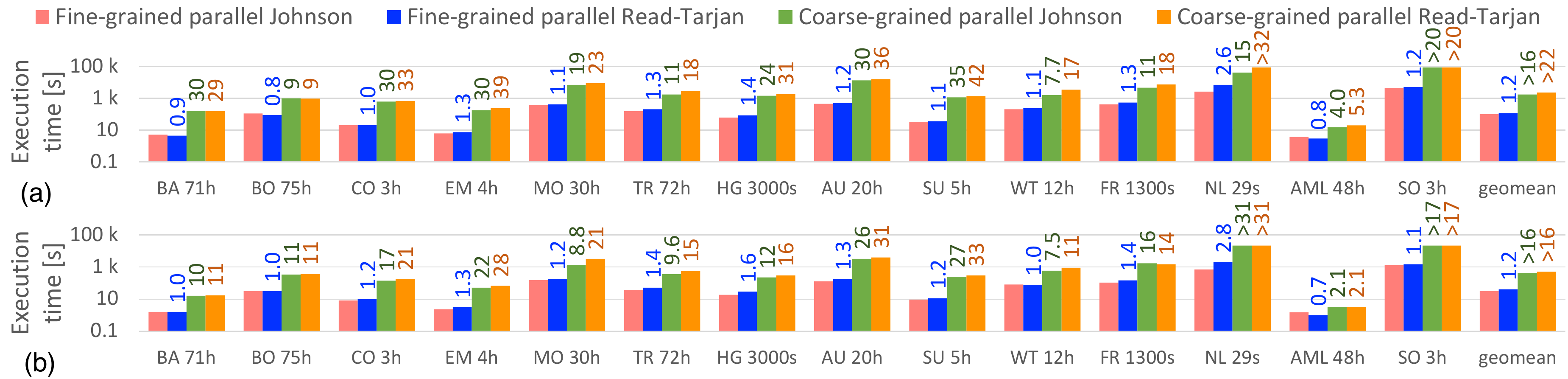}
        }
	\vspace{-.05in}
	\caption{Performance of parallel algorithms for simple cycle enumeration on (a) the Intel KNL cluster using $1024$ threads and (b) the Intel Xeon Skylake cluster using $480$ threads.
    The values above the bars show the execution time of each algorithm relative to that of our fine-grained parallel Johnson algorithm for the same benchmark.
    The values that contain the symbol $>$ represent the experiments that did not finish within the given time limit.
	}
	\vspace{-.1in}
	\label{fig:simple-cycle_comparison}
\end{figure*}

%% file: figures/scalfig_sc.tex
\begin{figure*}[t]
	\centerline{
		\includegraphics[width=0.87\linewidth]{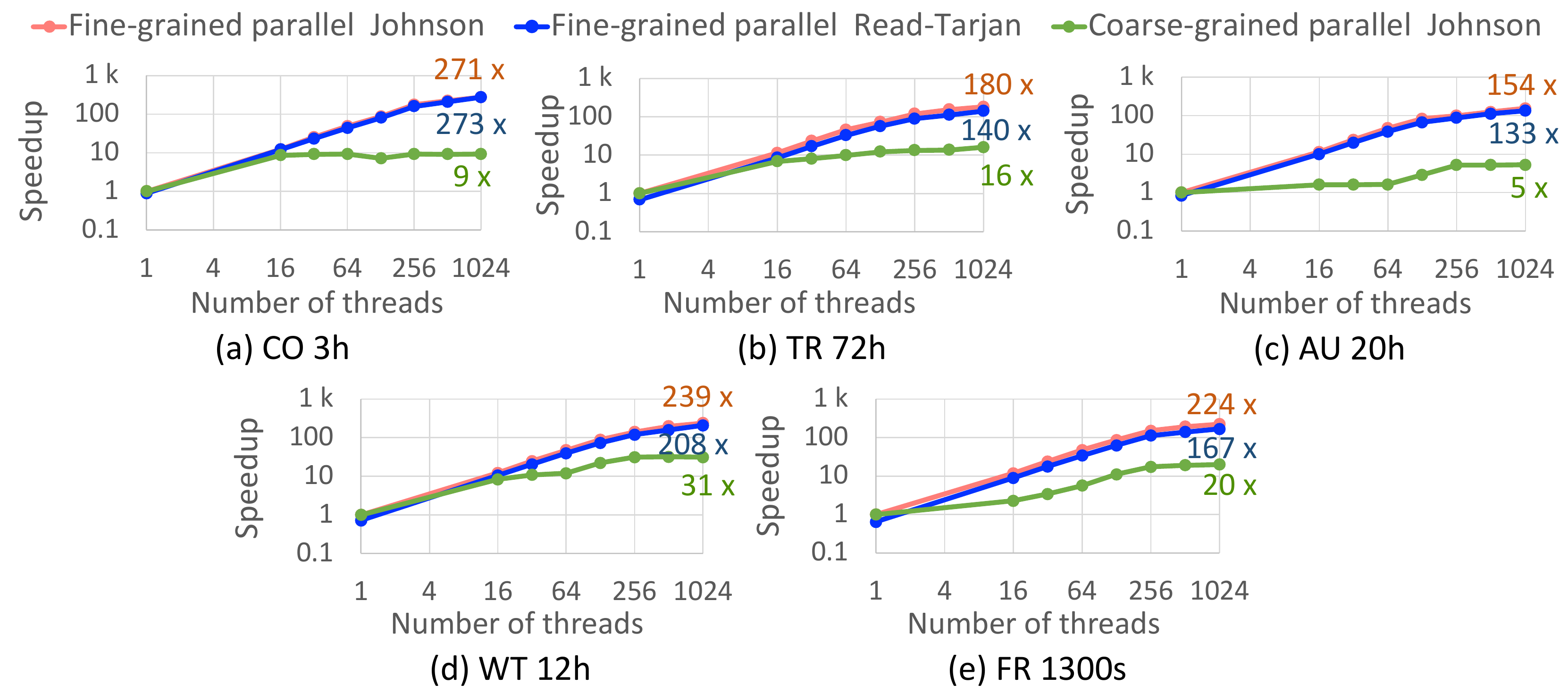}
	}
	\vspace{-.1in}
	\caption{Scalability evaluation of parallel simple cycle enumeration algorithms executed on the Intel KNL cluster. The speedup values are relative to the single-threaded execution of the Johnson algorithm.
 Evaluation on other graphs is omitted for brevity.
	}
	\label{fig:scalfig_sc}
	\vspace{-.1in}
\end{figure*}

%% file: figures/numCycles_simple.tex
\begin{figure*}[t]
	\centerline{
		\includegraphics[width=1\linewidth]{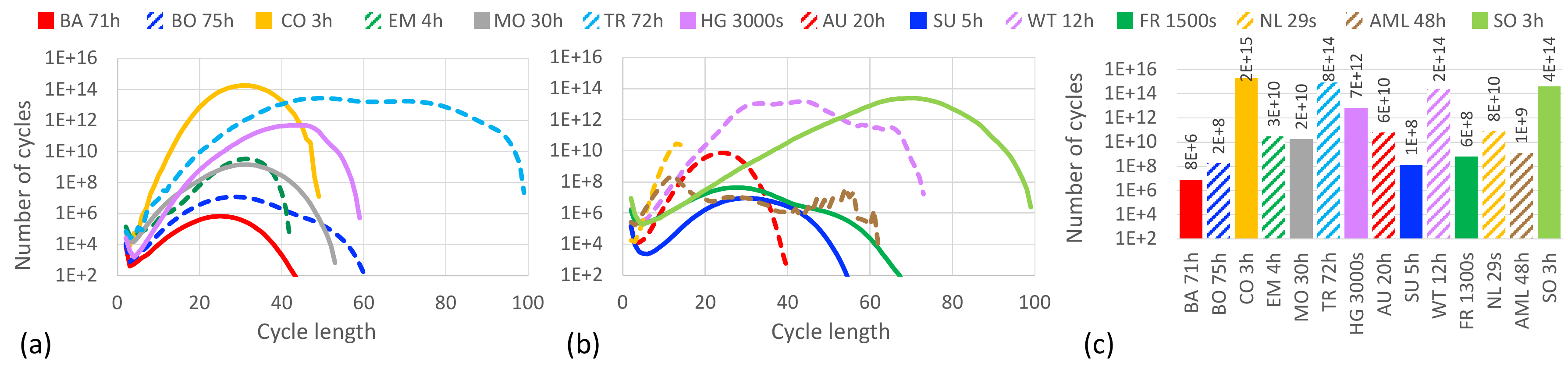}
	}
	\vspace{-.1in}
	\caption{(a), (b) Frequency distribution of simple cycles for different cycle lengths and (c) the total number of simple cycles discovered during the experiments shown in \fig{fig:simple-cycle_comparison}. Simple cycles tend to be longer than temporal cycles, despite using a smaller time window for simple cycle enumeration (see \fig{fig:numcyc_temp}).
	}
	\label{fig:numcyc_simple}
	\vspace{-.1in}
\end{figure*}

%% file: figures/simple-cycle_comparison_tiernan.tex
\begin{figure*}[t!]
 	\centerline{
            \includegraphics[width=0.82\linewidth]{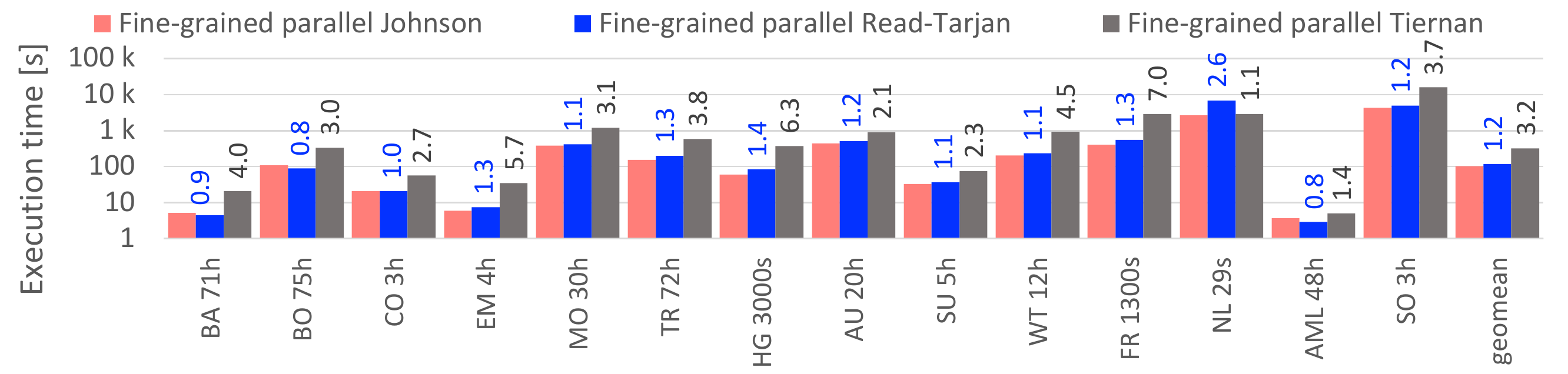}
        }
	\vspace{-.1in}
	\caption{Comparison of our fine-grained parallel algorithms for simple cycle enumeration with the fine-grained parallel Tiernan algorithm~\cite{tiernan_efficient_1970} on the Intel KNL cluster using $1024$ threads.
     The values above the bars show the execution time of each algorithm relative to that of our fine-grained parallel Johnson algorithm for the same benchmark. 
    Our fine-grained parallel Johnson algorithm is up to $7\times$ faster than the fine-grained parallel Tiernan algorithm.
	}
	\vspace{-.1in}
	\label{fig:simple-cycle_comparison_tiernan}
\end{figure*}

%% file: figures/rt-opt-exp.tex
\begin{figure}[t]
	\centerline{
		\includegraphics[width=1\linewidth]{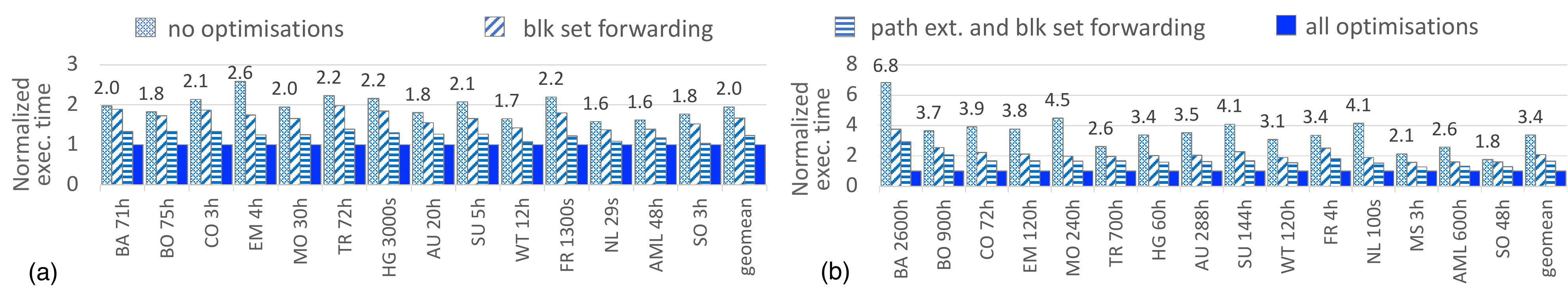}
	}
	\vspace{-.1in}
	\caption{Effect of the pruning improvements to our fine-grained parallel Read-Tarjan algorithm for (a) simple and (b) temporal cycle enumeration.
    Execution times are normalised to the case that includes all optimisations. 
    Our optimisations accelerate this algorithm by up to~$6.8\times$.
	}
	\label{fig:rt-opt-exp}
	\vspace{-.15in}
\end{figure}

%% file: sections/conclusions.tex
\section{Conclusions}
\label{sect:conclusion}

This work has made three contributions to the area of parallel cycle enumeration.
First, we have introduced scalable fine-grained parallel versions of the state-of-the-art Johnson and Read-Tarjan algorithms for enumerating simple cycles.
In particular, we have shown that the novel fine-grained parallel approach we contributed for parallelising the Johnson algorithm can be adapted to support the enumeration of temporal and hop-constrained cycles as well.
Our fine-grained parallel algorithms for enumerating the aforementioned types of cycles achieve a near-linear performance scaling on a compute cluster with a total number of $256$ CPU cores that can execute $1024$ simultaneous software threads.

Secondly, we have shown that our fine-grained parallel cycle enumeration algorithms are scalable both in theory and in practice.
In contrast, their coarse-grained parallel versions do not share this property.
When using $1024$ software threads, our fine-grained parallel algorithms are on average an order of magnitude faster than their coarse-grained counterparts.
In addition, the performance gap between the fine-grained and coarse-grained parallel algorithms widens as we use more physical CPU cores. This performance gap also widens when increasing the time window in the case of temporal cycle enumeration and when increasing the hop constraint in the case of hop-constrained cycle enumeration.

Thirdly, we have shown that, whereas our fine-grained parallel Read-Tarjan algorithm is work efficient, our fine-grained parallel Johnson algorithm is not.
In general, the former is competitive against the latter because of the new pruning methods we introduced, yet the latter outperforms the former in most experiments.
In some rare cases, our fine-grained parallel Johnson algorithm can suffer from synchronisation overheads. 
In such cases, our fine-grained parallel Read-Tarjan algorithm offers a more scalable alternative.

%% file: main.bbl

\begin{thebibliography}{74}

%%% the Web address.

\ifx \showCODEN    \undefined \def \showCODEN     #1{\unskip}     \fi
\ifx \showDOI      \undefined \def \showDOI       #1{#1}\fi
\ifx \showISBNx    \undefined \def \showISBNx     #1{\unskip}     \fi
\ifx \showISBNxiii \undefined \def \showISBNxiii  #1{\unskip}     \fi
\ifx \showISSN     \undefined \def \showISSN      #1{\unskip}     \fi
\ifx \showLCCN     \undefined \def \showLCCN      #1{\unskip}     \fi
\ifx \shownote     \undefined \def \shownote      #1{#1}          \fi
\ifx \showarticletitle \undefined \def \showarticletitle #1{#1}   \fi
\ifx \showURL      \undefined \def \showURL       {\relax}        \fi
% invisible to TeX
\providecommand\bibfield[2]{#2}
\providecommand\bibinfo[2]{#2}
\providecommand\natexlab[1]{#1}
\providecommand\showeprint[2][]{arXiv:#2}

\bibitem[\protect\citeauthoryear{Abdelhamid, Abdelaziz, Kalnis, Khayyat, and
  Jamour}{Abdelhamid et~al\mbox{.}}{2016}]%
        {abdelhamid_scalemine_2016}
\bibfield{author}{\bibinfo{person}{Ehab Abdelhamid}, \bibinfo{person}{Ibrahim
  Abdelaziz}, \bibinfo{person}{Panos Kalnis}, \bibinfo{person}{Zuhair Khayyat},
  {and} \bibinfo{person}{Fuad Jamour}.} \bibinfo{year}{2016}\natexlab{}.
\newblock \showarticletitle{{ScaleMine}: {Scalable} {Parallel} {Frequent}
  {Subgraph} {Mining} in a {Single} {Large} {Graph}}. In
  \bibinfo{booktitle}{\emph{{SC16}: {International} {Conference} for {High}
  {Performance} {Computing}, {Networking}, {Storage} and {Analysis}}}.
  \bibinfo{publisher}{IEEE}, \bibinfo{address}{Salt Lake City, UT, USA},
  \bibinfo{pages}{716--727}.
\newblock
\showISBNx{978-1-4673-8815-3}
\newblock
\shownote{doi: 10.1109/SC.2016.60}.


\bibitem[\protect\citeauthoryear{Agarwal and Ramachandran}{Agarwal and
  Ramachandran}{2016}]%
        {agarwal_finding_2016}
\bibfield{author}{\bibinfo{person}{Udit Agarwal} {and} \bibinfo{person}{Vijaya
  Ramachandran}.} \bibinfo{year}{2016}\natexlab{}.
\newblock \showarticletitle{{Finding k Simple Shortest Paths and Cycles}}. In
  \bibinfo{booktitle}{\emph{27th International Symposium on Algorithms and
  Computation (ISAAC 2016)}} \emph{(\bibinfo{series}{Leibniz International
  Proceedings in Informatics (LIPIcs)}, Vol.~\bibinfo{volume}{64})},
  \bibfield{editor}{\bibinfo{person}{Seok-Hee Hong}} (Ed.).
  \bibinfo{publisher}{Schloss Dagstuhl--Leibniz-Zentrum fuer Informatik},
  \bibinfo{address}{Dagstuhl, Germany}, \bibinfo{pages}{8:1--8:12}.
\newblock
\showISBNx{978-3-95977-026-2}
\showISSN{1868-8969}
\urldef\tempurl%
\url{https://doi.org/10.4230/LIPIcs.ISAAC.2016.8}
\showDOI{\tempurl}


\bibitem[\protect\citeauthoryear{Aggarwal and Wang}{Aggarwal and Wang}{2010}]%
        {aggarwal_managing_2010}
\bibfield{editor}{\bibinfo{person}{Charu~C. Aggarwal} {and}
  \bibinfo{person}{Haixun Wang}} (Eds.). \bibinfo{year}{2010}\natexlab{}.
\newblock \bibinfo{booktitle}{\emph{Managing and {Mining} {Graph} {Data}}}.
  \bibinfo{series}{Advances in {Database} {Systems}},
  Vol.~\bibinfo{volume}{40}.
\newblock \bibinfo{publisher}{Springer US}, \bibinfo{address}{Boston, MA}.
\newblock
\showISBNx{978-1-4419-6044-3 978-1-4419-6045-0}
\newblock
\shownote{doi: 10.1007/978-1-4419-6045-0}.


\bibitem[\protect\citeauthoryear{Aldred and Thomassen}{Aldred and
  Thomassen}{2008}]%
        {aldred_maximum_2008}
\bibfield{author}{\bibinfo{person}{R.~E.~L. Aldred} {and}
  \bibinfo{person}{Carsten Thomassen}.} \bibinfo{year}{2008}\natexlab{}.
\newblock \showarticletitle{On the maximum number of cycles in a planar graph}.
\newblock \bibinfo{journal}{\emph{J. Graph Theory}} \bibinfo{volume}{57},
  \bibinfo{number}{3} (\bibinfo{date}{March} \bibinfo{year}{2008}),
  \bibinfo{pages}{255--264}.
\newblock
\showISSN{03649024, 10970118}
\newblock
\shownote{doi: 10.1002/jgt.20290}.


\bibitem[\protect\citeauthoryear{Altman}{Altman}{2021}]%
        {amldata}
\bibfield{author}{\bibinfo{person}{Erik Altman}.}
  \bibinfo{year}{2021}\natexlab{}.
\newblock \bibinfo{title}{AML-Data}.
\newblock \bibinfo{howpublished}{Available online:
  \url{https://github.com/IBM/AML-Data}}.
\newblock
\newblock
\shownote{Accessed: 2022-05-30}.


\bibitem[\protect\citeauthoryear{Bader}{Bader}{1999}]%
        {Bader99apractical}
\bibfield{author}{\bibinfo{person}{David~A. Bader}.}
  \bibinfo{year}{1999}\natexlab{}.
\newblock \bibinfo{title}{A {Practical} {Parallel} {Algorithm} for {Cycle}
  {Detection} in {Partitioned} {Digraphs}}.
\newblock
  \bibinfo{howpublished}{\url{https://digitalrepository.unm.edu/ece_rpts/45}}.
\newblock


\bibitem[\protect\citeauthoryear{Balakrishnan}{Balakrishnan}{1997}]%
        {Balakrishnan1997}
\bibfield{author}{\bibinfo{person}{V~K Balakrishnan}.}
  \bibinfo{year}{1997}\natexlab{}.
\newblock \bibinfo{booktitle}{\emph{Graph Theory}}.
\newblock \bibinfo{publisher}{McGraw-Hill Professional}, \bibinfo{address}{New
  York, NY}.
\newblock


\bibitem[\protect\citeauthoryear{Barabási and Pósfai}{Barabási and
  Pósfai}{2016}]%
        {barabasi_network_2016}
\bibfield{author}{\bibinfo{person}{Albert-László Barabási} {and}
  \bibinfo{person}{Márton Pósfai}.} \bibinfo{year}{2016}\natexlab{}.
\newblock \bibinfo{booktitle}{\emph{Network science}}.
\newblock \bibinfo{publisher}{Cambridge University Press},
  \bibinfo{address}{Cambridge, United Kingdom}, Chapter The scale-free
  property, \bibinfo{pages}{1--57}.
\newblock
\showISBNx{978-1-107-07626-6}


\bibitem[\protect\citeauthoryear{Birmelé, Ferreira, Grossi, Marino, Pisanti,
  Rizzi, and Sacomoto}{Birmelé et~al\mbox{.}}{2013}]%
        {birmele_optimal_2013}
\bibfield{author}{\bibinfo{person}{Etienne Birmelé}, \bibinfo{person}{Rui
  Ferreira}, \bibinfo{person}{Roberto Grossi}, \bibinfo{person}{Andrea Marino},
  \bibinfo{person}{Nadia Pisanti}, \bibinfo{person}{Romeo Rizzi}, {and}
  \bibinfo{person}{Gustavo Sacomoto}.} \bibinfo{year}{2013}\natexlab{}.
\newblock \showarticletitle{Optimal {Listing} of {Cycles} and st-{Paths} in
  {Undirected} {Graphs}}. In \bibinfo{booktitle}{\emph{Proceedings of the
  {Twenty}-{Fourth} {Annual} {ACM}-{SIAM} {Symposium} on {Discrete}
  {Algorithms}}}. \bibinfo{publisher}{SIAM}, \bibinfo{address}{Philadelphia,
  PA}, \bibinfo{pages}{1884--1896}.
\newblock
\showISBNx{978-1-61197-251-1 978-1-61197-310-5}
\newblock
\shownote{doi: 10.1137/1.9781611973105.134}.


\bibitem[\protect\citeauthoryear{Blanu\v{s}a, Ienne, and Atasu}{Blanu\v{s}a
  et~al\mbox{.}}{2022}]%
        {blanusa_scalable_2022}
\bibfield{author}{\bibinfo{person}{Jovan Blanu\v{s}a}, \bibinfo{person}{Paolo
  Ienne}, {and} \bibinfo{person}{Kubilay Atasu}.}
  \bibinfo{year}{2022}\natexlab{}.
\newblock \showarticletitle{Scalable {Fine}-{Grained} {Parallel} {Cycle}
  {Enumeration} {Algorithms}}. In \bibinfo{booktitle}{\emph{{ACM} {Symposium}
  on {Parallelism} in {Algorithms} and {Architectures} (SPAA)}}.
  \bibinfo{publisher}{ACM}, \bibinfo{address}{Philadelphia PA USA},
  \bibinfo{pages}{247--258}.
\newblock
\showISBNx{978-1-4503-9146-7}
\urldef\tempurl%
\url{https://doi.org/10.1145/3490148.3538585}
\showDOI{\tempurl}


\bibitem[\protect\citeauthoryear{Blanu\v{s}a, Stoica, Ienne, and
  Atasu}{Blanu\v{s}a et~al\mbox{.}}{2020}]%
        {blanusa_manycore_2020}
\bibfield{author}{\bibinfo{person}{Jovan Blanu\v{s}a}, \bibinfo{person}{Radu
  Stoica}, \bibinfo{person}{Paolo Ienne}, {and} \bibinfo{person}{Kubilay
  Atasu}.} \bibinfo{year}{2020}\natexlab{}.
\newblock \showarticletitle{Manycore clique enumeration with fast set
  intersections}.
\newblock \bibinfo{journal}{\emph{PVLDB}} \bibinfo{volume}{13},
  \bibinfo{number}{12} (\bibinfo{date}{Aug.} \bibinfo{year}{2020}),
  \bibinfo{pages}{2676--2690}.
\newblock
\showISSN{2150-8097}
\newblock
\shownote{doi: 10.14778/3407790.3407853}.


\bibitem[\protect\citeauthoryear{Blelloch and Maggs}{Blelloch and
  Maggs}{2010}]%
        {par_algos}
\bibfield{author}{\bibinfo{person}{Guy~E. Blelloch} {and}
  \bibinfo{person}{Bruce~M. Maggs}.} \bibinfo{year}{2010}\natexlab{}.
\newblock \bibinfo{booktitle}{\emph{Parallel Algorithms}}.
\newblock \bibinfo{publisher}{CRC Press}, \bibinfo{address}{London, England},
  Chapter~25, \bibinfo{pages}{25.1--25.40}.
\newblock


\bibitem[\protect\citeauthoryear{Blumofe, Joerg, Kuszmaul, Leiserson, Randall,
  and Zhou}{Blumofe et~al\mbox{.}}{1996}]%
        {blumofe_cilk_1996}
\bibfield{author}{\bibinfo{person}{Robert~D. Blumofe},
  \bibinfo{person}{Christopher~F. Joerg}, \bibinfo{person}{Bradley~C.
  Kuszmaul}, \bibinfo{person}{Charles~E. Leiserson}, \bibinfo{person}{Keith~H.
  Randall}, {and} \bibinfo{person}{Yuli Zhou}.}
  \bibinfo{year}{1996}\natexlab{}.
\newblock \showarticletitle{Cilk: {An} {Efficient} {Multithreaded} {Runtime}
  {System}}.
\newblock \bibinfo{journal}{\emph{J. Parallel and Distrib. Comput.}}
  \bibinfo{volume}{37}, \bibinfo{number}{1} (\bibinfo{date}{Aug.}
  \bibinfo{year}{1996}), \bibinfo{pages}{55--69}.
\newblock
\showISSN{07437315}
\newblock
\shownote{doi: 10.1006/jpdc.1996.0107}.


\bibitem[\protect\citeauthoryear{Blumofe and Leiserson}{Blumofe and
  Leiserson}{1999}]%
        {blumofe_scheduling_1999}
\bibfield{author}{\bibinfo{person}{Robert~D. Blumofe} {and}
  \bibinfo{person}{Charles~E. Leiserson}.} \bibinfo{year}{1999}\natexlab{}.
\newblock \showarticletitle{Scheduling multithreaded computations by work
  stealing}.
\newblock \bibinfo{journal}{\emph{J. ACM}} \bibinfo{volume}{46},
  \bibinfo{number}{5} (\bibinfo{date}{Sept.} \bibinfo{year}{1999}),
  \bibinfo{pages}{720--748}.
\newblock
\showISSN{00045411}
\newblock
\shownote{doi: 10.1145/324133.324234}.


\bibitem[\protect\citeauthoryear{Brent}{Brent}{1974}]%
        {brent_parallel_1974}
\bibfield{author}{\bibinfo{person}{Richard~P. Brent}.}
  \bibinfo{year}{1974}\natexlab{}.
\newblock \showarticletitle{The {Parallel} {Evaluation} of {General}
  {Arithmetic} {Expressions}}.
\newblock \bibinfo{journal}{\emph{J. ACM}} \bibinfo{volume}{21},
  \bibinfo{number}{2} (\bibinfo{date}{April} \bibinfo{year}{1974}),
  \bibinfo{pages}{201--206}.
\newblock
\showISSN{0004-5411, 1557-735X}
\newblock
\shownote{doi: 10.1145/321812.321815}.


\bibitem[\protect\citeauthoryear{Broido and Clauset}{Broido and
  Clauset}{2019}]%
        {broido_scale-free_2019}
\bibfield{author}{\bibinfo{person}{Anna~D. Broido} {and} \bibinfo{person}{Aaron
  Clauset}.} \bibinfo{year}{2019}\natexlab{}.
\newblock \showarticletitle{Scale-free networks are rare}.
\newblock \bibinfo{journal}{\emph{Nat Commun}} \bibinfo{volume}{10},
  \bibinfo{number}{1} (\bibinfo{date}{Dec.} \bibinfo{year}{2019}),
  \bibinfo{pages}{1017}.
\newblock
\showISSN{2041-1723}
\newblock
\shownote{doi: 10.1038/s41467-019-08746-5}.


\bibitem[\protect\citeauthoryear{{CORPORATE The MPI Forum}}{{CORPORATE The MPI
  Forum}}{1993}]%
        {mpi_1993}
\bibfield{author}{\bibinfo{person}{{CORPORATE The MPI Forum}}.}
  \bibinfo{year}{1993}\natexlab{}.
\newblock \showarticletitle{{MPI}: a message passing interface}. In
  \bibinfo{booktitle}{\emph{Proceedings of the 1993 {ACM}/{IEEE} conference on
  {Supercomputing} - {Supercomputing} '93}}. \bibinfo{publisher}{ACM Press},
  \bibinfo{address}{Portland, Oregon, United States},
  \bibinfo{pages}{878--883}.
\newblock
\showISBNx{978-0-8186-4340-8}
\newblock
\shownote{doi: 10.1145/169627.169855}.


\bibitem[\protect\citeauthoryear{Cui, Niu, Zhou, and Shu}{Cui
  et~al\mbox{.}}{2017}]%
        {cui_multi-threading_2017}
\bibfield{author}{\bibinfo{person}{Huanqing Cui}, \bibinfo{person}{Jian Niu},
  \bibinfo{person}{Chuanai Zhou}, {and} \bibinfo{person}{Minglei Shu}.}
  \bibinfo{year}{2017}\natexlab{}.
\newblock \showarticletitle{A {Multi}-{Threading} {Algorithm} to {Detect} and
  {Remove} {Cycles} in {Vertex}- and {Arc}-{Weighted} {Digraph}}.
\newblock \bibinfo{journal}{\emph{Algorithms}} \bibinfo{volume}{10},
  \bibinfo{number}{4} (\bibinfo{date}{Oct.} \bibinfo{year}{2017}),
  \bibinfo{pages}{115}.
\newblock
\showISSN{1999-4893}
\newblock
\shownote{doi: 10.3390/a10040115}.


\bibitem[\protect\citeauthoryear{Danielson}{Danielson}{1968}]%
        {danielson_finding_1968}
\bibfield{author}{\bibinfo{person}{G. Danielson}.}
  \bibinfo{year}{1968}\natexlab{}.
\newblock \showarticletitle{On finding the simple paths and circuits in a
  graph}.
\newblock \bibinfo{journal}{\emph{IEEE Trans. Circuit Theory}}
  \bibinfo{volume}{15}, \bibinfo{number}{3} (\bibinfo{date}{Sept.}
  \bibinfo{year}{1968}), \bibinfo{pages}{294--295}.
\newblock
\showISSN{0018-9324}
\newblock
\shownote{doi: 10.1109/TCT.1968.1082837}.


\bibitem[\protect\citeauthoryear{Das, Sanei-Mehri, and Tirthapura}{Das
  et~al\mbox{.}}{2020}]%
        {das_shared-memory_2020}
\bibfield{author}{\bibinfo{person}{Apurba Das}, \bibinfo{person}{Seyed-Vahid
  Sanei-Mehri}, {and} \bibinfo{person}{Srikanta Tirthapura}.}
  \bibinfo{year}{2020}\natexlab{}.
\newblock \showarticletitle{Shared-memory {Parallel} {Maximal} {Clique}
  {Enumeration} from {Static} and {Dynamic} {Graphs}}.
\newblock \bibinfo{journal}{\emph{ACM Trans. Parallel Comput.}}
  \bibinfo{volume}{7}, \bibinfo{number}{1} (\bibinfo{date}{April}
  \bibinfo{year}{2020}), \bibinfo{pages}{1--28}.
\newblock
\showISSN{2329-4949, 2329-4957}
\urldef\tempurl%
\url{https://doi.org/10.1145/3380936}
\showDOI{\tempurl}


\bibitem[\protect\citeauthoryear{Erdős and Gallai}{Erdős and Gallai}{1959}]%
        {erdos_maximal_1959}
\bibfield{author}{\bibinfo{person}{P. Erdős} {and} \bibinfo{person}{T.
  Gallai}.} \bibinfo{year}{1959}\natexlab{}.
\newblock \showarticletitle{On maximal paths and circuits of graphs}.
\newblock \bibinfo{journal}{\emph{Acta Mathematica Academiae Scientiarum
  Hungaricae}} \bibinfo{volume}{10}, \bibinfo{number}{3-4}
  (\bibinfo{date}{Sept.} \bibinfo{year}{1959}), \bibinfo{pages}{337--356}.
\newblock
\showISSN{0001-5954, 1588-2632}
\newblock
\shownote{doi: 10.1007/BF02024498}.


\bibitem[\protect\citeauthoryear{Fleischer, Hendrickson, and Pınar}{Fleischer
  et~al\mbox{.}}{2000}]%
        {fleischer_identifying_2000}
\bibfield{author}{\bibinfo{person}{Lisa~K. Fleischer}, \bibinfo{person}{Bruce
  Hendrickson}, {and} \bibinfo{person}{Ali Pınar}.}
  \bibinfo{year}{2000}\natexlab{}.
\newblock \showarticletitle{On {Identifying} {Strongly} {Connected}
  {Components} in {Parallel}}.
\newblock In \bibinfo{booktitle}{\emph{Parallel and {Distributed}
  {Processing}}}. Vol.~\bibinfo{volume}{1800}. \bibinfo{publisher}{Springer
  Berlin Heidelberg}, \bibinfo{address}{Berlin, Heidelberg},
  \bibinfo{pages}{505--511}.
\newblock
\showISBNx{978-3-540-67442-9 978-3-540-45591-2}
\newblock
\shownote{doi: 10.1007/3-540-45591-4{\_}68}.


\bibitem[\protect\citeauthoryear{Fraigniaud and Olivetti}{Fraigniaud and
  Olivetti}{2019}]%
        {fraigniaud_distributed_2019}
\bibfield{author}{\bibinfo{person}{Pierre Fraigniaud} {and}
  \bibinfo{person}{Dennis Olivetti}.} \bibinfo{year}{2019}\natexlab{}.
\newblock \showarticletitle{Distributed {Detection} of {Cycles}}.
\newblock \bibinfo{journal}{\emph{ACM Trans. Parallel Comput.}}
  \bibinfo{volume}{6}, \bibinfo{number}{3} (\bibinfo{date}{Dec.}
  \bibinfo{year}{2019}), \bibinfo{pages}{1--20}.
\newblock
\showISSN{2329-4949, 2329-4957}
\urldef\tempurl%
\url{https://doi.org/10.1145/3322811}
\showDOI{\tempurl}


\bibitem[\protect\citeauthoryear{Gibbs}{Gibbs}{1969}]%
        {gibbs_cycle_1969}
\bibfield{author}{\bibinfo{person}{Norman~E. Gibbs}.}
  \bibinfo{year}{1969}\natexlab{}.
\newblock \showarticletitle{A {Cycle} {Generation} {Algorithm} for {Finite}
  {Undirected} {Linear} {Graphs}}.
\newblock \bibinfo{journal}{\emph{J. ACM}} \bibinfo{volume}{16},
  \bibinfo{number}{4} (\bibinfo{date}{Oct.} \bibinfo{year}{1969}),
  \bibinfo{pages}{564--568}.
\newblock
\showISSN{0004-5411, 1557-735X}
\newblock
\shownote{doi: 10.1145/321541.321545}.


\bibitem[\protect\citeauthoryear{Giscard, Rochet, and Wilson}{Giscard
  et~al\mbox{.}}{2017}]%
        {giscard_evaluating_2017}
\bibfield{author}{\bibinfo{person}{Pierre-Louis Giscard}, \bibinfo{person}{Paul
  Rochet}, {and} \bibinfo{person}{Richard~C. Wilson}.}
  \bibinfo{year}{2017}\natexlab{}.
\newblock \showarticletitle{Evaluating balance on social networks from their
  simple cycles}.
\newblock \bibinfo{journal}{\emph{Journal of Complex Networks}}
  \bibinfo{volume}{5} (\bibinfo{date}{May} \bibinfo{year}{2017}),
  \bibinfo{pages}{750--775}.
\newblock
\showISSN{2051-1310, 2051-1329}
\newblock
\shownote{doi: 10.1093/comnet/cnx005}.


\bibitem[\protect\citeauthoryear{{Google Cloud}}{{Google Cloud}}{2022}]%
        {gcloud_n1}
\bibfield{author}{\bibinfo{person}{{Google Cloud}}.}
  \bibinfo{year}{2022}\natexlab{}.
\newblock \bibinfo{title}{General-purpose machine family: {N1} machine series}.
\newblock \bibinfo{howpublished}{Available online:
  \url{https://cloud.google.com/compute/docs/general-purpose-machines}}.
\newblock
\newblock
\shownote{Accessed: 2022-11-14}.


\bibitem[\protect\citeauthoryear{Grossi}{Grossi}{2016}]%
        {kao_enumeration_2016}
\bibfield{author}{\bibinfo{person}{Roberto Grossi}.}
  \bibinfo{year}{2016}\natexlab{}.
\newblock \showarticletitle{Enumeration of {Paths}, {Cycles}, and {Spanning}
  {Trees}}.
\newblock In \bibinfo{booktitle}{\emph{Encyclopedia of {Algorithms}}}.
  \bibinfo{publisher}{Springer New York}, \bibinfo{address}{New York, NY},
  \bibinfo{pages}{640--645}.
\newblock
\showISBNx{978-1-4939-2863-7 978-1-4939-2864-4}
\newblock
\shownote{doi: 10.1007/978-1-4939-2864-4{\_}728}.


\bibitem[\protect\citeauthoryear{Gupta and Selvidge}{Gupta and
  Selvidge}{2005}]%
        {gupta_acyclic_2005}
\bibfield{author}{\bibinfo{person}{A. Gupta} {and} \bibinfo{person}{C.
  Selvidge}.} \bibinfo{year}{2005}\natexlab{}.
\newblock \showarticletitle{Acyclic modeling of combinational loops}. In
  \bibinfo{booktitle}{\emph{{ICCAD}-2005. {IEEE}/{ACM} {International}
  {Conference} on {Computer}-{Aided} {Design}, 2005.}}
  \bibinfo{publisher}{IEEE}, \bibinfo{address}{San Jose, CA},
  \bibinfo{pages}{343--348}.
\newblock
\showISBNx{978-0-7803-9254-0}
\newblock
\shownote{doi: 10.1109/ICCAD.2005.1560091}.


\bibitem[\protect\citeauthoryear{Gupta and Suzumura}{Gupta and
  Suzumura}{2021}]%
        {gupta_finding_2021}
\bibfield{author}{\bibinfo{person}{Anshul Gupta} {and}
  \bibinfo{person}{Toyotaro Suzumura}.} \bibinfo{year}{2021}\natexlab{}.
\newblock \bibinfo{title}{Finding {All} {Bounded}-{Length} {Simple} {Cycles} in
  a {Directed} {Graph}}.
\newblock
\newblock
\showeprint[arxiv]{2105.10094}~[cs.DS]


\bibitem[\protect\citeauthoryear{Hajdu and Kr{\'e}sz}{Hajdu and
  Kr{\'e}sz}{2020}]%
        {hajdu_temporal_2020}
\bibfield{author}{\bibinfo{person}{L{\'a}szl{\'o} Hajdu} {and}
  \bibinfo{person}{Mikl{\'o}s Kr{\'e}sz}.} \bibinfo{year}{2020}\natexlab{}.
\newblock \showarticletitle{Temporal {Network} {Analytics} for {Fraud}
  {Detection} in the {Banking} {Sector}}. In \bibinfo{booktitle}{\emph{ADBIS,
  TPDL and EDA 2020 Common Workshops and Doctoral Consortium}}.
  \bibinfo{publisher}{Springer}, \bibinfo{address}{Cham, Switzerland},
  \bibinfo{pages}{145--157}.
\newblock
\showISBNx{978-3-030-55814-7}
\newblock
\shownote{doi: 10.1007/978-3-030-55814-7{\_}12}.


\bibitem[\protect\citeauthoryear{Islam, Haque, Alam, and Tarikuzzaman}{Islam
  et~al\mbox{.}}{2009}]%
        {islam_approach_2009}
\bibfield{author}{\bibinfo{person}{Md.~Nazrul Islam},
  \bibinfo{person}{S.~M.~Rafizul Haque}, \bibinfo{person}{Kaji~Masudul Alam},
  {and} \bibinfo{person}{Md. Tarikuzzaman}.} \bibinfo{year}{2009}\natexlab{}.
\newblock \showarticletitle{An approach to improve collusion set detection
  using {MCL} algorithm}. In \bibinfo{booktitle}{\emph{2009 12th
  {International} {Conference} on {Computers} and {Information} {Technology}}}.
  \bibinfo{publisher}{IEEE}, \bibinfo{address}{Dhaka, Bangladesh},
  \bibinfo{pages}{237--242}.
\newblock
\showISBNx{978-1-4244-6281-0}
\newblock
\shownote{doi: 10.1109/ICCIT.2009.5407133}.


\bibitem[\protect\citeauthoryear{JaJa}{JaJa}{1992}]%
        {JaJa1992-va}
\bibfield{author}{\bibinfo{person}{Joseph JaJa}.}
  \bibinfo{year}{1992}\natexlab{}.
\newblock \bibinfo{booktitle}{\emph{Introduction to parallel algorithms}}.
\newblock \bibinfo{publisher}{Addison Wesley}, \bibinfo{address}{Boston, MA}.
\newblock


\bibitem[\protect\citeauthoryear{Jankowski, Michalski, and Bródka}{Jankowski
  et~al\mbox{.}}{2017}]%
        {jankowski_spreading_2017}
\bibfield{author}{\bibinfo{person}{Jaroslaw Jankowski},
  \bibinfo{person}{Radosław Michalski}, {and} \bibinfo{person}{Piotr
  Bródka}.} \bibinfo{year}{2017}\natexlab{}.
\newblock \bibinfo{title}{Spreading processes in multilayer complex network
  within virtual world}.
\newblock
\newblock
\newblock
\shownote{doi: 10.7910/DVN/V6AJRV}.


\bibitem[\protect\citeauthoryear{Jiang, Xie, Xiong, Zhang, Zhang, and
  Zhou}{Jiang et~al\mbox{.}}{2013}]%
        {jiang_trading_2013}
\bibfield{author}{\bibinfo{person}{Zhi-Qiang Jiang}, \bibinfo{person}{Wen-Jie
  Xie}, \bibinfo{person}{Xiong Xiong}, \bibinfo{person}{Wei Zhang},
  \bibinfo{person}{Yong-Jie Zhang}, {and} \bibinfo{person}{Wei-Xing Zhou}.}
  \bibinfo{year}{2013}\natexlab{}.
\newblock \showarticletitle{Trading networks, abnormal motifs and stock
  manipulation}.
\newblock \bibinfo{journal}{\emph{Quantitative Finance Letters}}
  \bibinfo{volume}{1}, \bibinfo{number}{1} (\bibinfo{date}{Dec.}
  \bibinfo{year}{2013}), \bibinfo{pages}{1--8}.
\newblock
\showISSN{2164-9502, 2164-9510}
\newblock
\shownote{doi: 10.1080/21649502.2013.802877}.


\bibitem[\protect\citeauthoryear{Johnson}{Johnson}{1975}]%
        {johnson_finding_1975}
\bibfield{author}{\bibinfo{person}{Donald~B. Johnson}.}
  \bibinfo{year}{1975}\natexlab{}.
\newblock \showarticletitle{Finding {All} the {Elementary} {Circuits} of a
  {Directed} {Graph}}.
\newblock \bibinfo{journal}{\emph{SIAM J. Comput.}} \bibinfo{volume}{4},
  \bibinfo{number}{1} (\bibinfo{date}{March} \bibinfo{year}{1975}),
  \bibinfo{pages}{77--84}.
\newblock
\showISSN{0097-5397, 1095-7111}
\newblock
\shownote{doi: 10.1137/0204007}.


\bibitem[\protect\citeauthoryear{Kamae}{Kamae}{1967}]%
        {kamae_systematic_1967}
\bibfield{author}{\bibinfo{person}{T. Kamae}.} \bibinfo{year}{1967}\natexlab{}.
\newblock \showarticletitle{A {Systematic} {Method} of {Finding} {All}
  {Directed} {Circuits} and {Enumerating} {All} {Directed} {Paths}}.
\newblock \bibinfo{journal}{\emph{IEEE Trans. Circuit Theory}}
  \bibinfo{volume}{14}, \bibinfo{number}{2} (\bibinfo{date}{June}
  \bibinfo{year}{1967}), \bibinfo{pages}{166--171}.
\newblock
\showISSN{0018-9324}
\newblock
\shownote{doi: 10.1109/TCT.1967.1082699}.


\bibitem[\protect\citeauthoryear{Klamt and von Kamp}{Klamt and von
  Kamp}{2009}]%
        {klamt_computing_2009}
\bibfield{author}{\bibinfo{person}{Steffen Klamt} {and} \bibinfo{person}{Axel
  von Kamp}.} \bibinfo{year}{2009}\natexlab{}.
\newblock \showarticletitle{Computing paths and cycles in biological
  interaction graphs}.
\newblock \bibinfo{journal}{\emph{BMC Bioinformatics}} \bibinfo{volume}{10},
  \bibinfo{number}{1} (\bibinfo{date}{Dec.} \bibinfo{year}{2009}),
  \bibinfo{pages}{181}.
\newblock
\showISSN{1471-2105}
\newblock
\shownote{doi: 10.1186/1471-2105-10-181}.


\bibitem[\protect\citeauthoryear{Kukanov}{Kukanov}{2007}]%
        {kukanov_foundations_2007}
\bibfield{author}{\bibinfo{person}{Alexey Kukanov}.}
  \bibinfo{year}{2007}\natexlab{}.
\newblock \showarticletitle{The {Foundations} for {Scalable} {Multicore}
  {Software} in {Intel} {Threading} {Building} {Blocks}}.
\newblock \bibinfo{journal}{\emph{ITJ}} \bibinfo{volume}{11},
  \bibinfo{number}{04} (\bibinfo{date}{Nov.} \bibinfo{year}{2007}),
  \bibinfo{pages}{309--322}.
\newblock
\showISSN{1535864X, 1535864X}
\newblock
\shownote{doi: 10.1535/itj.1104.05}.


\bibitem[\protect\citeauthoryear{Kumar and Calders}{Kumar and Calders}{2018}]%
        {kumar_2scent_2018}
\bibfield{author}{\bibinfo{person}{Rohit Kumar} {and} \bibinfo{person}{Toon
  Calders}.} \bibinfo{year}{2018}\natexlab{}.
\newblock \showarticletitle{{2SCENT}: an efficient algorithm for enumerating
  all simple temporal cycles}.
\newblock \bibinfo{journal}{\emph{PVLDB}} \bibinfo{volume}{11},
  \bibinfo{number}{11} (\bibinfo{date}{July} \bibinfo{year}{2018}),
  \bibinfo{pages}{1441--1453}.
\newblock
\showISSN{2150-8097}
\newblock
\shownote{doi: 10.14778/3236187.3236197}.


\bibitem[\protect\citeauthoryear{Kunegis}{Kunegis}{2013}]%
        {kunegis_konect_2013}
\bibfield{author}{\bibinfo{person}{Jérôme Kunegis}.}
  \bibinfo{year}{2013}\natexlab{}.
\newblock \showarticletitle{{KONECT}: the {Koblenz} network collection}. In
  \bibinfo{booktitle}{\emph{Proceedings of the 22nd {International}
  {Conference} on {World} {Wide} {Web} - {WWW} '13 {Companion}}}.
  \bibinfo{publisher}{ACM Press}, \bibinfo{address}{Rio de Janeiro, Brazil},
  \bibinfo{pages}{1343--1350}.
\newblock
\showISBNx{978-1-4503-2038-2}
\newblock
\shownote{doi: 10.1145/2487788.2488173}.


\bibitem[\protect\citeauthoryear{Kwon and Cho}{Kwon and Cho}{2007}]%
        {kwon_analysis_2007}
\bibfield{author}{\bibinfo{person}{Yung-Keun Kwon} {and}
  \bibinfo{person}{Kwang-Hyun Cho}.} \bibinfo{year}{2007}\natexlab{}.
\newblock \showarticletitle{Analysis of feedback loops and robustness in
  network evolution based on {Boolean} models}.
\newblock \bibinfo{journal}{\emph{BMC Bioinformatics}}  \bibinfo{volume}{8}
  (\bibinfo{year}{2007}), \bibinfo{pages}{430}.
\newblock
\showISSN{1471-2105}
\newblock
\shownote{doi: 10.1186/1471-2105-8-430}.


\bibitem[\protect\citeauthoryear{Leskovec and Krevl}{Leskovec and
  Krevl}{2014}]%
        {snapnets}
\bibfield{author}{\bibinfo{person}{Jure Leskovec} {and} \bibinfo{person}{Andrej
  Krevl}.} \bibinfo{year}{2014}\natexlab{}.
\newblock \bibinfo{title}{{SNAP Datasets}: {Stanford} Large Network Dataset
  Collection}.
\newblock \bibinfo{howpublished}{Available online:
  \url{https://snap.stanford.edu/data}}.
\newblock
\newblock
\shownote{Accessed: 2022-05-30}.


\bibitem[\protect\citeauthoryear{Li, Liu, Li, Han, Shi, Hooi, Huang, and
  Cheng}{Li et~al\mbox{.}}{2020}]%
        {li_flowscope_2020}
\bibfield{author}{\bibinfo{person}{Xiangfeng Li}, \bibinfo{person}{Shenghua
  Liu}, \bibinfo{person}{Zifeng Li}, \bibinfo{person}{Xiaotian Han},
  \bibinfo{person}{Chuan Shi}, \bibinfo{person}{Bryan Hooi},
  \bibinfo{person}{He Huang}, {and} \bibinfo{person}{Xueqi Cheng}.}
  \bibinfo{year}{2020}\natexlab{}.
\newblock \showarticletitle{{FlowScope}: {Spotting} {Money} {Laundering}
  {Based} on {Graphs}}.
\newblock \bibinfo{journal}{\emph{Proceedings of the AAAI Conference on
  Artificial Intelligence}}  \bibinfo{volume}{34} (\bibinfo{date}{April}
  \bibinfo{year}{2020}), \bibinfo{pages}{4731--4738}.
\newblock
\showISSN{2374-3468, 2159-5399}
\urldef\tempurl%
\url{https://doi.org/10.1609/aaai.v34i04.5906}
\showDOI{\tempurl}


\bibitem[\protect\citeauthoryear{Loizou and Thanisch}{Loizou and
  Thanisch}{1982}]%
        {loizou_enumerating_1982}
\bibfield{author}{\bibinfo{person}{G. Loizou} {and} \bibinfo{person}{P.
  Thanisch}.} \bibinfo{year}{1982}\natexlab{}.
\newblock \showarticletitle{Enumerating the cycles of a digraph: {A} new
  preprocessing strategy}.
\newblock \bibinfo{journal}{\emph{Information Sciences}} \bibinfo{volume}{27},
  \bibinfo{number}{3} (\bibinfo{date}{Aug.} \bibinfo{year}{1982}),
  \bibinfo{pages}{163--182}.
\newblock
\showISSN{00200255}
\newblock
\shownote{doi: 10.1016/0020-0255(82)90023-8}.


\bibitem[\protect\citeauthoryear{Malewicz, Austern, Bik, Dehnert, Horn, Leiser,
  and Czajkowski}{Malewicz et~al\mbox{.}}{2010}]%
        {malewicz_pregel_2010}
\bibfield{author}{\bibinfo{person}{Grzegorz Malewicz},
  \bibinfo{person}{Matthew~H. Austern}, \bibinfo{person}{Aart~J.C Bik},
  \bibinfo{person}{James~C. Dehnert}, \bibinfo{person}{Ilan Horn},
  \bibinfo{person}{Naty Leiser}, {and} \bibinfo{person}{Grzegorz Czajkowski}.}
  \bibinfo{year}{2010}\natexlab{}.
\newblock \showarticletitle{Pregel: a system for large-scale graph processing}.
  In \bibinfo{booktitle}{\emph{Proceedings of the 2010 {ACM} {SIGMOD}
  {International} {Conference} on {Management} of data}}.
  \bibinfo{publisher}{ACM}, \bibinfo{address}{Indianapolis Indiana USA},
  \bibinfo{pages}{135--146}.
\newblock
\showISBNx{978-1-4503-0032-2}
\newblock
\shownote{doi: 10.1145/1807167.1807184}.


\bibitem[\protect\citeauthoryear{Mateti and Deo}{Mateti and Deo}{1976}]%
        {mateti_algorithms_1976}
\bibfield{author}{\bibinfo{person}{Prabhaker Mateti} {and}
  \bibinfo{person}{Narsingh Deo}.} \bibinfo{year}{1976}\natexlab{}.
\newblock \showarticletitle{On {Algorithms} for {Enumerating} {All} {Circuits}
  of a {Graph}}.
\newblock \bibinfo{journal}{\emph{SIAM J. Comput.}} \bibinfo{volume}{5},
  \bibinfo{number}{1} (\bibinfo{date}{March} \bibinfo{year}{1976}),
  \bibinfo{pages}{90--99}.
\newblock
\showISSN{0097-5397, 1095-7111}
\newblock
\shownote{doi: 10.1137/0205007}.


\bibitem[\protect\citeauthoryear{Mathur}{Mathur}{2017}]%
        {neo4j_whitepaper_finance}
\bibfield{author}{\bibinfo{person}{Nav Mathur}.}
  \bibinfo{year}{2017}\natexlab{}.
\newblock \bibinfo{booktitle}{\emph{Graph {Technology} for {Financial}
  {Services}}}.
\newblock \bibinfo{type}{{T}echnical {R}eport}. \bibinfo{institution}{Neo4J}.
  \bibinfo{pages}{1--14} pages.
\newblock
\urldef\tempurl%
\url{https://neo4j.com/use-cases/financial-services}
\showURL{%
\tempurl}
\newblock
\shownote{Accessed: 2022-05-30}.


\bibitem[\protect\citeauthoryear{McCune, Weninger, and Madey}{McCune
  et~al\mbox{.}}{2015}]%
        {mccune_thinking_2015}
\bibfield{author}{\bibinfo{person}{Robert~Ryan McCune}, \bibinfo{person}{Tim
  Weninger}, {and} \bibinfo{person}{Greg Madey}.}
  \bibinfo{year}{2015}\natexlab{}.
\newblock \showarticletitle{Thinking {Like} a {Vertex}: {A} {Survey} of
  {Vertex}-{Centric} {Frameworks} for {Large}-{Scale} {Distributed} {Graph}
  {Processing}}.
\newblock \bibinfo{journal}{\emph{ACM Comput. Surv.}} \bibinfo{volume}{48},
  \bibinfo{number}{2} (\bibinfo{date}{Nov.} \bibinfo{year}{2015}),
  \bibinfo{pages}{1--39}.
\newblock
\showISSN{0360-0300, 1557-7341}
\newblock
\shownote{doi: 10.1145/2818185}.


\bibitem[\protect\citeauthoryear{Meusel, Vigna, Lehmberg, and Bizer}{Meusel
  et~al\mbox{.}}{2014}]%
        {meusel_graph_2014}
\bibfield{author}{\bibinfo{person}{Robert Meusel}, \bibinfo{person}{Sebastiano
  Vigna}, \bibinfo{person}{Oliver Lehmberg}, {and} \bibinfo{person}{Christian
  Bizer}.} \bibinfo{year}{2014}\natexlab{}.
\newblock \showarticletitle{Graph structure in the web --- revisited: a trick
  of the heavy tail}. In \bibinfo{booktitle}{\emph{Proceedings of the 23rd
  {International} {Conference} on {World} {Wide} {Web} - {WWW} '14
  {Companion}}}. \bibinfo{publisher}{ACM Press}, \bibinfo{address}{Seoul,
  Korea}, \bibinfo{pages}{427--432}.
\newblock
\showISBNx{978-1-4503-2745-9}
\newblock
\shownote{doi: 10.1145/2567948.2576928}.


\bibitem[\protect\citeauthoryear{Noel, Harley, Tam, Limiero, and Share}{Noel
  et~al\mbox{.}}{2016}]%
        {noel_cygraph_2016}
\bibfield{author}{\bibinfo{person}{S. Noel}, \bibinfo{person}{E. Harley},
  \bibinfo{person}{K.H. Tam}, \bibinfo{person}{M. Limiero}, {and}
  \bibinfo{person}{M. Share}.} \bibinfo{year}{2016}\natexlab{}.
\newblock \showarticletitle{{CyGraph}: {Graph}-{Based} {Analytics} and
  {Visualization} for {Cybersecurity}}.
\newblock In \bibinfo{booktitle}{\emph{Handbook of {Statistics}}}.
  Vol.~\bibinfo{volume}{35}. \bibinfo{publisher}{Elsevier},
  \bibinfo{address}{Oxford, England}, \bibinfo{pages}{117--167}.
\newblock
\showISBNx{978-0-444-63744-4}
\newblock
\shownote{doi: 10.1016/bs.host.2016.07.001}.


\bibitem[\protect\citeauthoryear{Oliva, Setola, Glielmo, and Hadjicostis}{Oliva
  et~al\mbox{.}}{2018}]%
        {oliva_distributed_2018}
\bibfield{author}{\bibinfo{person}{Gabriele Oliva}, \bibinfo{person}{Roberto
  Setola}, \bibinfo{person}{Luigi Glielmo}, {and}
  \bibinfo{person}{Christoforos~N. Hadjicostis}.}
  \bibinfo{year}{2018}\natexlab{}.
\newblock \showarticletitle{Distributed {Cycle} {Detection} and {Removal}}.
\newblock \bibinfo{journal}{\emph{IEEE Trans. Control Netw. Syst.}}
  \bibinfo{volume}{5}, \bibinfo{number}{1} (\bibinfo{date}{March}
  \bibinfo{year}{2018}), \bibinfo{pages}{194--204}.
\newblock
\showISSN{2325-5870}
\urldef\tempurl%
\url{https://doi.org/10.1109/TCNS.2016.2593264}
\showDOI{\tempurl}


\bibitem[\protect\citeauthoryear{Palshikar and Apte}{Palshikar and
  Apte}{2008}]%
        {palshikar_collusion_2008}
\bibfield{author}{\bibinfo{person}{Girish~Keshav Palshikar} {and}
  \bibinfo{person}{Manoj~M. Apte}.} \bibinfo{year}{2008}\natexlab{}.
\newblock \showarticletitle{Collusion set detection using graph clustering}.
\newblock \bibinfo{journal}{\emph{Data Min Knowl Disc}} \bibinfo{volume}{16},
  \bibinfo{number}{2} (\bibinfo{date}{April} \bibinfo{year}{2008}),
  \bibinfo{pages}{135--164}.
\newblock
\showISSN{1384-5810, 1573-756X}
\newblock
\shownote{doi: 10.1007/s10618-007-0076-8}.


\bibitem[\protect\citeauthoryear{Paranjape, Benson, and Leskovec}{Paranjape
  et~al\mbox{.}}{2017}]%
        {paranjape_motifs_2017}
\bibfield{author}{\bibinfo{person}{Ashwin Paranjape},
  \bibinfo{person}{Austin~R. Benson}, {and} \bibinfo{person}{Jure Leskovec}.}
  \bibinfo{year}{2017}\natexlab{}.
\newblock \showarticletitle{Motifs in {Temporal} {Networks}}. In
  \bibinfo{booktitle}{\emph{Proceedings of the {Tenth} {ACM} {International}
  {Conference} on {Web} {Search} and {Data} {Mining}}}.
  \bibinfo{publisher}{ACM}, \bibinfo{address}{Cambridge, United Kingdom},
  \bibinfo{pages}{601--610}.
\newblock
\showISBNx{978-1-4503-4675-7}
\newblock
\shownote{doi: 10.1145/3018661.3018731}.


\bibitem[\protect\citeauthoryear{Peng, Zhang, Lin, Zhang, Qin, and Zhou}{Peng
  et~al\mbox{.}}{2019}]%
        {peng_towards_2019}
\bibfield{author}{\bibinfo{person}{You Peng}, \bibinfo{person}{Ying Zhang},
  \bibinfo{person}{Xuemin Lin}, \bibinfo{person}{Wenjie Zhang},
  \bibinfo{person}{Lu Qin}, {and} \bibinfo{person}{Jingren Zhou}.}
  \bibinfo{year}{2019}\natexlab{}.
\newblock \showarticletitle{Towards bridging theory and practice:
  hop-constrained s-t simple path enumeration}.
\newblock \bibinfo{journal}{\emph{PVLDB}} \bibinfo{volume}{13},
  \bibinfo{number}{4} (\bibinfo{date}{Dec.} \bibinfo{year}{2019}),
  \bibinfo{pages}{463--476}.
\newblock
\showISSN{2150-8097}
\newblock
\shownote{doi: 10.14778/3372716.3372720}.


\bibitem[\protect\citeauthoryear{Ponstein}{Ponstein}{1966}]%
        {ponstein_self-avoiding_1966}
\bibfield{author}{\bibinfo{person}{J. Ponstein}.}
  \bibinfo{year}{1966}\natexlab{}.
\newblock \showarticletitle{Self-{Avoiding} {Paths} and the {Adjacency}
  {Matrix} of a {Graph}}.
\newblock \bibinfo{journal}{\emph{SIAM J. Appl. Math.}} \bibinfo{volume}{14},
  \bibinfo{number}{3} (\bibinfo{date}{March} \bibinfo{year}{1966}),
  \bibinfo{pages}{600--609}.
\newblock
\showISSN{0036-1399, 1095-712X}
\newblock
\shownote{doi: 10.1137/0114051}.


\bibitem[\protect\citeauthoryear{Pothukuchi and Dhuria}{Pothukuchi and
  Dhuria}{2021}]%
        {pothukuchi_dhuria_2021}
\bibfield{author}{\bibinfo{person}{Sri~Harsha Pothukuchi} {and}
  \bibinfo{person}{Amit Dhuria}.} \bibinfo{year}{2021}\natexlab{}.
\newblock \bibinfo{title}{Deterministic loop breaking in multi-mode
  multi-corner static timing analysis of integrated circuits}.
\newblock
\newblock
\newblock
\shownote{Patent No. 11003821}.


\bibitem[\protect\citeauthoryear{Qing, Yuan, Chen, Lin, and Ma}{Qing
  et~al\mbox{.}}{2020}]%
        {nah_efficient_2020}
\bibfield{author}{\bibinfo{person}{Zhu Qing}, \bibinfo{person}{Long Yuan},
  \bibinfo{person}{Zi Chen}, \bibinfo{person}{Jingjing Lin}, {and}
  \bibinfo{person}{Guojie Ma}.} \bibinfo{year}{2020}\natexlab{}.
\newblock \showarticletitle{Efficient {Parallel} {Cycle} {Search} in {Large}
  {Graphs}}.
\newblock In \bibinfo{booktitle}{\emph{Database {Systems} for {Advanced}
  {Applications}}}. Vol.~\bibinfo{volume}{12113}. \bibinfo{publisher}{Springer
  International}, \bibinfo{address}{Cham, Switzerland},
  \bibinfo{pages}{349--367}.
\newblock
\showISBNx{978-3-030-59415-2 978-3-030-59416-9}
\newblock
\shownote{doi: 10.1007/978-3-030-59416-9{\_}21}.


\bibitem[\protect\citeauthoryear{Qiu, Cen, Qian, Peng, Zhang, Lin, and
  Zhou}{Qiu et~al\mbox{.}}{2018}]%
        {qiu_real-time_2018}
\bibfield{author}{\bibinfo{person}{Xiafei Qiu}, \bibinfo{person}{Wubin Cen},
  \bibinfo{person}{Zhengping Qian}, \bibinfo{person}{You Peng},
  \bibinfo{person}{Ying Zhang}, \bibinfo{person}{Xuemin Lin}, {and}
  \bibinfo{person}{Jingren Zhou}.} \bibinfo{year}{2018}\natexlab{}.
\newblock \showarticletitle{Real-time constrained cycle detection in large
  dynamic graphs}.
\newblock \bibinfo{journal}{\emph{PVLDB}} \bibinfo{volume}{11},
  \bibinfo{number}{12} (\bibinfo{date}{Aug.} \bibinfo{year}{2018}),
  \bibinfo{pages}{1876--1888}.
\newblock
\showISSN{2150-8097}
\newblock
\shownote{doi: 10.14778/3229863.3229874}.


\bibitem[\protect\citeauthoryear{Quinn}{Quinn}{2004}]%
        {quinn_parallel_2004}
\bibfield{author}{\bibinfo{person}{Michael~J. Quinn}.}
  \bibinfo{year}{2004}\natexlab{}.
\newblock \bibinfo{booktitle}{\emph{Parallel programming in {C} with {MPI} and
  {openMP}}}.
\newblock \bibinfo{publisher}{McGraw-Hill}, \bibinfo{address}{Dubuque, Iowa}.
\newblock
\showISBNx{978-0-07-282256-4}


\bibitem[\protect\citeauthoryear{Read and Tarjan}{Read and Tarjan}{1975}]%
        {read_bounds_1975}
\bibfield{author}{\bibinfo{person}{R.~C. Read} {and} \bibinfo{person}{R.~E.
  Tarjan}.} \bibinfo{year}{1975}\natexlab{}.
\newblock \showarticletitle{Bounds on {Backtrack} {Algorithms} for {Listing}
  {Cycles}, {Paths}, and {Spanning} {Trees}}.
\newblock \bibinfo{journal}{\emph{Networks}} \bibinfo{volume}{5},
  \bibinfo{number}{3} (\bibinfo{date}{July} \bibinfo{year}{1975}),
  \bibinfo{pages}{237--252}.
\newblock
\showISSN{00283045}
\newblock
\shownote{doi: 10.1002/net.1975.5.3.237}.


\bibitem[\protect\citeauthoryear{Rocha and Thatte}{Rocha and Thatte}{2015}]%
        {rocha_distributed_2015}
\bibfield{author}{\bibinfo{person}{Rodrigo~Caetano Rocha} {and}
  \bibinfo{person}{Bhalchandra~D. Thatte}.} \bibinfo{year}{2015}\natexlab{}.
\newblock \showarticletitle{Distributed cycle detection in large-scale sparse
  graphs}. In \bibinfo{booktitle}{\emph{2015 Simpósio Brasileiro de Pesquisa
  Operacional (SBPO)}}. \bibinfo{publisher}{SOBRAPO}, \bibinfo{address}{Porto
  de Galinhas, Pernambuco, Brasil}, \bibinfo{pages}{1--12}.
\newblock
\newblock
\shownote{doi: 10.13140/RG.2.1.1233.8640}.


\bibitem[\protect\citeauthoryear{Roy, Mihailovic, and Zwaenepoel}{Roy
  et~al\mbox{.}}{2013}]%
        {roy_x-stream_2013}
\bibfield{author}{\bibinfo{person}{Amitabha Roy}, \bibinfo{person}{Ivo
  Mihailovic}, {and} \bibinfo{person}{Willy Zwaenepoel}.}
  \bibinfo{year}{2013}\natexlab{}.
\newblock \showarticletitle{X-{Stream}: edge-centric graph processing using
  streaming partitions}. In \bibinfo{booktitle}{\emph{Proceedings of the
  {Twenty}-{Fourth} {ACM} {Symposium} on {Operating} {Systems} {Principles}}}.
  \bibinfo{publisher}{ACM}, \bibinfo{address}{Farminton Pennsylvania},
  \bibinfo{pages}{472--488}.
\newblock
\showISBNx{978-1-4503-2388-8}
\newblock
\shownote{doi: 10.1145/2517349.2522740}.


\bibitem[\protect\citeauthoryear{{SAS}}{{SAS}}{2021}]%
        {sas_example}
\bibfield{author}{\bibinfo{person}{{SAS}}.} \bibinfo{year}{2021}\natexlab{}.
\newblock \bibinfo{title}{{SAS} {OPTGRAPH} {Procedure}: {Graph} {Algorithms}
  and {Network} {Analysis}}.
\newblock \bibinfo{howpublished}{Available online:
  \url{https://documentation.sas.com/doc/en/pgmsascdc/9.4_3.5/procgralg/procgralg_optgraph_examples.htm}}.
\newblock
\newblock
\shownote{Accessed: 2022-05-30}.


\bibitem[\protect\citeauthoryear{Sodani}{Sodani}{2015}]%
        {sodani_knights_2015}
\bibfield{author}{\bibinfo{person}{Avinash Sodani}.}
  \bibinfo{year}{2015}\natexlab{}.
\newblock \showarticletitle{Knights landing ({KNL}): 2nd {Generation} {Intel}®
  {Xeon} {Phi} processor}. In \bibinfo{booktitle}{\emph{2015 {IEEE} {Hot}
  {Chips} 27 {Symposium} ({HCS})}}. \bibinfo{publisher}{IEEE},
  \bibinfo{address}{Cupertino, CA, USA}, \bibinfo{pages}{1--24}.
\newblock
\showISBNx{978-1-4673-8885-6}
\newblock
\shownote{doi: 10.1109/HOTCHIPS.2015.7477467}.


\bibitem[\protect\citeauthoryear{Suzumura and Kanezashi}{Suzumura and
  Kanezashi}{2021}]%
        {AMLSim}
\bibfield{author}{\bibinfo{person}{Toyotaro Suzumura} {and}
  \bibinfo{person}{Hiroki Kanezashi}.} \bibinfo{year}{2021}\natexlab{}.
\newblock \bibinfo{title}{{Anti-Money Laundering Datasets}: {InPlusLab}
  Anti-Money Laundering Datasets}.
\newblock \bibinfo{howpublished}{Available online:
  \url{https://github.com/IBM/AMLSim}}.
\newblock
\newblock
\shownote{Accessed: 2022-05-30}.


\bibitem[\protect\citeauthoryear{Szwarcfiter and Lauer}{Szwarcfiter and
  Lauer}{1976}]%
        {Szwarcfiter1976ASS}
\bibfield{author}{\bibinfo{person}{J. Szwarcfiter} {and} \bibinfo{person}{P.
  Lauer}.} \bibinfo{year}{1976}\natexlab{}.
\newblock \showarticletitle{A search strategy for the elementary cycles of a
  directed graph}.
\newblock \bibinfo{journal}{\emph{BIT Numerical Mathematics}}
  \bibinfo{volume}{16} (\bibinfo{year}{1976}), \bibinfo{pages}{192--204}.
\newblock


\bibitem[\protect\citeauthoryear{Tarjan}{Tarjan}{1973}]%
        {tarjan_enumeration_1973}
\bibfield{author}{\bibinfo{person}{Robert Tarjan}.}
  \bibinfo{year}{1973}\natexlab{}.
\newblock \showarticletitle{Enumeration of the {Elementary} {Circuits} of a
  {Directed} {Graph}}.
\newblock \bibinfo{journal}{\emph{SIAM J. Comput.}} \bibinfo{volume}{2},
  \bibinfo{number}{3} (\bibinfo{date}{Sept.} \bibinfo{year}{1973}),
  \bibinfo{pages}{211--216}.
\newblock
\showISSN{0097-5397, 1095-7111}
\newblock
\shownote{doi: 10.1137/0202017}.


\bibitem[\protect\citeauthoryear{Tiernan}{Tiernan}{1970}]%
        {tiernan_efficient_1970}
\bibfield{author}{\bibinfo{person}{James~C. Tiernan}.}
  \bibinfo{year}{1970}\natexlab{}.
\newblock \showarticletitle{An efficient search algorithm to find the
  elementary circuits of a graph}.
\newblock \bibinfo{journal}{\emph{Commun. ACM}} \bibinfo{volume}{13},
  \bibinfo{number}{12} (\bibinfo{date}{Dec.} \bibinfo{year}{1970}),
  \bibinfo{pages}{722--726}.
\newblock
\showISSN{0001-0782, 1557-7317}
\newblock
\shownote{doi: 10.1145/362814.362819}.


\bibitem[\protect\citeauthoryear{Valiant}{Valiant}{1990}]%
        {valiant_bridging_1990}
\bibfield{author}{\bibinfo{person}{Leslie~G. Valiant}.}
  \bibinfo{year}{1990}\natexlab{}.
\newblock \showarticletitle{A bridging model for parallel computation}.
\newblock \bibinfo{journal}{\emph{Commun. ACM}} \bibinfo{volume}{33},
  \bibinfo{number}{8} (\bibinfo{date}{Aug.} \bibinfo{year}{1990}),
  \bibinfo{pages}{103--111}.
\newblock
\showISSN{00010782}
\newblock
\shownote{doi: 10.1145/79173.79181}.


\bibitem[\protect\citeauthoryear{Wang, Cui, Pei, Song, and Zang}{Wang
  et~al\mbox{.}}{2020}]%
        {wang_recent_2020}
\bibfield{author}{\bibinfo{person}{Fei Wang}, \bibinfo{person}{Peng Cui},
  \bibinfo{person}{Jian Pei}, \bibinfo{person}{Yangqiu Song}, {and}
  \bibinfo{person}{Chengxi Zang}.} \bibinfo{year}{2020}\natexlab{}.
\newblock \showarticletitle{Recent {Advances} on {Graph} {Analytics} and {Its}
  {Applications} in {Healthcare}}. In \bibinfo{booktitle}{\emph{Proceedings of
  the 26th {ACM} {SIGKDD} {International} {Conference} on {Knowledge}
  {Discovery} \& {Data} {Mining}}}. \bibinfo{publisher}{ACM},
  \bibinfo{address}{Virtual Event CA USA}, \bibinfo{pages}{3545--3546}.
\newblock
\showISBNx{978-1-4503-7998-4}
\newblock
\shownote{doi: 10.1145/3394486.3406469}.


\bibitem[\protect\citeauthoryear{Webber}{Webber}{2021}]%
        {neo4j_whitepaper_reccomendation}
\bibfield{author}{\bibinfo{person}{Jim Webber}.}
  \bibinfo{year}{2021}\natexlab{}.
\newblock \bibinfo{booktitle}{\emph{Powering {Real-Time} {Recommendations} with
  {Graph} {Database} {Technology}}}.
\newblock \bibinfo{type}{{T}echnical {R}eport}. \bibinfo{institution}{Neo4J}.
  \bibinfo{pages}{1--7} pages.
\newblock
\urldef\tempurl%
\url{https://neo4j.com/use-cases/real-time-recommendation-engine}
\showURL{%
\tempurl}
\newblock
\shownote{Accessed: 2022-05-30}.


\bibitem[\protect\citeauthoryear{Weinblatt}{Weinblatt}{1972}]%
        {weinblatt_new_1972}
\bibfield{author}{\bibinfo{person}{Herbert Weinblatt}.}
  \bibinfo{year}{1972}\natexlab{}.
\newblock \showarticletitle{A {New} {Search} {Algorithm} for {Finding} the
  {Simple} {Cycles} of a {Finite} {Directed} {Graph}}.
\newblock \bibinfo{journal}{\emph{J. ACM}} \bibinfo{volume}{19},
  \bibinfo{number}{1} (\bibinfo{date}{Jan.} \bibinfo{year}{1972}),
  \bibinfo{pages}{43--56}.
\newblock
\showISSN{0004-5411, 1557-735X}


\bibitem[\protect\citeauthoryear{Welch}{Welch}{1965}]%
        {welch_numerical_1965}
\bibfield{author}{\bibinfo{person}{J.~T. Welch}.}
  \bibinfo{year}{1965}\natexlab{}.
\newblock \showarticletitle{Numerical applications: {Cycle} algorithms for
  undirected linear graphs and some immediate applications}. In
  \bibinfo{booktitle}{\emph{Proceedings of the 1965 20th National Conference}}.
  \bibinfo{publisher}{ACM Press}, \bibinfo{address}{Cleveland, Ohio, United
  States}, \bibinfo{pages}{296--301}.
\newblock
\newblock
\shownote{doi: 10.1145/800197.806053}.


\bibitem[\protect\citeauthoryear{Zhou, Liang, Zhao, and Zhang}{Zhou
  et~al\mbox{.}}{2018}]%
        {zhou_cycle_2018}
\bibfield{author}{\bibinfo{person}{Xiaoping Zhou}, \bibinfo{person}{Xun Liang},
  \bibinfo{person}{Jichao Zhao}, {and} \bibinfo{person}{Shusen Zhang}.}
  \bibinfo{year}{2018}\natexlab{}.
\newblock \showarticletitle{Cycle {Based} {Network} {Centrality}}.
\newblock \bibinfo{journal}{\emph{Sci Rep}} \bibinfo{volume}{8},
  \bibinfo{number}{1} (\bibinfo{date}{Dec.} \bibinfo{year}{2018}),
  \bibinfo{pages}{11749}.
\newblock
\showISSN{2045-2322}
\newblock
\shownote{doi: 10.1038/s41598-018-30249-4}.


\end{thebibliography}
